\documentclass[12pt,a4paper]{article}
\usepackage[T1]{fontenc}
\usepackage[latin1]{inputenc}
\usepackage{color}
\usepackage[dvipsnames]{xcolor}	
\usepackage{enumitem}
\usepackage{graphicx}
\usepackage{booktabs}
\usepackage{multicol}
\usepackage{multirow}
\usepackage{float}
\usepackage{amsfonts}
\usepackage{amsmath}
\usepackage{amssymb}
\usepackage{theorem}
\usepackage{mathtools}
\usepackage{enumerate}
\usepackage{enumitem}
\usepackage{upgreek}

\theorembodyfont{\slshape}
\newtheorem{theorem}{Theorem}[section]
\newtheorem{lemma}[theorem]{Lemma}
\newtheorem{proposition}[theorem]{Proposition}
\newtheorem{corollary}[theorem]{Corollary}
\newtheorem{definition}[theorem]{Definition}

\theorembodyfont{\upshape\small}

\newtheorem{remark}[theorem]{Remark}

\theorembodyfont{\ttfamily}

\newenvironment{proof}{\list{}{\itemindent-\leftmargin}%
	\item\textbf{Proof: }\small}{\hbox{}\hfill$\blacksquare$\newline\endlist\vspace{-2mm}}
\newcounter{num}
\renewcommand{\thenum}{\Alph{num}}
\newenvironment{parlist}
{%
	\begin{list}{(\thenum)}{%
			\usecounter{num}%
		}%
	}{\end{list}}

\numberwithin{equation}{section}
\newcommand{\ba}{\begin{equation}}
	\newcommand{\ea}{\end{equation}}
\newcommand{\bano}{\begin{equation*}}
	\newcommand{\eano}{\end{equation*}}

\newcommand{\landau}{\mbox{$\scriptstyle{\mathcal{O}}$}}

\newcommand{\bmu}{\mbox{\boldmath $\mu$}}
\newcommand{\bnu}{\mbox{\boldmath $\nu$}}
\newcommand{\D}{\mbox{\boldmath $D$}}
\newcommand{\Hes}{\mbox{\boldmath $H$}}
\newcommand{\X}{\mbox{\boldmath $X$}}
\newcommand{\x}{\mbox{\boldmath $x$}}
\newcommand{\Y}{\mbox{\boldmath $Y$}}

\newcommand{\f}{\mbox{\boldmath $f$}}

\newcommand{\bin}{\textup{Bin}}
\newcommand{\poi}{\textup{Poi}}
\newcommand{\skw}{\textup{Skew}}
\newcommand{\norm}{\mathcal{N}}

\newcommand{\bbn}{\mathbb{N}}
\newcommand{\bbz}{\mathbb{Z}}
\newcommand{\e}{\mathbb{E}}

\newcommand{\Pro}{\mathbb{P}}

\newcommand{\iid}{i.i.d.\,}
\newcommand{\iidno}{i.i.d.}
\newcommand{\ie}{i.\,e., }
\newcommand{\eg}{e.\,g., }

\usepackage{dsfont}

\usepackage{natbib}

\usepackage{hyperref}
\hypersetup{
	colorlinks,
	linkcolor={red!90!black}, 
	citecolor={green!50!black}, 
	urlcolor={black}
}

\begin{document}
	
	
	
	\parindent 0cm
	
	\title{Marginal Analysis of Count Time Series in the Presence of Missing Observations}
	\author{
		Simon Nik\thanks{
			Helmut Schmidt University, Department of Mathematics and Statistics, Hamburg, Germany.}
	}

	\maketitle

\begin{abstract}
	\noindent
	Time series in real-world applications often have missing observations, making typical analytical methods unsuitable. One method for dealing with missing data is the concept of amplitude modulation. While this principle works with any data, here, missing data for unbounded and bounded count time series are investigated, where tailor-made dispersion and skewness statistics are used for model diagnostics. General closed-form asymptotic formulas are derived for such statistics with only weak assumptions on the underlying process. Moreover, closed-form formulas are derived for the popular special cases of Poisson and binomial autoregressive processes, always under the assumption that missingness occurs. The finite-sample performances of the considered asymptotic approximations are analyzed with simulations. The practical application of the corresponding dispersion and skewness tests under missing data is demonstrated with three real-data examples.
	
\end{abstract}

%
\section{Introduction}
%
\numberwithin{theorem}{section}
Count time series have a discrete sample space, \ie the time series consists of non-negative integers (counts) from $ \bbn_0=\{0,1,\ldots\} $. These counts could either be unbounded, \ie the range is given by the full set $ \bbn_0 $, or the range could be the finite subset $ \{0,\ldots,n\} $ with a given upper bound $ n\in\bbn=\{1,2,\ldots\} $. Many models for count time series have been proposed during the last decades, see \citet{Weiss18} for an overview. For example, for unbounded counts, the first-order integer-valued autoregressive (INAR(1)) model by \cite{McK85} is quite popular, which is defined by the model recursion 
\begin{align}\label{INAR1}
	X_t=\rho\circ X_{t-1}+\epsilon_t 
\end{align} 
with $ \rho\in [0,1) $. Here, the multiplication of the ordinary AR(1) model's recursion is replaced by ``$ \circ $'', which denotes the binomial thinning operation \citep{SteuHarn79}, and $ (\epsilon_t)_{\bbz=\{\ldots,-1,0,1,\ldots\}} $ represents an independent and identically distributed (\iidno) count process of innovations. The binomial thinning is defined by requiring that $ \rho\circ X\vert X \sim \bin(X,\rho)$, and we assume that these thinnings are executed independently of other thinnings and innovations. The most popular instance of the INAR(1) family is the Poisson INAR(1) model, abbreviated as Poi-INAR(1). This model has Poisson-distributed innovations, which lead to Poisson-distributed observations, \ie $ \epsilon_t\sim \poi(\lambda) $ and $ X_t\sim \poi(\mu) $ where $ \lambda>0$ and $ \mu=\e[X_t]=\lambda/(1-\rho) $. A detailed survey about the properties of the  INAR(1) model can be found in \citet[Chapter 2.1]{Weiss18}.

Together with the INAR(1) model, \cite{McK85} also introduced the binomial AR(1) model, abbreviated as BAR(1). This model is used if one is concerned with count data that has a fixed upper limit $ n\in\bbn $. In this case, the INAR(1)'s innovation term is replaced by a further thinning, $ \beta\circ(n-X_{t-1}) $,  such that this term cannot be larger than $ n-X_{t-1} $. Thus, for a fixed upper bound $ n\in\bbn $, the BAR(1) model is defined by the recursion
\begin{align}\label{BAR1}
	X_t=\alpha\circ X_{t-1}+\beta\circ(n-X_{t-1}),
\end{align} 
for $ \pi\in(0,1)$, $ \rho\in \big(\max\bigl\{\tfrac{-\pi}{1-\pi},\tfrac{1-\pi}{-\pi} \bigr\};\ 1 \big) $ and $ \beta\coloneqq\pi(1-\rho) $, $ \alpha\coloneqq\beta+\rho $. Here, all thinnings are performed independently of each other, and the thinnings at time $ t $ are independent of $ (X_s)_{s<t} $. In addition, the condition on $ \rho $ guarantees that the thinning parameters $ \alpha,\beta\in(0,1) $. For more details about the properties of the  BAR(1) model, we refer to \citet[Chapter 3.3]{Weiss18}.

It is known that the Poi-INAR(1) and BAR(1) processes establish a stationary and ergodic Markov chain with marginal distribution \poi$ (\lambda) $ and \bin$ (n,\pi) $, respectively. Furthermore, the autocorrelation function (ACF) is \\$ \rho(h)\coloneqq \mathbb{C}orr[X_t,X_{t-h}] =\rho^h $ for $ h\geq0 $. Moreover, the Poi-INAR(1) process as well as the BAR(1) process are $\alpha$-mixing with exponentially decreasing weights \cite[p.24 \& p.60]{Weiss18}. Hence, the requirements for the central limit theorem (CLT) according to \cite{Ibram62} are satisfied and will be used later on. For further properties, the reader is referred to \cite{Alzaid88, Weiss18}.

A wide variety of models and techniques for time series analysis can be found in \citep{Box15, Weiss18}. However, when faced with missing data in real-world applications -- stemming from issues such as measuring device malfunctions \citep{Neave70}, outlier removal in statistical process monitoring \citep{WeissTest15}, or sporadic failures \citep{Scheinok65, bloom70} -- many existing models and tools often cannot be used. The challenge of incomplete data in count time series has already been acknowledged in literature \citep{AndKarl10, WangZang14, YanWang22, Zhang23}. But while these existing works predominantly focus on parameter estimation for model fitting using imputation methods, our emphasis lies in exploring common types of diagnostic statistics. Model diagnostics are crucial within the classical Box--Jenkins program, where they are used for model identification and checks of model adequacy. This perspective has not been explored for count time series so far.

Existing methods adapting standard analytical tools for real-valued time series to missing data, such as the sample ACF or spectral estimators \citep{Parzen63, Scheinok65, bloom70, DunRob81, YajNish99}, are influenced by the work of \cite{Parzen63}. They treat real-valued time series with missing observations as a special case of amplitude-modulated time series, where the amplitude-modulating process is binary and independent of the actual process. Subsequently, tools for time series are applied to the amplitude-modulated time series, with variations in their asymptotic behaviors based on the type of missing data. Here, \cite{Scheinok65} assumed missing data to be \iidno, while \cite{bloom70} allowed for serial dependence.

In this article, we derive the general asymptotic formulas for unbounded and bounded dispersion and skewness indices in the context of missing data. Moreover, we are able to derive closed-form formulas for popular special cases of  Poisson and binomial autoregressive processes, and use these expressions for model diagnostics. This article is organized as follows. Section \ref{Ch_models} provides a brief background on common univariate and bivariate count models. Section \ref{Ch_MissData} introduces the concept of amplitude modulation, setting the foundation for our subsequent derivations in Sections \ref{Ch_PInd}--\ref{Ch_SkewInd}. These sections present crucial results on the asymptotics of various indices, with full derivations available in the Appendix. Section \ref{Ch_SimStudy} investigates the finite-sample performance of our asymptotic approximations through simulations. In Section \ref{Ch_RealData}, we apply our asymptotic results to three real-data sets. Finally, Section \ref{Conclusions} draws conclusions and outlines potential future research directions.
%
\section{On univariate and bivariate count models}\label{Ch_models} 
%
We consider two types of distribution for a univariate count random variable $ X $: the Poisson distribution, abbreviated as $ \poi(\mu) $ with $ \mu\in(0,\, \infty) $ if we are concerned with unbounded counts, and the binomial distribution, abbreviated as $ \bin(n,\pi) $ with $ \pi\in(0,1) $, if we are concerned with a given upper bound $ n\in\bbn $. The Poisson distribution's mean and variance coincide and equal $ \mu $, whereas the binomial distribution's mean and variance are $ \mu=n\pi$ and $ \sigma^2= n\pi(1-\pi)$, see \citet[Chapter 3--4]{Johnson05} for more details on these distributions. We pay particular attention to the factorial moments of both distributions and therefore make use of the following definition.
\begin{definition}
	Let $ (X_t) $ be a stationary count process with existing moments. We
	define  the $k$th factorial moment as  $\mu_{(k)} = \e[(X_t)_{(k)}]$ with $k\in\bbn_0$ and $\mu_{(1)} = \mu$, where $ x_{(k)}\coloneqq x\cdots(x-k+1) $ for $k\in\mathbb{N}$ and $x_{(0)}=1$ denote the falling factorial. Furthermore, we define the mixed factorial moments as 
	\begin{align}\label{PistolPete}
		\mu_{(k,l)}(h)\coloneqq\e[(X_t)_{_{(k)}}\cdot (X_{t-h})_{_{(l)}}],
	\end{align} 
	with $ k,l\in\bbn_0 $ and $ h\in \bbz $. Note that
	
	\vspace{0.3cm} 
	\begin{tabular}{llll}
		$(i)$ & $  \mu_{(0,0)}(h) = 1, $ & $(ii)$	& $  \mu_{(k,0)}(h) = \mu_{(k)},  $\\
		$(iii)$ & $ \mu_{(0,l)}(h) = \mu_{(l)}, $ & $(iv)$ & $ \mu_{(k,l)}(h) = \mu_{(l,k)}(-h). $ 
	\end{tabular}\\
\end{definition}
For count data, factorial moments often take simple expressions, such as  $ \mu_{(k)} =\mu^k $ for the Poisson distribution and $ \mu_{(k)} = n_{(k)}\pi^k$ for the binomial distribution. When it comes to mixed factorial moments, we make use of two specific bivariate distributions: If we consider the pair $ (X_t,X_{t-h}) $ from a Poi-INAR(1) process, with lag $h\in\mathbb{N}$, then this pair is bivariate Poisson distributed, see \cite{Alzaid88}, \ie
\begin{align}\label{Xavi}
	(X_t,X_{t-h}) \sim \text{BPoi}\Big((1-\rho^h)\mu,\ (1-\rho^h)\mu,\ \rho^h\mu\Big).
\end{align}
When considering the pair $(X_t,X_{t-h})  $ from a BAR(1) process, this pair is bivariate binomial distributed, see \cite{AlexWeiss22}, \ie
\begin{align}\label{Alves}
	(X_t,X_{t-h}) \sim \text{BBin}(n,\pi,\pi,\rho^h).
\end{align}
\citet[Section 3.2 and 4]{Koch14} provide a comprehensive overview of bivariate Poisson and binomial distributions. Closed-form expressions for general mixed factorial moments can be found there. In particular, we obtain for \eqref{Xavi} and \eqref{Alves} the following propositions.
\begin{proposition}\label{Kobe}
	For the BPoi-distributed pair $ (X_t,X_{t-h}) $ from a Poi-INAR(1) process, we have a closed-form expression for the joint factorial moments of lag $h\in\mathbb{N}$, that is, 
	\begin{align*}
		\mu_{(k,s)}(h)
		=\mu_{(k)}\mu_{(s)}\sum_{i=0}^{\min{\{k,s\}}}\binom{k}{i}\binom{s}{i}i!\Bigg(\frac{\rho^h}{\mu} \Bigg)^{i}.
	\end{align*}
\end{proposition}   
\begin{proposition}\label{Hulk}
	For the BBin-distributed pair $ (X_t,X_{t-h}) $ from a BAR(1) process, we have a closed-form expression for the joint factorial moments of lag $h\in\mathbb{N}$, that is, 
	\begin{align*}
		\mu_{(k,s)}(h) 
		=n_{(k)}n_{(s)}\pi^{k+s}\sum_{i=0}^{\min{\{k,s\}}}\frac{\binom{k}{i}\binom{n-k}{s-i}}{\binom{n}{s}}\Big(1+\tfrac{1-\pi}{\pi}\rho^h \Big)^{i}.
	\end{align*}
\end{proposition}
We have a symmetric parameterization in both propositions, hence $ \mu_{(k,s)}(h)=\mu_{(s,k)}(h)$.
%
\section{Count time series with missing data}\label{Ch_MissData}
%
We adopt the idea of amplitude modulation introduced by \citet{Parzen63} to handle missing data in count time series. For a count process $ (X_t) $, we define the amplitude-modulation process $ (O_t) $ as 
\begin{align*}
	O_t=
	\begin{cases}
		1 & \text{if }X_t\text{ is observed},\\
		0 & \text{otherwise}.\\
	\end{cases}
\end{align*}
The concept of amplitude modulation can be applied to any data, but we shall focus on unbounded and bounded counts. \cite{Rubin76} concludes that the type of missingness, whether it is Missing Completely at Random (MCAR), Missing at Random (MAR), or Missing Not at Random (MNAR), may affect statistical inference. We follow \cite{bloom70, DunRob81, YajNish99} by assuming that $ (O_t) $ is independent of $ (X_t) $, which corresponds to an MCAR-type of missingness. \citet[Example 2]{Weiss21} illustrates the challenges that arise if $ (O_t) $ is not independent of $ (X_t) $. The derivation of the subsequent asymptotics under different types of missingness should be further investigated in future research.  

We shall consider diagnostic procedures for marginal distributions that use factorial moment-based statistics. Thus, for $ (X_t) $, we consider the vector-valued process $ (\X_t )$ given by
\begin{align}
	\X_t\coloneqq(X_{t,1},\ldots,X_{t,m})^\top=((X_t)_{_{(1)}},(X_t)_{_{(2)}},\ldots,(X_t)_{_{(m)}})^\top,
\end{align}
which comprises the first $ m $ falling factorials. Furthermore, let us define the amplitude modulation of $ (\X_t) $ as $ (O_t\cdot\X_t) $, and $ \overline{O\,\X}\coloneqq \tfrac{1}{n}\sum_{t=1}^{n}O_t\X_t $. Thus, the mean of $\overline{O\,\X}  $ is given by
\begin{align}\label{Magic}
	\e[\overline{O\,\X}] = \tfrac{1}{n}\sum_{t=1}^{n}\e[O_t]\e[\X_t]=\Big(\tfrac{1}{n}\sum_{t=1}^{n}\e[O_t]\Big)\bmu=\e[\overline{O}]\bmu,
\end{align}
with $\bmu\coloneqq(\mu_{(1)},\ldots,\mu_{(m)})^\top=(\e[(X_t)_{_{(1)}}],\ldots,\e[(X_t)_{_{(m)}}])^\top  $. This, however, implies to estimate $ \bmu $ by
\begin{align}\label{Kareem}
	\hat{\bmu}\coloneqq \frac{\tfrac{1}{n}\sum_{t=1}^{n}O_t\X_t}{\tfrac{1}{n}\sum_{t=1}^{n}O_t}\eqqcolon \frac{\overline{O\,\X}}{\overline{O}},
\end{align}
analogous to \cite{Weiss21}. From now on, we assume that $(O_t) $ is stationary with $ \e[O_t] =\uptau$ and autocovariance function $ \gamma_O(h)\coloneqq CoV[O_h,O_0]$, so $ \uptau(h)\coloneqq\e[O_hO_0]=\gamma_O(h)+\uptau^2 $.  Let us define the vector $ (\X_t^\ast )$ as
\begin{equation}
	\X_t^\ast\coloneqq(X_{t,0}^\ast,\ldots,X_{t,m}^\ast)^\top=O_t(1,(X_t)_{_{(1)}},(X_t)_{_{(2)}},\ldots,(X_t)_{_{(m)}})^\top.
\end{equation}
For the mean of  $ (\X_t^\ast) $, we obtain 
\begin{equation}
	\bmu^\ast\coloneqq \e[\X_t^\ast]= \uptau(1,\bmu^\top)^\top.\label{Malone}
\end{equation}
Additionally, let us assume that appropriate mixing assumptions on $ (X_t) $ and $ (O_t) $ hold, \ie such that \citet[Theorem 1.7]{Ibram62} is applicable:
\begin{parlist}
	\item moments $\e\vert X_t\vert^{2m+\delta}<\infty$  of order $2m+\delta$ with some $\delta > 0 $ exist, and the process is $\alpha$-mixing with exponentially decreasing weights. \label{A}
\end{parlist}
Here, the stationary process $(X_t) $ with probability space $(\Omega,\mathfrak{A},\Pro)$ is said to be $\alpha$-mixing with exponentially decreasing weights if there exists a non-negative sequence $(\alpha_k)_\bbn$ of weights such that $\alpha_k\leq e^{-\vartheta k}$ for some $\vartheta>0$ and each $k\in\bbn$ and such that for each $t\in\bbz$, $k\in\bbn$ and all events $E_1\in\mathfrak{A}\big(X_t,X_{t-1},\ldots\big)$, $E_2\in\mathfrak{A}\big(X_{t+k},X_{t+k+1},\ldots\big)$ the following holds:
\begin{align*}
	\vert \Pro(E_1\cap E_2)-\Pro(E_1) \Pro(E_2)\vert\leq e^{-\vartheta k}.
\end{align*}
Condition \eqref{A} is satisfied for both the Poi-INAR(1) and the BAR(1) process, see \cite{SW14, Weiss18}. The main task here is to derive a CLT for $ (\overline{\X_t^\ast})$ (see Appendix \ref{A_CLT}) and then apply the Delta method to it, resulting in the following theorem.
\begin{theorem}\label{CLT_Fac}
	Under the aforementioned assumption \eqref{A}, it holds that
	\begin{align*}
		\sqrt{T}\Big(\hat{\bmu}-\bmu\Big) \ \xrightarrow{\text{d}}\ \norm\Big(\bold{0}, \bold{\Sigma}\Big)\quad \text{with}\quad \bold{\Sigma}=(\sigma_{ij})_{i,j=1,\ldots,m},
	\end{align*}
	where $ \norm(\bold{0}, \bold{\Sigma}) $ denotes the multivariate normal distribution, and where
	\begin{align*}
		\sigma_{ij}=\tfrac{1}{\uptau}(\mu_{(i,j)}(0)-\mu_{(i)}\mu_{(j)})+\tfrac{1}{\uptau^2}\sum_{h=1}^\infty \uptau(h)\Big(\mu_{(j,i)}(h)+\mu_{(i,j)}(h)-2\mu_{(i)}\mu_{(j)}\Big).
	\end{align*}
	The bias satisfies $ \e[\hat{\bmu}-\bmu]=0+\landau (T^{-1}) $. 
\end{theorem}
\begin{proof}
	See Appendix \ref{A_CLT}.
\end{proof}
An alternative proof of the CLT with raw moments instead of factorial ones can be found in Appendix \ref{A_RawCLT}.
\begin{corollary}\label{Ilaxis}
	For the Poisson model \eqref{Xavi} and binomial model \eqref{Alves}, the covariances $\sigma_{ij}$ from Theorem \ref{CLT_Fac} are computed as
	\begin{align*}
		(i)\quad\sigma_{ij}^{\text{Poi}}&=\tfrac{1}{\uptau}\Big[\mu_{(i,j)}(0)-\mu^{i+j}+\tfrac{2}{\uptau}\sum_{h=1}^\infty \uptau(h)\Big(\mu_{(i,j)}(h)-\mu^{i+j}\Big)\Big],\\
		(ii)\quad\sigma_{ij}^{\text{Bin}}&=\tfrac{1}{\uptau}\Big[\mu_{(i,j)}(0)-n_{(i)}n_{(j)}\pi^{i+j}+\tfrac{2}{\uptau}\sum_{h=1}^\infty \uptau(h)\Big(\mu_{(i,j)}(h)-n_{(i)}n_{(j)}\pi^{i+j}\Big)\Big].
	\end{align*}
\end{corollary} 
%
\section{Poisson index of dispersion}\label{Ch_PInd}
%
One of the most well-known diagnostic statistic for unbounded counts is the index of dispersion, defined as $ I^\poi =\sigma^2/\mu$. The Poisson distribution has equidispersion since mean and variance always coincide, \ie $ I^\poi =1$. The sample counterpart to the index of dispersion is given by $ \hat{I}^\poi =\hat{\mu}_{(2)}/\hat{\mu}-\hat{\mu}+1$. Starting from Theorem \ref{CLT_Fac}, we derive the asymptotics of $ \hat{I}^\poi $, and define the function $ g $ by
\begin{align*}	
	g(x_1,x_2)=\frac{x_2}{x_1}-x_1+1. 
\end{align*}
Note that $ g(\hat{\bmu})=(g\circ\f)(\overline{\X^\ast}) $ and $ g(\bmu)= (g\circ\f)(\bmu^\ast) $ with $ \f(x_0,x_1,x_2)=(\frac{x_1}{x_0},\frac{x_2}{x_0})^\top $. Then, $ g $ has the partial derivatives
\begin{align*}
	\frac{\partial}{\partial x_1}g=-\frac{x_2}{x_1^2}-1,\quad \frac{\partial}{\partial x_2}g= \frac{1}{x_1}, \quad \frac{\partial^2}{\partial x_1^2}g= \frac{2x_2}{x_1^3},
	\quad	\frac{\partial^2}{\partial x_2^2}g= 0, \quad \frac{\partial^2}{\partial x_2 \partial x_1}g= -\frac{1}{x_1^2}.
\end{align*}
So, we get the Jacobian $ \D $ and the Hessian $ \Hes $ by evaluating the partial derivatives in $ \bmu$, which leads to 
\begin{align}\label{Aguero}
	\D=\Big(-\frac{\mu_{(2)}}{\mu^2}-1,\frac{1}{\mu}\Big), \qquad \Hes=\frac{1}{\mu^3}\begin{pmatrix} 
		2\mu_{(2)} & -\mu\\ 
		-\mu& 0 \\
	\end{pmatrix}.
\end{align}
This implies the Taylor approximation $ \hat{I}^\poi\approx I^\poi + \D(\hat{\bmu}-\bmu)+\tfrac{1}{2}(\hat{\bmu}-\bmu)^\top\Hes(\hat{\bmu}-\bmu) $, which can be used to conclude the asymptotic variance and bias of $ \hat{I}^\poi$, that is
\begin{align}
	\sigma_{\hat{I}^\poi}^2&=\tfrac{1}{T}\Big(d_1^2\sigma_{11}+d_2^2\sigma_{22}+2d_1d_2\sigma_{12}  \Big),\label{Doncic}\\
	\mathbb{B}_{\hat{I}^\poi}&=\tfrac{1}{T} \Big( \tfrac{1}{2}h_{11}\sigma_{11} +h_{12}\sigma_{12}\Big).
\end{align}
Without making specific assumptions about $ \uptau(h) $, let us derive general expressions for the asymptotic variance and bias.
\begin{theorem}\label{DeJong}
	For any count process $ (X_t) $ and amplitude-modulating process $ (O_t) $ that satisfy assumption \eqref{A}, the asymptotic variance and bias of $ \hat{I}^\poi $ are given by 
	\begin{align*}
		(i)\quad\sigma_{\hat{I}^\poi}^2&=\tfrac{1}{T\uptau\mu^2} \Bigg[ \big(\tfrac{\mu_{(2)}}{\mu}+\mu\big)^2(\mu_{(2)}+\mu)-2\big(\tfrac{\mu_{(2)}}{\mu}+\mu\big)(\mu_{(3)}+2\mu_{(2)})
		\\[1ex]&\quad+\mu_{(4)}+4\mu_{(3)}+2\mu_{(2)}-\mu^4+\tfrac{2}{\uptau}\sum_{h=1}^\infty\uptau(h)\Big( \big(\tfrac{\mu_{(2)}}{\mu}+\mu\big)^2\mu_{(1,1)}(h)
		\\[1ex]&\quad-\big(\tfrac{\mu_{(2)}}{\mu}+\mu\big)\big(\mu_{(2,1)}(h)+\mu_{(1,2)}(h)\big)+\mu_{(2,2)}(h)-\mu^4 \Big) \Bigg],\\
		(ii)\quad	\mathbb{B}_{\hat{I}^\poi}&=\tfrac{1}{T\uptau\mu^3} \Bigg[  \mu_{(2)}^2-\mu\big(\mu_{(2)}+\mu_{(3)}\big)+\tfrac{2}{\uptau}\sum_{h=1}^\infty\uptau(h)\Big(\mu_{(2)}\mu_{(1,1)}(h)
		\\[1ex]&\quad-\tfrac{\mu}{2}\big(\mu_{(2,1)}(h)+\mu_{(1,2)}(h)\big)\Big)\Bigg].
	\end{align*}
\end{theorem} 
The proof of Theorem \ref{DeJong} can be found in Appendix \ref{A_PInd_ThDeJong}, and its analogon with raw moments in Appendix \ref{A_RawPIndex}. Theorem \ref{DeJong} is model independent and is meant to be applicable across various thinning operators, such as those discussed in works of  \citet{Ris13,Borg16,Zhang17}. The (Poisson) index of dispersion is commonly used to test a Poisson null hypothesis, making it particularly relevant for time series models with Poisson marginal distribution. Therefore, we selected the Poi-INAR(1) process as an illustrative example in the following corollary to specify the asymptotic variance and bias.
\begin{corollary}\label{Nico}
	If we are concerned with a Poi-INAR(1) process and  assume that the missing data follow a stationary Markov model, \ie $ \uptau(h)=\uptau^2+\uptau(1-\uptau)r^h $ with $ \uptau,r\in(0,1) $, 
	then we obtain the asymptotic variance and bias of $ \hat{I}^\poi $ as 
	\begin{align*} 
		\sigma_{\hat{I}^\poi}^2=\tfrac{2}{T}\cdot \kappa(2), \qquad
		\mathbb{B}_{\hat{I}^\poi}=-\tfrac{1}{T}\cdot \kappa(1),
	\end{align*}	
	where
	\begin{align}\label{Kappa}
		\kappa(s)\coloneqq  \frac{1}{\uptau}\frac{1+r\rho^s}{1-r\rho^s}+\frac{2(1-r)\rho^s}{(1-r\rho^s)(1-\rho^s)} \text{ for } s\in\bbn.
	\end{align}	
\end{corollary}
\begin{proof}
	See Appendix \ref{A_PInd_CorNico}.
\end{proof}
The variance and bias from Corollary \ref{Nico} are illustrated in Figure \ref{GraphPIndex}. The graphs may be interpreted as follows: We consider missing data between the case where only 25\% of the data are known (\ie $\uptau=0.25$, or equivalently 75\% of the data are missing), and the case where all data points are known $ (\uptau=1) $. A higher percentage of missing data is not reasonable from a practical point of view. In the case of complete  data $ (\uptau=1) $, we always have $ \kappa(s)=  (1+\rho^s)/(1-\rho^s)$, so Corollary \ref{Nico} coincides to \cite{SW14}, where
\begin{align*}
	T\cdot\sigma_{\hat{I}^\poi}^2&=2\cdot\frac{1+\rho^2}{1-\rho^2}.
\end{align*}
\begin{figure}[h]
	\center\small
	\includegraphics[scale=0.47, viewport=1 40 700 310, clip
	]{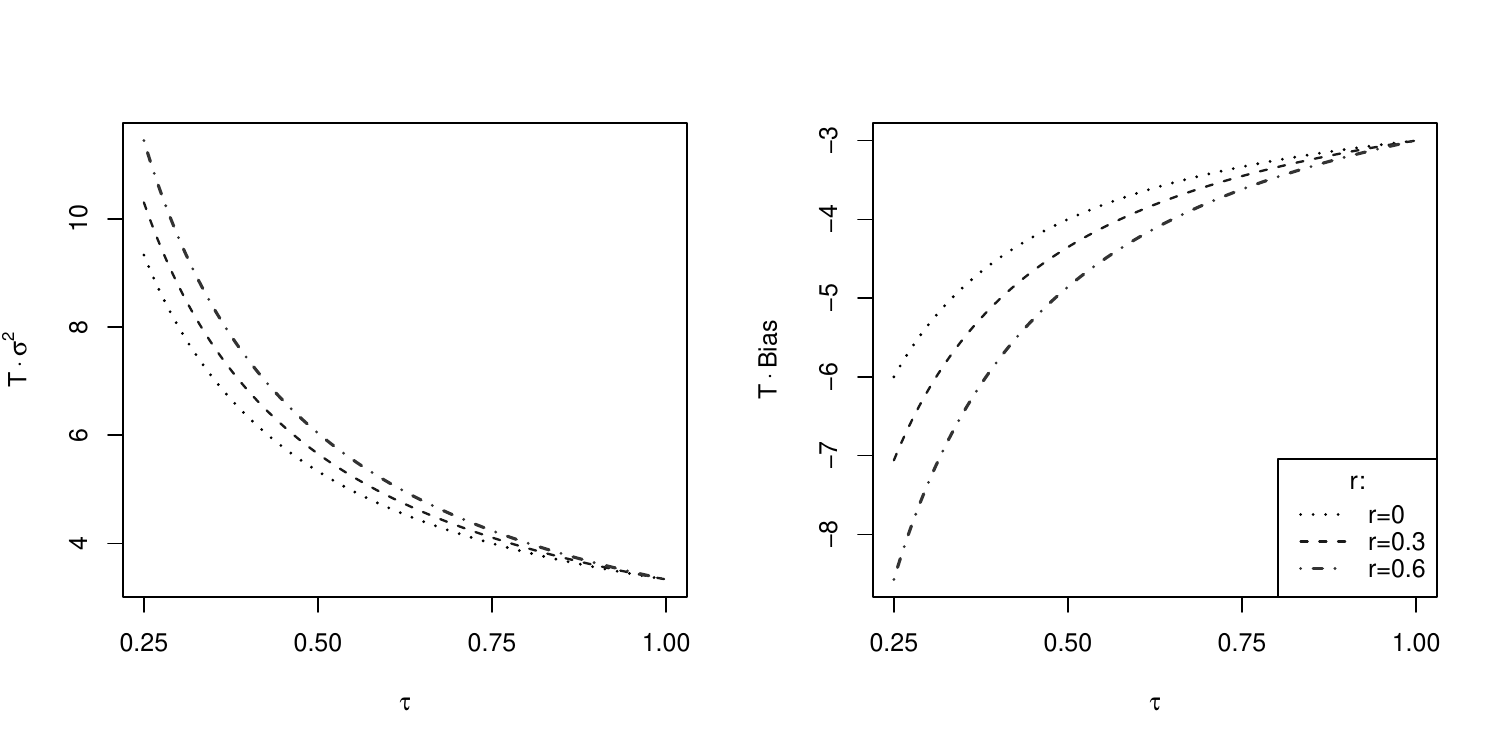} $ \uptau  $
	\caption{Plot of $ T\cdot $variance and $ T\cdot $bias according to Corollary \ref{Nico} for dispersion index of Poi-INAR(1) process with fixed $ \rho=0.5 $ and dependence parameter $ r\in\{0, 0.3, 0.6\} $.}
	\label{GraphPIndex}
\end{figure}
For the fixed $ \rho=0.5 $ in our example from Figure \ref{GraphPIndex}, we obtain a value of $\frac{10}{3} $. As one might expect, the variance rises as dependence and the extent of missing data increase. Similar results are obtained for the bias. The $ T $-fold bias has a value of $ -3 $, which is in the upper right corner, for no dependency and completely known data. Then, increasing dependence and the rise of missing data result in stronger negative values for the bias.
%
\section{Binomial index of dispersion}\label{Ch_BinInd}
%
If we would choose the (Poisson) index of dispersion for bounded counts, then the binomial distribution, which has a finite range, would result in underdispersion, \ie $ I^\poi<1 $. This phenomenon is caused by the fact the $ I^\poi $ does not account for the boundedness of the range. Thus, the binomial index of dispersion, defined as $ I^\bin =\sigma^2/(\mu(1-\frac{\mu}{n}))$, is the more appropriate choice, which results in equidispersion if being concerned with binomial counts. The binomial index of dispersion's sample counterpart is given by  $ \hat{I}^\bin =(\hat{\mu}_{(2)}+\hat{\mu}-\hat{\mu}^2)/(\hat{\mu}(1-\frac{\hat{\mu}}{n}))$. Starting from Theorem \ref{CLT_Fac}, we derive the asymptotics of $ \hat{I}^\bin $, and define the function $ g$  by
\begin{align*}	
	g(x_1,x_2)=\frac{n(x_2+x_1-x_1^2)}{x_1(n-x_1)}. 
\end{align*}
Then, $ g $ has the partial derivatives
\begin{align*}
	&\frac{\partial}{\partial x_1}g=\frac{n(x_1^2(1-n)-nx_2+2x_1x_2)}{x_1^2(n-x_1)^2}, \qquad 	\frac{\partial}{\partial x_2}g= \frac{n}{x_1(n-x_1)}, \qquad	\frac{\partial^2}{\partial x_2^2}g= 0, \\[0.1cm] 
	&\frac{\partial^2}{\partial x_1^2}g=\frac{2n\big(x_1^3(1-n)+n^2x_2+3x_1x_2(x_1-n)\big)}{x_1^3(n-x_1)^3},
	\qquad \frac{\partial^2}{\partial x_2 \partial x_1}g= \frac{n(2x_1-n)}{x_1^2(n-x_1)^2}.
\end{align*}
So, we get the Jacobian $ \D $ and the Hessian $ \Hes $ by evaluating the partial derivatives in $ \bmu$, which leads to 
\begin{align}\label{Gavi1}
	\D&=\frac{n}{\mu(n-\mu)}\Bigg(\frac{\mu^2(1-n)-n\mu_{(2)}+2\mu\mu_{(2)}}{\mu(n-\mu)} ,\ 1\Bigg),\\[0.1cm] 
	h_{11}&=\frac{2n\big(\mu^3(1-n)+n^2\mu_{(2)}+3\mu\mu_{(2)}(\mu-n)\big)}{\mu^3(n-\mu)^3},\qquad h_{12}= \frac{n(2\mu-n)}{\mu^2(n-\mu)^2}.\label{Gavi2}
\end{align}
This leads to the Taylor approximation $ \hat{I}^\bin\approx I^\bin + \D(\hat{\bmu}-\bmu)+\tfrac{1}{2}(\hat{\bmu}-\bmu)^\top\Hes(\hat{\bmu}-\bmu) $, which can be used to conclude the asymptotic variance and bias of $ \hat{I}^{\bin} $ as
\begin{align}
	\sigma_{\hat{I}^\bin}^2&=\tfrac{1}{T} \Big( d_1^2\sigma_{11}  +d_2^2\sigma_{22}+2 d_1d_2\sigma_{12}  \Big),\label{Depay}\\
	\mathbb{B}_{\hat{I}^\bin}&=\tfrac{1}{T} \Big( \tfrac{1}{2}h_{11}\sigma_{11} +h_{12}\sigma_{12}  \Big).
\end{align}
Again, we first derive general expressions for the asymptotic variance and bias, without making assumptions on $ \uptau(h) $, analogous to Section \ref{Ch_PInd}.
\begin{theorem}\label{Memphis}
	For any count process $ (X_t) $ and amplitude-modulating process $ (O_t)$ that satisfies assumption \eqref{A}, the asymptotic variance and bias of $ \hat{I}^\bin $ are given by 
	\begin{align*}
		(i)\quad \sigma_{\hat{I}^\bin}^2&=\tfrac{n^2}{T\uptau\mu^4(n-\mu)^4} \Bigg[ \big(\mu^2(1-n)-n\mu_{(2)}+2\mu\mu_{(2)}\big)^2\big(\mu_{(2)}+\mu-\mu^2\big)
		\\[1ex]&\quad +2\mu(n-\mu)\big(\mu^2(1-n)-n\mu_{(2)}+2\mu\mu_{(2)}\big)\big(\mu_{(3)}+2\mu_{(2)}-\mu\mu_{(2)}\big)
		\\[1ex]&\quad+\mu^2(n-\mu)^2\big(\mu_{(4)}+4\mu_{(3)}+2\mu_{(2)}-\mu_{(2)}^2\big)
		\\[1ex]&\quad +\tfrac{2}{\uptau}\sum_{h=1}^\infty\uptau(h)\Bigg( \big(\mu^2(1-n)-n\mu_{(2)}+2\mu\mu_{(2)}\big)^2\big(\mu_{(1,1)}(h)-\mu^2\big)
		\\[1ex]&\quad+\mu(n-\mu) \big(\mu^2(1-n)-n\mu_{(2)}+2\mu\mu_{(2)}\big)\big(\mu_{(2,1)}(h)+\mu_{(1,2)}(h)-2\mu\mu_{(2)}\big)
		\\[1ex]&\quad +\mu^2(n-\mu)^2\big(\mu_{(2,2)}(h)-\mu_{(2)}^2\big) \Bigg) \Bigg],\\
		(ii)\quad \mathbb{B}_{\hat{I}^\bin}&=\tfrac{n}{T\uptau\mu^3(n-\mu)^3} \Bigg[\big(\mu^3(1-n)+n^2\mu_{(2)}+3\mu\mu_{(2)}(\mu-n)\big)\big(\mu_{(2)}+\mu-\mu^2\big)
		\\[1ex]&\quad+\mu(n-\mu)(2\mu-n)(\mu_{(3)}+2\mu_{(2)}-\mu\mu_{(2)})
		\\[1ex]&\quad+\tfrac{1}{\uptau^2}\sum_{h=1}^\infty\uptau(h)\Bigg(2\big(\mu^3(1-n)+n^2\mu_{(2)}+3\mu\mu_{(2)}(\mu-n)\big)\big(\mu_{(1,1)}(h)-\mu^2\big)
		\\[1ex]&\quad+\mu(n-\mu)(2\mu-n)\big(\mu_{(2,1)}(h)+\mu_{(1,2)}(h)-2\mu\mu_{(2)}\big)\Bigg)\Bigg].
	\end{align*}
\end{theorem} 
The proof of Theorem \ref{Memphis} can be found in Appendix \ref{A_PInd_ThMemphis}. Once again, the theorem is model independent and is meant to be applicable across various data generating processes as mentioned before. This time, we have chosen the BAR(1) process as an illustrative example, because it has binomial marginal distribution. Thus, the binomial index of dispersion can be used to test a binomial null hypothesis. The following corollary specify the asymptotic variance and bias of the BAR(1) process.
\begin{corollary}\label{Demir}
	If we are concerned with a BAR(1) process and assume that the missing data follow a Markov model, \ie $ \uptau(h)=\uptau^2+\uptau(1-\uptau)r^h $, 
	then we obtain the asymptotic variance and bias of $ \hat{I}^\bin $ as
	\begin{align*}
		\sigma_{\hat{I}^\bin}^2=\tfrac{2}{T}\Big(1-\tfrac{1}{n}\Big)\cdot\kappa(2),
		\qquad
		\mathbb{B}_{\hat{I}^\bin}=-\tfrac{1}{T}\Big(1-\tfrac{1}{n}\Big)\cdot\kappa(1),
	\end{align*}	
	where $ \kappa(1) $, $ \kappa(2) $ are given by \eqref{Kappa}.
\end{corollary}
\begin{proof}
	See Appendix \ref{A_PInd_CorDemir}.
\end{proof}
\begin{remark}
For the binomial index of dispersion, we can infer the following: If $ n\rightarrow\infty $, then the asymptotic variance and asymptotic bias coincide with the asymptotics for the Poi-INAR(1) process according to Corollary \ref{Nico}. In fact, the variance and bias are just compressed by the factor $(1-\frac{1}{n})$, thus we omit a further illustration here.	Note that setting the missingness parameter $ \uptau=1 $ leads to 
	\begin{align*}
		\sigma_{\hat{I}^{\text{Bin}}}^2=\frac{2}{T}\Big(1-\tfrac{1}{n}\Big) \frac{1+\rho^2}{1-\rho^2},
	\end{align*}
	which coincides with \citet[p.62]{Weiss18}. 
\end{remark}
%
\section{Skewness index}\label{Ch_SkewInd}
%
To analyze the marginal distribution of a count process, one should not solely look at dispersion properties, but also consider further shape characteristics such as skewness. Thus, we define the skewness index $ I_\skw =\mu_{(3)}/(\mu_{(2)}\mu)$ as in\citep{AlexWeiss22, AlexWeiss23}, which, unlike the Poisson/binomial index of dispersion, is not constrained to unbounded or bounded counts. The skewness index is $ I_\skw^\poi =1$ for the Poisson distribution since $ \mu_{(k)}=\mu^k $, and $ I_\skw^\bin =1-2/n$ for the binomial distribution since $  \mu_{(k)}=n_{(k)}\pi^k  $. The sample counterpart to the skewness index is given by $ \hat{I}_\skw =\hat{\mu}_{(3)}/(\hat{\mu}_{(2)}\hat{\mu})$. Similar to how the two other indices' asymptotic variance and bias are derived, those for the skewness index can be found in Appendix \ref{A_SkewInd}.
In the case that the skewness index is applied to Poisson or binomial counts, the results for the  Jacobian $ \D $ and the Hessian $ \Hes $ can be found in Lemma \ref{Batman}.
Now, utilizing the Taylor approximation $ \hat{I}_\skw\approx I_\skw + \D(\hat{\bmu}-\bmu)+\tfrac{1}{2}(\hat{\bmu}-\bmu)^\top\Hes(\hat{\bmu}-\bmu) $, we conclude the asymptotic variance and bias of $ \hat{I}_\skw $ as
\begin{align}
	\sigma_{\hat{I}_\skw}^2&=\tfrac{1}{T} \Big( d_1^2\sigma_{11} +d_2^2\sigma_{22} +d_3^2\sigma_{33} +2 d_1d_2\sigma_{12} +2 d_1d_3\sigma_{13}  +2 d_2d_3\sigma_{23} \Big),\\
	\mathbb{B}_{\hat{I}_\skw}&=\tfrac{1}{T} \Big( \tfrac{1}{2}\big( h_{11}\sigma_{11} +h_{22}\sigma_{22}\big) +h_{12}\sigma_{12} +h_{13}\sigma_{13}  +h_{23}\sigma_{23} \Big).
\end{align}
To calculate the variance and bias of the following theorems, we consider the Poi-INAR(1) and BAR(1) processes.
\begin{theorem}\label{Ronaldo1}
	If we are concerned with a Poi-INAR(1) process and assume that the missing data follow a Markov model, \ie $ \uptau(h)=\uptau^2+\uptau(1-\uptau)r^h $, we obtain the asymptotic variance and bias of  $ \hat{I}_\skw^\poi $ as
	\begin{align*} 
		\sigma_{\hat{I}_\skw^\poi}^2=\tfrac{1}{T\mu^3}\Big[8\mu\cdot\kappa(2)+6\cdot\kappa(3)\Big],\qquad
		\mathbb{B}_{\hat{I}_\skw^\poi}=-\tfrac{2}{T\mu^2}\Big[
		\mu\cdot\kappa(1)+2\cdot\kappa(2)\Big],
	\end{align*}	
	recall \eqref{Kappa} for the definition of $ \kappa(s) $.
\end{theorem}
\begin{proof}
	See Appendix \ref{A_SkewInd_ThRo1}.
\end{proof}
Figure \ref{GraphSkewPIndex} shows the variance and bias from Theorem \ref{Ronaldo1}.
By contrast to Corollary \ref{Nico}, we have dependence on the parameter $ \mu $, and thus the variance and bias both get smaller with increasing $ \mu $. Other than that, the results for growing missingness and increasing dependence are comparable to those from Figure \ref{GraphPIndex}.
\begin{figure}[h]
	\center\small
	\includegraphics[scale=0.47, viewport=1 40 700 310, clip
	]{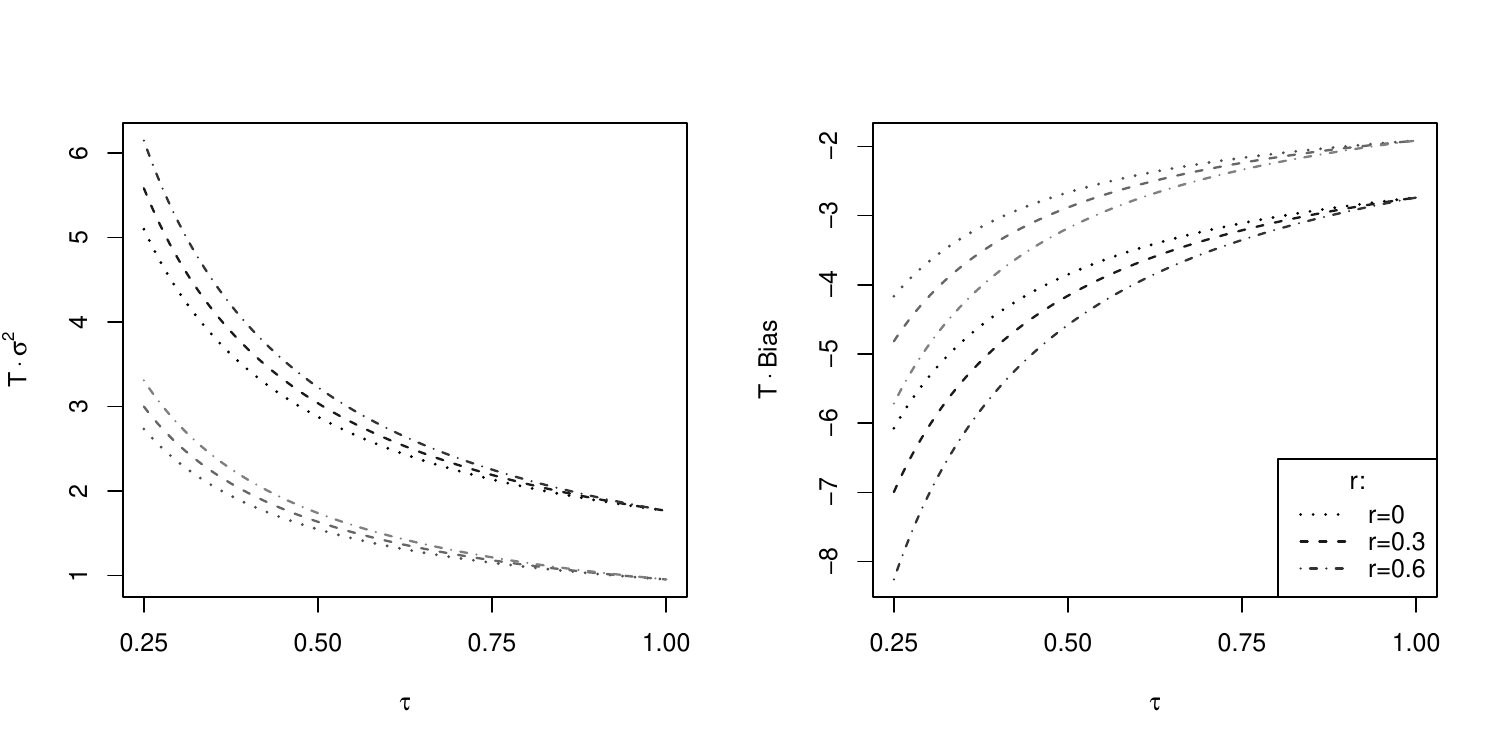} $ \uptau $
	\caption{Plot of the $ T\cdot $variance and $ T\cdot $bias for the skewness index of a Poi-INAR(1) process with fixed $ \rho=0.5 $, $ \mu\in\{3,4\}$, and dependence parameter $ r\in\{0, 0.3, 0.6\} $. Here, we have lighter gray levels for increasing $ \mu $.}
	\label{GraphSkewPIndex}
\end{figure}
\begin{theorem}\label{Ronaldo2}
	If we are concerned with a BAR(1) process and assume that the missing data follow a Markov model, \ie $ \uptau(h)=\uptau^2+\uptau(1-\uptau)r^h $, we obtain the asymptotic variance and bias of  $ \hat{I}_\skw^\bin $ as
	\begin{align*} 
		\sigma_{\hat{I}_\skw^\bin}^2&=\Big(\tfrac{(n-2)(n-\mu)^3}{(n-1)n^3}\Big)\tfrac{1}{T\mu^3}\Bigg[\tfrac{n-2}{n-\mu}\cdot8\mu\cdot\kappa(2)+6\cdot\kappa(3)\Bigg],\\
		\mathbb{B}_{\hat{I}_\skw^\bin}&=-\Big(\tfrac{(n-2)(n-\mu)^2}{(n-1)n^2}\Big)\tfrac{2}{T\mu^2}\Bigg[\tfrac{n-1}{n-\mu}\cdot\mu\cdot\kappa(1)+2\cdot\kappa(2)\Bigg].
	\end{align*}	
\end{theorem}
\begin{proof}
	See Appendix \ref{A_SkewInd_ThRo2}.
\end{proof}
The variance and bias from Theorem \ref{Ronaldo2} are displayed in Figure \ref{GraphSkewBinIndex}. Since the variance and bias each have an additional parameter $ n $ in their equations, we analyze three options and compare them against each other. The asymptotic variance of the skewness index for a BAR(1) process is quite similar to that of the skewness index for a Poi-INAR(1) process, as can be seen by comparing the aforementioned theorems. For the variance and bias, respectively, we plotted nine graphs. They are in groups of three, where each group addresses the same dependence parameter $ r\in\{0, 0.3, 0.6\} $ and one choice of $ n $. The graphs are shifted upwards for the variance and downwards for the bias if $ n $ gets larger. Again, the results for growing dependence and missingness are comparable to those from Figure \ref{GraphPIndex} for both variance and bias. The skewness index for the Poi-INAR(1) and BAR(1) processes are clearly linked as follows:
\begin{align*}
	\sigma_{\hat{I}_\skw^\bin}^2 \xrightarrow[n\to\infty]{} \sigma_{\hat{I}_\skw^\poi}^2,\quad 	\mathbb{B}_{\hat{I}_\skw^\bin} \xrightarrow[n\to\infty]{}	\mathbb{B}_{\hat{I}_\skw^\poi},
\end{align*}
where $ \mu $ is kept fixed.\\
\begin{figure}[h]
	\center\small
	\includegraphics[scale=0.47, viewport=1 40 700 310, clip]{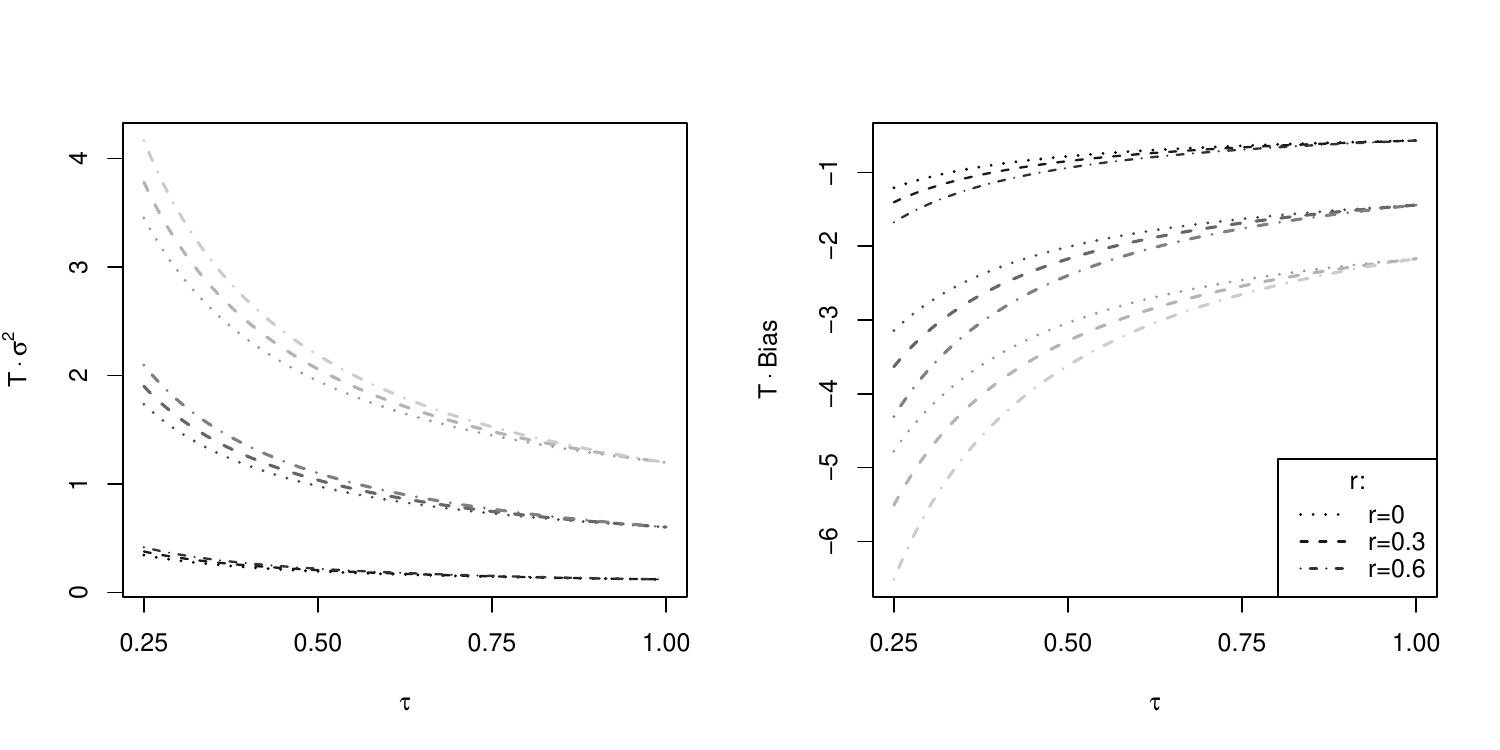} $ \uptau $
	\caption{Plot of the $ T\cdot $variance and $ T\cdot $bias for the skewness index of a BAR(1) process with fixed $ \rho=0.5 $, $ \mu=3 $, dependence parameter $ r\in\{0, 0.3, 0.6\} $, and $ n\in\{5, 10, 25\} $. Here, we have lighter gray levels for increasing $ n $.}
	\label{GraphSkewBinIndex}
\end{figure}
%
\section{Simulation Study}\label{Ch_SimStudy}
%
A simulation analysis with 10,000 replications per scenario is done to assess the finite-sample performance of the asymptotic results established in Sections \ref{Ch_PInd}--\ref{Ch_SkewInd}. We assume the Poi-INAR(1) model for unbounded counts, which has the Poisson marginal distribution \poi($\mu$) with $\mu\in(0,\infty)$, and the BAR(1) model for bounded counts, which has the binomial marginal distribution \bin($n,\pi$) with $\pi\in(0,1)$ and $ n\in\bbn $. Moreover, in both scenarios, we take into account the sample sizes $ T\in\{100,250,500,1000\} $, the dependence parameters $ r\in\{0,0.3,0.6\} $ ($ r=0 $ leads to \iid $ (O_t) $), and the probability $ \uptau\in\{1,0.8,0.6,0.4\} $ of the amplitude-modulating process ($ O_t $). Furthermore, we fix $ \rho=0.5 $ and $ \mu=3 $ ($ n\pi=3 $), and we consider $ n\in \{10,25\}$ in the bounded case. We compute the $ I^\poi $, $ I^\bin $, and $ I_\skw$ for each of these scenarios and compare the sample properties obtained from the simulated data to the asymptotic formulas derived in Sections  \ref{Ch_PInd}--\ref{Ch_SkewInd}.
\begin{table}[h]
	\centering
	\caption{Asymptotic vs. simulated mean and SD of $\hat{I}^\poi$ and $\hat{I}^\poi_{\skw}$ data; time series of length $ T $ is generated by 
		Poi-INAR(1)  counts with fixed $\mu=3$, $ \rho=0.5 $.}
	\label{TabPoi}
	\smallskip
	\scalebox{0.85}{
		\begin{tabular}{llr|rlrl|rlrl}
			\hline
			&&& \multicolumn{2}{c}{mean of $\hat{I}^\poi$} & \multicolumn{2}{c|}{SD of $\hat{I}^\poi$}& \multicolumn{2}{c}{mean of $\hat{I}^\poi_{\skw}$} & \multicolumn{2}{c}{SD of $\hat{I}^\poi_{\skw}$} \\
			$\uptau$ & $r$ & $T$ & $ \quad $sim$ \quad $ & asym$ \quad $ & $ \quad $sim$ \quad $ & asym$ \quad $  & $ \quad $sim$ \quad $ & asym$ \quad $ & $ \quad $sim$ \quad $ & asym$ \quad $ \\
			\hline
			0.8 & 0    & 100 & 0.967 & 0.968 & 0.189 & 0.196 & 0.971 & 0.970 & 0.132 & 0.143 \\ 
			& 0.6 & 100 & 0.968 & 0.965 & 0.191 & 0.200 & 0.971 & 0.968 & 0.132 & 0.146 \\ \cline{2-11}
			& 0    & 250 & 0.987 & 0.987 & 0.124 & 0.124 & 0.989 & 0.988 & 0.089 & 0.090 \\ 
			& 0.6 & 250 & 0.987 & 0.986 & 0.126 & 0.127 & 0.987 & 0.987 & 0.089 & 0.092 \\ \cline{2-11}
			& 0    & 500 & 0.994 & 0.994 & 0.088 & 0.088 & 0.995 & 0.994 & 0.063 & 0.064 \\ 
			& 0.6 & 500 & 0.994 & 0.993 & 0.088 & 0.090 & 0.993 & 0.994 & 0.063 & 0.065 \\ 
			\hline\hline
			0.6 & 0    & 100 & 0.965 & 0.963 & 0.210 & 0.216 & 0.966 & 0.965 & 0.146 & 0.158 \\ 		
			& 0.6 & 100 & 0.956 & 0.958 & 0.219 & 0.227 & 0.961 & 0.960 & 0.151 & 0.166 \\ \cline{2-11}
			& 0    & 250 & 0.987 & 0.985 & 0.134 & 0.137 & 0.987 & 0.986 & 0.096 & 0.100 \\ 
			& 0.6 & 250 & 0.984 & 0.983 & 0.140 & 0.143 & 0.985 & 0.984 & 0.100 & 0.105 \\ \cline{2-11}
			& 0    & 500 & 0.992 & 0.993 & 0.096 & 0.097 & 0.992 & 0.993 & 0.069 & 0.071 \\ 
			& 0.6 & 500 & 0.991 & 0.992 & 0.100 & 0.101 & 0.991 & 0.992 & 0.072 & 0.074 \\ 
			\hline
	\end{tabular}}
\end{table}
The mean and standard deviation (SD) (simulated vs.\ asymptotic) of the $\hat{I}^\poi$ and $\hat{I}^\poi_{\skw}$ are presented in Table \ref{TabPoi}. Only two degrees of missingness are shown here: $ \uptau=0.8 $ (20\% missing data) and $ \uptau=0.6 $ (40\% missing data), as well as the dependence $ r\in\{0,0.6\} $ and sample sizes $ T\in\{100,250,500\} $. Appendix \ref{A_Tables} contains the entire table. With decreasing $ \uptau $, the negative bias becomes stronger, while the SDs increase in all scenarios. The crucial takeaway is that the asymptotics do a good job in capturing these changes, with the exception of $ T=100 $, where there is a minor difference between simulated and asymptotic SD for $\hat{I}^\poi$. The differences are not  problematic for practice for $ T=100 $, because test methods based on these asymptotic formulas have only a conservative effect. Interestingly, the same statements hold for $\hat{I}^\poi_{\skw}$, where we again see a modest deviation between simulation and asymptotic SD. As a result, the derived asymptotics are well-suited to approximate the true finite-sample distributions of $\hat{I}^\poi$ and $\hat{I}^\poi_{\skw}$.
\begin{table}[h]
	\centering
	\caption{Asymptotic vs. simulated mean and SD of $\hat{I}^\bin$ and $\hat{I}^\bin_{\skw}$ data; time series of length $ T $ is generated by 
		BAR(1)  counts with  fixed $\mu=3$ and $ \rho=0.5 $.}
	\label{TabBin}
	\smallskip
	\scalebox{0.8}{
		\begin{tabular}{lllr|rlrl|rlrl}
			\hline
			&&&& \multicolumn{2}{c}{mean of $\hat{I}^\bin$} & \multicolumn{2}{c|}{SD of $\hat{I}^\bin$}& \multicolumn{2}{c}{mean of $\hat{I}^\bin_{\skw}$} & \multicolumn{2}{c}{SD of $\hat{I}^\bin_{\skw}$} \\
			$\uptau$ & $ n $ & $r$ & $T$  & $ \quad $sim$ \quad $ & asym$ \quad $ & $ \quad $sim$ \quad $ & asym$ \quad $  & $ \quad $sim$ \quad $ & asym$ \quad $ & $ \quad $sim$ \quad $ & asym$ \quad $ \\
			\hline		
			0.8 & 10 & 0    & 250 & 0.987 & 0.988 & 0.118 & 0.117 & 0.793 & 0.794 & 0.053 & 0.053 \\ 
			& 10 & 0.6 & 250 & 0.986 & 0.988 & 0.119 & 0.120 & 0.793 & 0.793 & 0.054 & 0.054 \\ 
			& 25 & 0    & 250 & 0.987 & 0.988 & 0.119 & 0.121 & 0.910 & 0.910 & 0.072 & 0.074 \\ 
			& 25 & 0.6 & 250 & 0.987 & 0.987 & 0.123 & 0.124 & 0.909 & 0.910 & 0.074 & 0.076 \\ \cline{2-12}
			& 10 & 0    & 500 & 0.994 & 0.994 & 0.082 & 0.083 & 0.797 & 0.797 & 0.037 & 0.037 \\ 
			& 10 & 0.6 & 500 & 0.993 & 0.994 & 0.085 & 0.085 & 0.796 & 0.797 & 0.038 & 0.038 \\ 
			& 25 & 0    & 500 & 0.995 & 0.994 & 0.085 & 0.086 & 0.916 & 0.915 & 0.052 & 0.053 \\ 
			& 25 & 0.6 & 500 & 0.995 & 0.993 & 0.089 & 0.088 & 0.916 & 0.915 & 0.054 & 0.054 \\ 
			\hline\hline
			0.6 & 10 & 0    & 250 & 0.986 & 0.987 & 0.128 & 0.130 & 0.792 & 0.793 & 0.057 & 0.058 \\ 
			& 10 & 0.6 & 250 & 0.986 & 0.985 & 0.135 & 0.136 & 0.792 & 0.792 & 0.061 & 0.061 \\ 
			& 25 & 0    & 250 & 0.986 & 0.986 & 0.132 & 0.134 & 0.909 & 0.909 & 0.080 & 0.082 \\ 
			& 25 & 0.6 & 250 & 0.982 & 0.984 & 0.137 & 0.140 & 0.906 & 0.907 & 0.083 & 0.086 \\ \cline{2-12}			
			& 10 & 0    & 500 & 0.992 & 0.993 & 0.090 & 0.092 & 0.796 & 0.796 & 0.040 & 0.041 \\ 
			& 10 & 0.6 & 500 & 0.991 & 0.992 & 0.096 & 0.096 & 0.795 & 0.796 & 0.043 & 0.043 \\ 
			& 25 & 0    & 500 & 0.991 & 0.993 & 0.093 & 0.095 & 0.914 & 0.915 & 0.057 & 0.058 \\ 			
			& 25 & 0.6 & 500 & 0.992 & 0.992 & 0.099 & 0.099 & 0.913 & 0.914 & 0.060 & 0.061 \\ 
			\hline
	\end{tabular}}
\end{table}
In Table \ref{TabBin}, the mean and SD (simulated vs. asymptotic) of $\hat{I}^\bin$ and $\hat{I}^\bin_{\skw}$ are displayed. In this case, we consider the sample sizes $ T\in\{250,500\} $ and $ n\in\{10,25\} $, and we again present two levels of missingness $ \uptau\in\{0.8,0.6\} $ as well as dependence $ r\in\{0,0.6\} $. The whole table can be found in Appendix \ref{A_Tables}. We see a fairly good agreement between the asymptotic approximation and the true sampling distribution for $\hat{I}^\bin$. Note that decreasing $ \uptau $ intensifies the negative bias and increases the SDs in any scenario. There is a minor influence of the two chosen $ n $, which has no effect on the bias but has a minor effect on the SD. By focusing on $\hat{I}^\bin_{\skw}$, we can confirm that the asymptotic approximation and the true sampling distribution are in good agreement once more. In every circumstance, decreasing $ \uptau $ intensifies the negative bias and increases the SDs. The choice of $ n $, on the other hand, has a considerable impact on the bias and SDs. A larger $ n $ causes a substantial spike in the negative bias as well as a large increase in the SD, but our asymptotics successfully capture such spikes. Because the parameter $ n $ is significantly more entangled in the equations in Theorem  \ref{Ronaldo2} than in Corollary \ref{Demir} for the binomial index of dispersion, this divergence for the skewness index  is easily explained. Consequently, the derived asymptotics are well-suited for approximating the true finite-sample distributions of $\hat{I}^\bin$ and $\hat{I}^\bin_{\skw}$.
%
\section{Real-data applications}\label{Ch_RealData}
%
The amplitude modulation technique introduced in Section \ref{Ch_MissData} was successfully utilized to derive the asymptotic variance and bias of different diagnostic statistics. The simulations of Section \ref{Ch_SimStudy} showed that these asymptotics provide a good approximation of their true finite sample distributions. Thus, it is justified to use the asymptotics to implement diagnostic tests regarding the count time series' marginal distribution. In what follows, we are faced with two cases for the amplitude-modulating process $( O_t )$, \ie $ (O_t) $ follows a Markov model or is \iidno {} Thus, we consider Corollary \ref{Demir} and Theorem \ref{Ronaldo2} for bounded counts and Corollary \ref{DeJong} and Theorem \ref{Ronaldo1} for unbounded counts, which all assume a Markov model, but also include the case of \iid missing data if $ r=0 $. The next sections discuss three data examples.
%
\subsection{Peak severity counts}\label{Ch_PeakCounts}
%
\begin{figure}[h]
	\center\small
	\includegraphics[scale=0.55, viewport=1 40 600 240, clip
	]{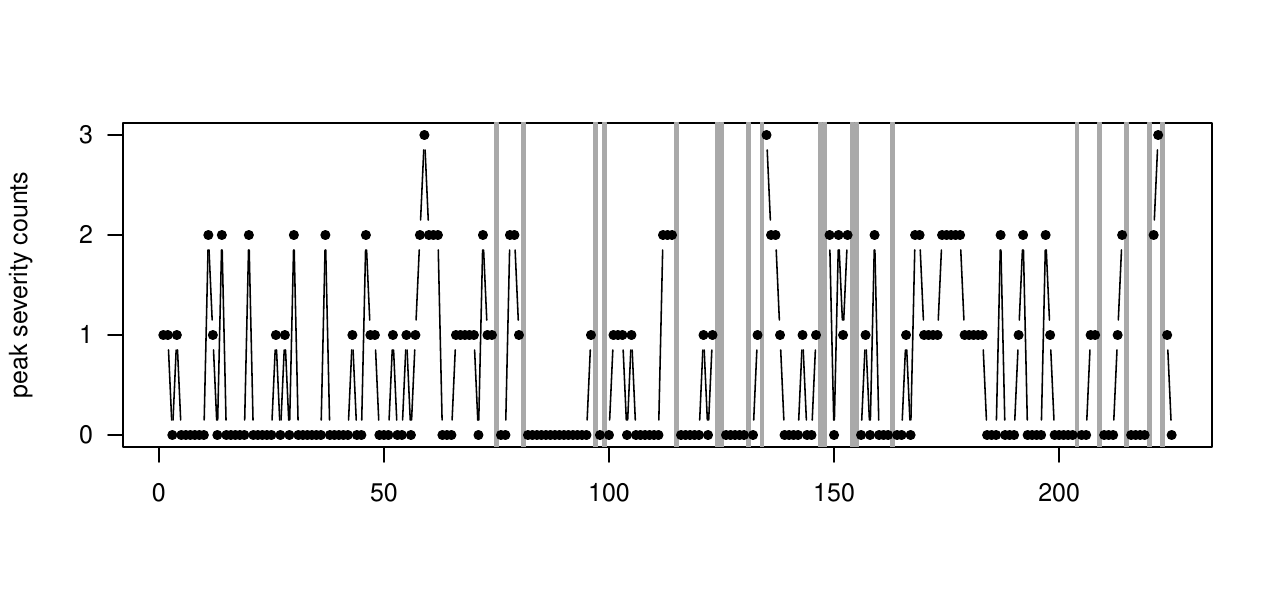}t
	\caption{Plot of daily peak severity counts. The vertical lines indicate missing data.}
	\label{MigraineData}
\end{figure}
\noindent In our first data example, we analyze a certain type of migraine data, which originates from a mobile app, N1-Headache$ ^{\text{TM}} $  that was developed by Curelator Inc. Here, patients log various information into the app, like migraine symptoms and medication as well as a range of factors (moods, weather, diet, etc.). We are interested in the pain peak severity counts with range $ \{0,\ldots,3\} $ which originate from applying a rank count approach to the original data. As an illustration, we look at patient $ A $ from \citet[Figure 1]{Weiss21}. Due to confidentiality, we read and reproduce the data in Figure \ref{MigraineData}. We get a time series of daily peak severity of length $ T=225 $ with $ (8.4\ \%) $ missing data represented by the gray vertical lines. Here, missingness occurs because the patient fails to provide the information, \eg to enter the severity into the app. 

For the peak severity counts, the data set is bounded with an upper bound of $ n=3 $. We use the asymptotics for bounded counts derived in Sections \ref{Ch_BinInd}--\ref{Ch_SkewInd}. First, let us analyze the pattern of the amplitude-modulation $ (O_t) $ in more detail. In the left part of Figure \ref{Migraine_Ot_Xt}, the sample  ACF and the sample partial ACF  (PACF) of the $ O_t $ are presented. Since there are no significant dependencies in eather the ACF or PACF, it is reasonable to assume that $ O_t $ can be treated as \iid  data, which is in agreement with analogous conclusions in \cite{Weiss21}. Thus, we can set the dependence parameter $ \hat{r}=0 $, which is the special case of the Markov model considered in our derivations. Furthermore, we have $ (8.4\ \%) $ missing values which corresponds to $ \hat{\uptau}=0.916 $.
\begin{figure}[h]
	\center\small
	\includegraphics[scale=0.57, viewport=1 40 700 300, clip]{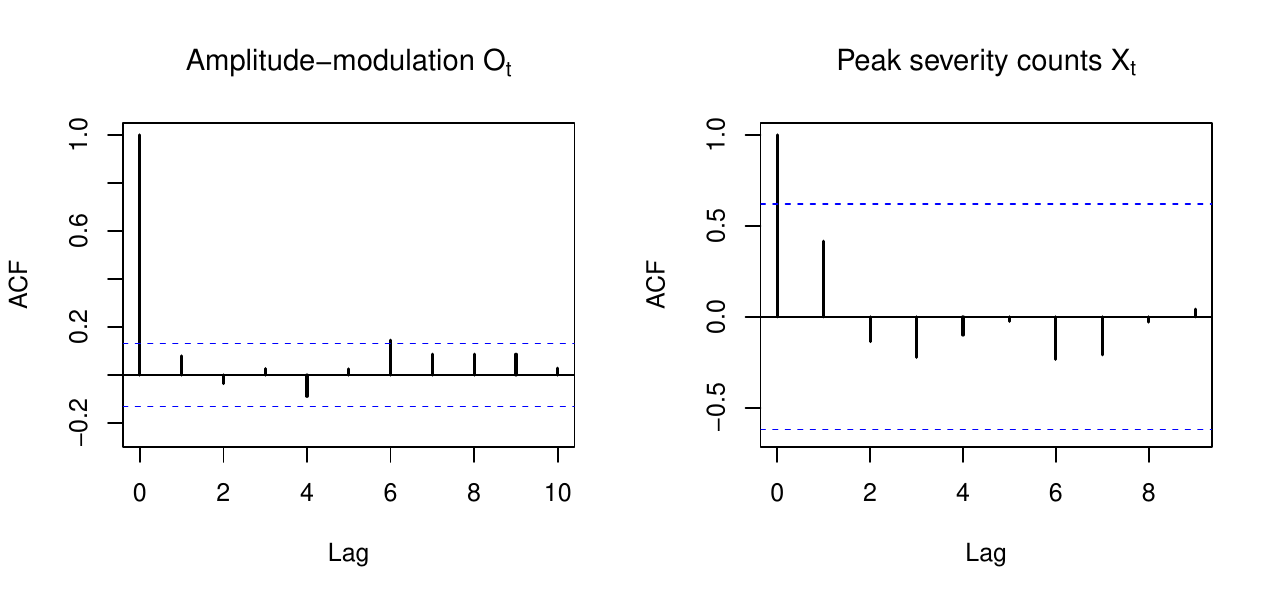}
	\includegraphics[scale=0.57, viewport=1 40 700 300, clip]{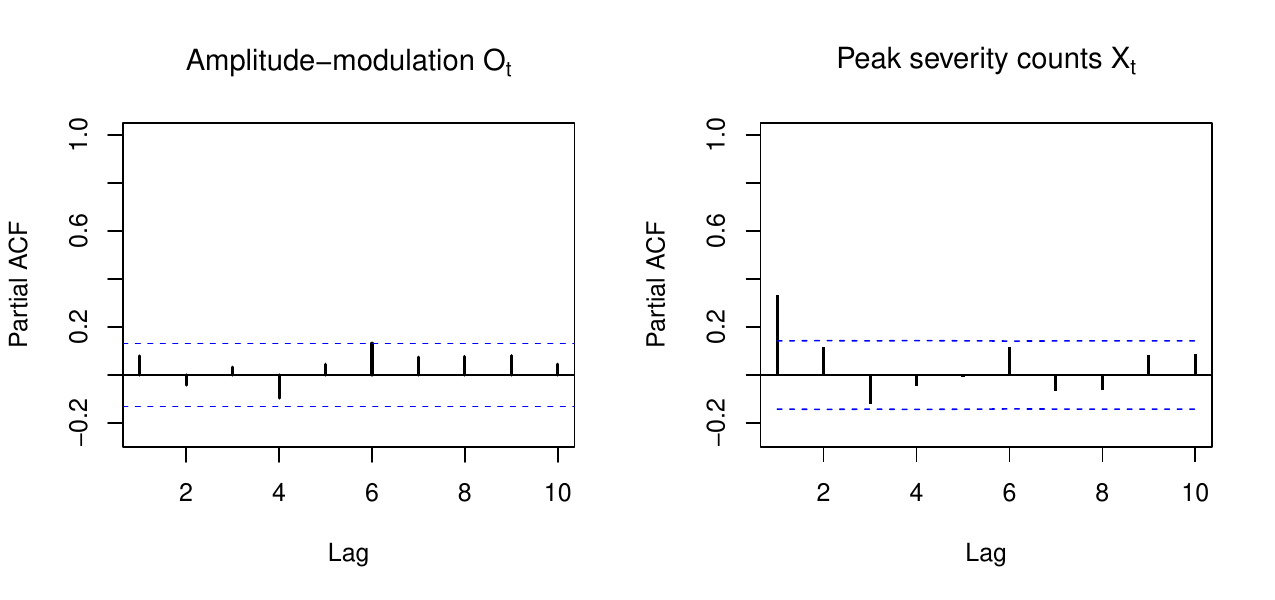} lag
	\caption{Sample ACF/PACF of $ O_t $ and $ X_t $ of the peak severity counts.}
	\label{Migraine_Ot_Xt}
\end{figure}
Next, we draw our attention to the peak severity counts $ X_t $ and calculate the sample PACF with consideration of missingness. Therefore, we first determine the sample ACF of $ X_t $ using the missing data approach of \cite{DunRob81} (for more details, see Appendix \ref{A_AppData}). Then, we calculate the sample PACF from the sample ACF by using the Durbin--Levinson algorithm, as plotted in the right part of Figure \ref{Migraine_Ot_Xt}. Here, we see a decrease to zero after lag $ 1 $, which indicates an AR(1)-like model. Since we are concerned with bounded counts, our null hypothesis is to assume a BAR(1) process with $ \hat{\mu}=0.6117 $ according to \eqref{Kareem} and $ \hat{\rho}=0.3325 $ according to \eqref{rhodunsmuir}. To test this hypothesis on level $ 5\% $, we look at the binomial index of dispersion and  skewness index from Sections \ref{Ch_BinInd}--\ref{Ch_SkewInd}. For the binomial index of dispersion, we get the critical upper value and the test statistic as
\begin{align*}
	1+\frac{\mathbb{B}_{\hat{I}^\bin}}{T}+ z_{0.975}\cdot\sqrt{\frac{\sigma_{\hat{I}^\bin}^2}{T}}=1.1685\quad<\quad\hat{I}^{\bin}=1.3451. 
\end{align*}
Here, $ z_\alpha $ denotes the $ \alpha $-quantile of the standard normal distribution, $ \norm(0,1) $. Obviously, the critical value is violated, so we are not concerned with a BAR(1) process, but are faced with extra-binomial variation. For the skewness index, in the case of a BAR(1) process, we compute the critical value and test statistic as
\begin{align*}
	1-\frac{2}{n}+\frac{\mathbb{B}_{\hat{I}^\bin_\skw}}{T}\pm z_{0.975}\cdot\sqrt{\frac{\sigma_{\hat{I}^\bin_\skw}^2}{T}},\quad\quad\hat{I}^{\skw}=0.3422.
\end{align*}
For the upper and lower critical value we get a $ 0.7337 $ and $ -0.1235 $, respectively. By contrast to the binomial index, the theoretical H$ _0 $-value for the skewness index is equal to $ 1/3 $. The critical values are not violated, so the skewness of the peak severity counts is in agreement with a binomial null. Altogether, the BAR(1) model is not appropriate, but a modification of it with extra-binomial variation would be needed.
\begin{remark}\label{Bem_noMiss}
	Just as an experiment, let us analyze the data again by simply ignoring the missing data. This corresponds to a shorter data length $ T= 206 $ and full observation $ \hat{\uptau}=1 $. Additionally, we can directly determine the (P)ACF by using the R-command \texttt{pacf} and \texttt{acf} together with the option \texttt{na.pass}. Then, our null hypothesis sets a BAR(1) process with $ \hat{\mu}=0.6117$  and $ \hat{\rho}=0.3605 $, the latter being slightly larger than if correctly considering the missing data. For testing this hypothesis, we again look at the binomial index of dispersion, for which we get the upper critical value and the test statistic as
	\begin{align*}
		1+\frac{\mathbb{B}_{\hat{I}^\bin}}{T}+ z_{0.975}\cdot\sqrt{\frac{\sigma_{\hat{I}^\bin}^2}{T}}=1.1728\quad<\quad\hat{I}^{\bin}=1.3451,
	\end{align*}
	Thus, although the critical value is affected by the ignorance of missing data, we again reject the null. For the skewness index, we compute the critical values and test statistic as
	\begin{align*}
		1-\frac{2}{n}+\frac{\mathbb{B}_{\hat{I}^\bin_\skw}}{T}\pm z_{0.975}\cdot\sqrt{\frac{\sigma_{\hat{I}^\bin_\skw}^2}{T}},\quad\quad\hat{I}^{\skw}=0.3422. 
	\end{align*}
	For the upper and lower critical value we get a $ 0.7391 $ and $ -0.1331 $, respectively. Although we end up with the same decisions, we see an effect on the critical values. The effect is not particularly large, because we are only confronted with a low percentage of missing data $ (8.4\%) $. For larger percentage, recall our results from Section \ref{Ch_SimStudy}, a more severe effect would be expected. Therefore, our approach is crucial to avoid incorrect decision making.
\end{remark}
%
\subsection{Cloud coverage counts}
%
\begin{figure}[h]
	\center\small
	\includegraphics[scale=0.55,viewport=1 40 600 240, clip]{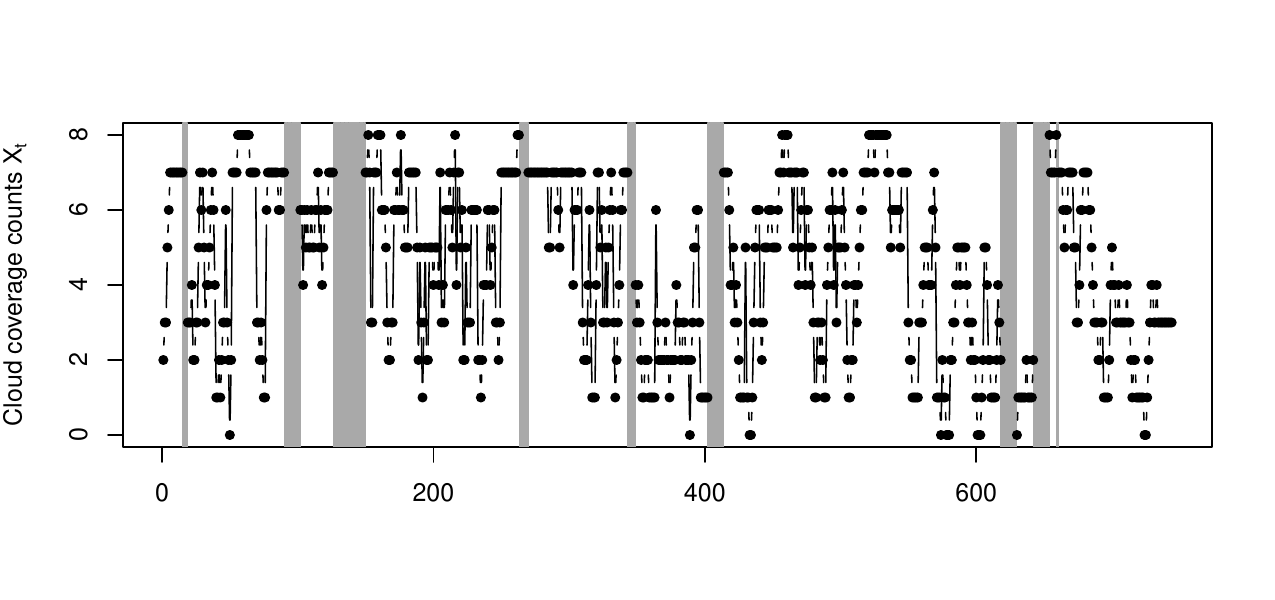}t
	\caption{Plot of hourly cloud coverage counts in Schleswig in August 2016. The vertical lines indicate missing data.}
	\label{Clouddata}
\end{figure}
\noindent``Okta'', or the number of eighths of the sky that are covered by clouds, is the unit of measurement for cloud coverage. Such Okta counts $ X_t $ have a bounded range $ \{0,\ldots,8\} $. The cloud coverage counts are collected hourly by the ``DWD Climate Data Center'' of Deutscher Wetterdienst (German Weather Service). They may contain missing data because of, \eg environmental factors such heavy precipitation, fog, or light pollution, as well as measurement device malfunctions, see \citet[p. 4676.]{Weiss21}. As an illustrative example, we study the weather station Glücksburg-Meierwik situated in Northern Germany in August 2016. Figure \ref{Clouddata} shows the corresponding hourly time series of length $ T=744 $ with $ (11\ \%) $ missing data represented by the gray vertical lines. 

For the cloud coverage counts the data set is bounded with an upper bound of $ n=8 $, we again consider the asymptotics for bounded counts derived in Sections \ref{Ch_BinInd}--\ref{Ch_SkewInd}. 

In this context, we examine the pattern of the amplitude-modulation $ (O_t) $ in greater detail. In the left part of Figure \ref{Cloud_Ot_Xt}, the sample ACF/PACF of the $ O_t $ is presented. We have a significant lag-$ 1 $ value this time, but a rapid decrease to zero after lag $ 1 $, so we can justify to treat $ O_t $ as a binary Markov chain. We estimate the dependence parameter as $ \hat{r}=0.8765 $ (PACF of lag $1$), and as we have $ 11\% $ missing values, we get $ \hat{\uptau}=0.8898 $.
\begin{figure}[h]
	\center\small
	\includegraphics[scale=0.57, viewport=1 40 700 300, clip]{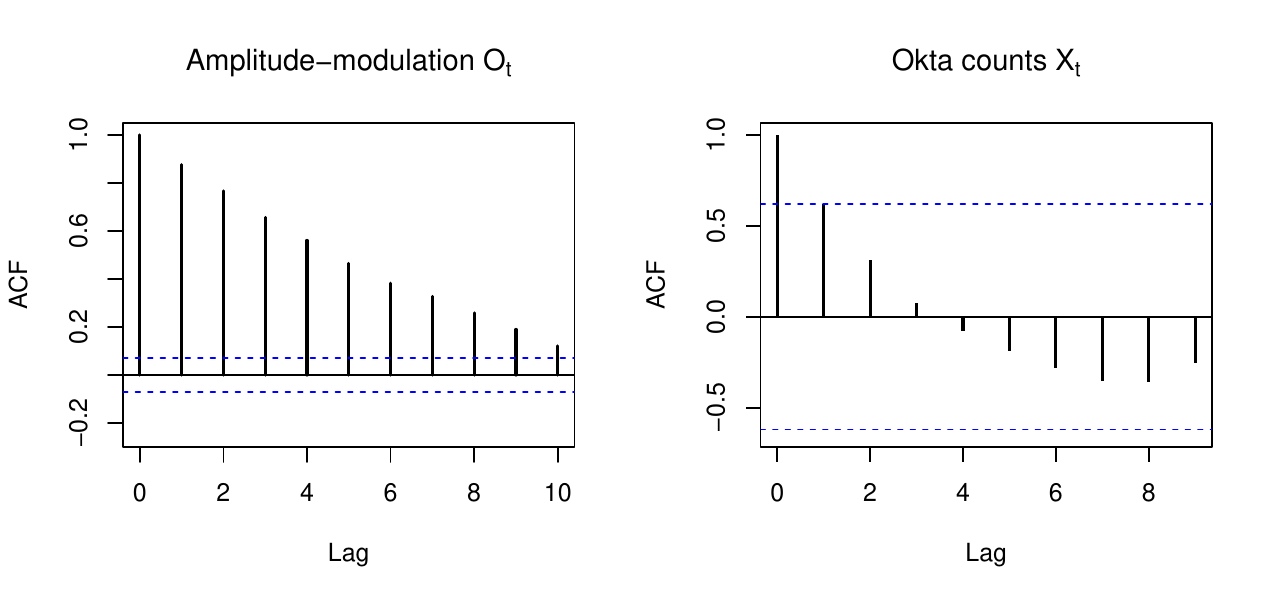}
	\includegraphics[scale=0.57, viewport=1 40 700 300, clip]{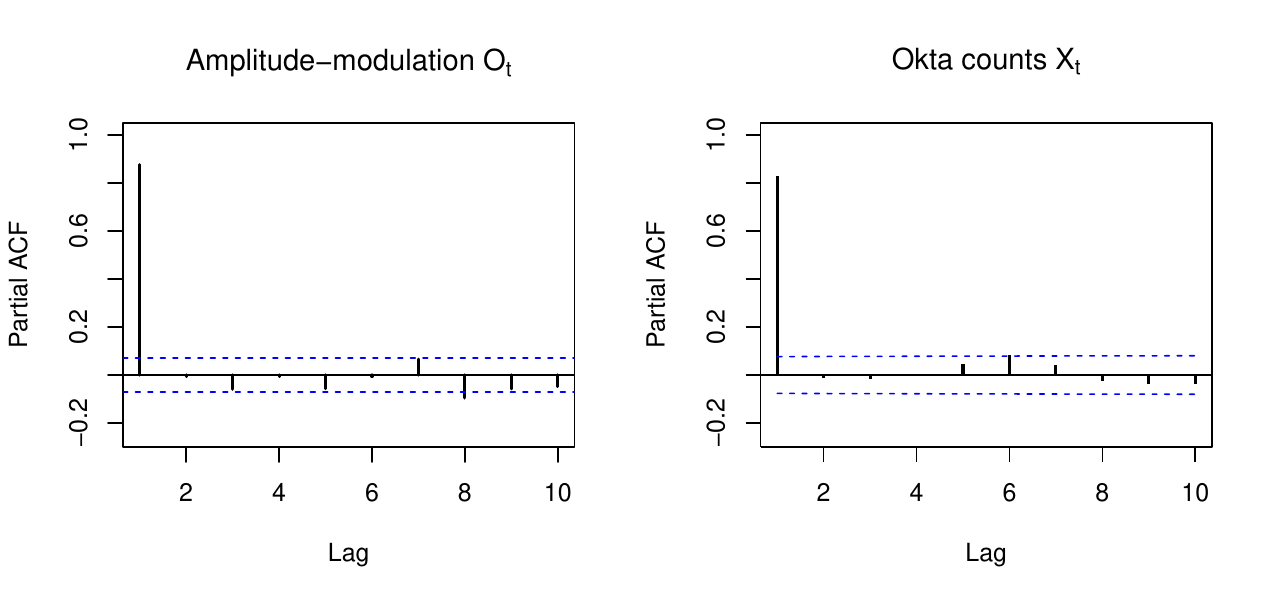} lag
	\caption{Sample ACF/PACF of $ O_t $ and $ X_t $ of the Okta counts.}
	\label{Cloud_Ot_Xt}
\end{figure}
Next, let us analyze the Okta counts $ X_t $. We compute the sample PACF under missingness in the same way as in Section \ref{Ch_PeakCounts}, leading to Figure \ref{Cloud_Ot_Xt}. Here, again, we see a rapid decrease to zero after lag $ 1 $, which indicates an AR(1)-like model. Since we are concerned with bounded counts, our null hypothesis is to assume a BAR(1) process with $ \hat{\mu}=4.4804 $ again according to \eqref{Kareem} and $ \hat{\rho}=0.8285 $ according to \eqref{rhodunsmuir}. To test this hypothesis on the $ 5\% $ level, we look at the binomial index of dispersion and the skewness index from Sections \ref{Ch_BinInd}--\ref{Ch_SkewInd}. For the binomial index of dispersion, the upper critical value is exceeded:
\begin{align*}
	1+\frac{\mathbb{B}_{\hat{I}^\bin}}{T}+ z_{0.975}\cdot\sqrt{\frac{\sigma_{\hat{I}^\bin}^2}{T}}=1.2169\quad<\quad\hat{I}^{\bin}=2.6908. 
\end{align*}
As the critical value is violated, we are  not concerned with a BAR(1) process. For the skewness index, the upper critical value and test statistic are computed as
\begin{align*}
	1-\frac{2}{n}+\frac{\mathbb{B}_{\hat{I}^\bin_\skw}}{T}+ z_{0.975}\cdot\sqrt{\frac{\sigma_{\hat{I}^\bin_\skw}^2}{T}}=0.7875\quad<\quad\hat{I}^{\skw}=0.9788.
\end{align*}
By contrast to the binomial index, the theoretical H$ _0 $-value for the skewness index is equal to $ 0.75 $. As the critical value is violated, we conclued again that we are not concerned with a BAR(1) process. Therefore, another AR(1)-type model with extra dispersion and skewness, like the beta-binomial AR(1) model, would be required for further modeling these data.
\begin{remark}
	Based on our analysis, we have concluded that the beta-binomial AR(1) model may be better suited for modeling cloud coverage counts compared to the BAR(1) model. To assess the performance of both models, we conducted model fitting followed by optimization of the log-likelihood function to obtain maximum likelihood estimates for the parameters of each model. For the BAR(1) model \eqref{BAR1}, we obtained the parameter estimates of $\hat{\pi}_\text{ML}= 0.566$ and $\hat{\rho}_\text{ML}=0.740 $. For the beta-binomial AR(1) model, the corresponding estimates were $\hat{\pi}_\text{ML}^\text{bbin}= 0.570$, $\hat{\rho}_\text{ML}^\text{bbin}=0.744 $, and  the dispersion parameter estimate $\hat{\phi}_\text{ML}^\text{bbin}= 0.209 $ whereas $\phi\to 0$ corresponds to the BAR(1) model.
	
	Subsequently, we computed the Akaike Information Criterion, resulting in a value of $\approx2165.51$ for the BAR(1) model and $\approx1974.67$ for the beta-binomial AR(1) model. This indicates that the beta-binomial AR(1) model provides a better fit for the cloud coverage counts. Further confirmation of this conclusion was obtained by comparing the binomial index of dispersion and skewness index between the fitted models. For the BAR(1) model, we get $I^\bin=1$ and $I_\skw^\bin=0.75$, while for the beta-binomial AR(1) model, we get $I^\bin=1.819$ and $I_\skw^\bin=0.868$. In this context, the indices obtained for the beta-binomial AR(1) model are closer to the  previously determined values of $\hat{I}^{\bin}=2.6908$ and $\hat{I}^{\skw}=0.9788$ for the cloud coverage data. These findings further support the conclusion that the beta-binomial AR(1) model provides a more appropriate representation for the cloud coverage data.
\end{remark}
%
\subsection{Compensation data}
%
\begin{figure}[h]
	\center\small
	\includegraphics[scale=0.55,viewport=1 40 600 240, clip]{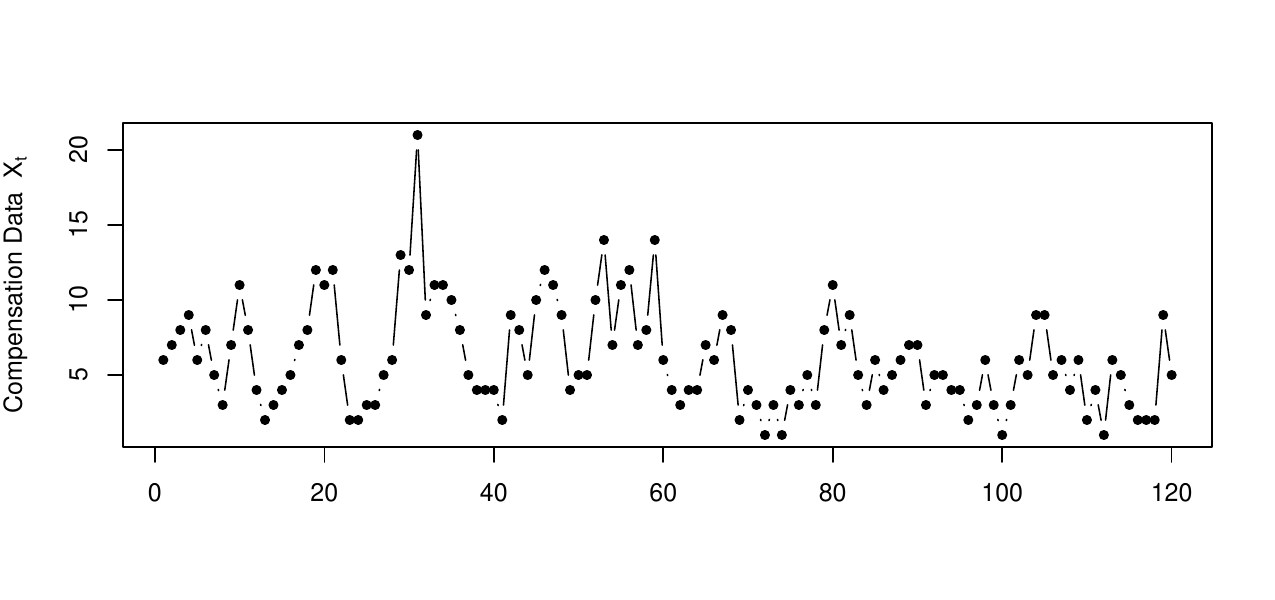}t
	\caption{Plot of monthly counts of workers collecting disability benefits from January 1985 to January 1994.}
	\label{Compdata}
\end{figure}
\noindent The dataset under consideration is from \cite{Free98}. The data originates from the logging industry, involving claimants who receive short-term disability benefits from the Workers' Compensation Board (WCB) of British Columbia. The study cohort exclusively consists of male claimants aged between 34 and 54, employed in the logging industry, having experienced injuries such as cuts, lacerations, and punctures. These claimants reported their claims to the Richmond service delivery location within the time frame spanning January 1985 to January 1994 \citep[p.157]{Free98}. The compensation counts $ X_t $ have the unbounded range  $ \{0,1,2\ldots\} $ this time, so now, the statistics derived in Sections \ref{Ch_PInd} and \ref{Ch_SkewInd} are relevant for model diagnostics. The monthly time series of length $ T=120 $ is illustrated in Figure \ref{Compdata}. \cite{Free98} concludes that a Poi-INAR(1) model is a plausible choice for modeling this dataset. Thus, our null hypothesis is to assume a Poi-INAR(1) process. Since the original data does not include missing observations, we know the true values of the relevant diagnostic statistics. Therefore, the data can be used to investigate the effect of different levels of missingness, namely by randomly excluding 15\%, 30\%, 45\% and 60\% of the observations. This leads, among others, to the modified time series presented in Figure \ref{CompMissData}.
\begin{figure}[h]
	\center\small
	\includegraphics[scale=0.55,viewport=1 40 600 240, clip]{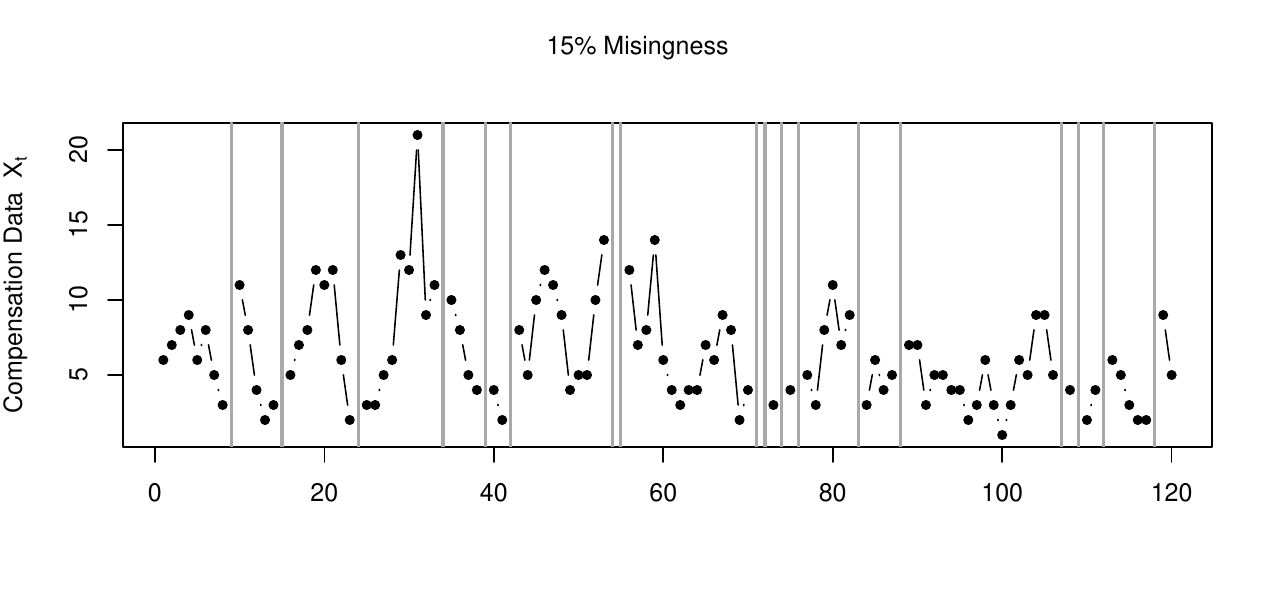}t
	\caption{Plot of monthly counts of workers collecting disability benefits from January 1985 to January 1994. The vertical lines indicate missing data (15\%).}
	\label{CompMissData}
\end{figure}
To test the aforementioned null hypothesis at a significance level of $ 5\%$, we look at  the Poisson index of dispersion and skewness index discussed in Sections \ref{Ch_PInd} and \ref{Ch_SkewInd}. For the Poisson index of dispersion and skewness index, respectively, we get the critical values and the test statistic as follows: 
\begin{align*}
	&1+\frac{\mathbb{B}_{\hat{I}^\poi}}{T}\pm z_{0.975}\cdot\sqrt{\frac{\sigma_{\hat{I}^\poi}^2}{T}},\qquad\hat{I}^{\poi}=\frac{\hat{\mu}_{(2)}}{\hat{\mu}}-\hat{\mu}+1,\\
	&1+\frac{\mathbb{B}_{\hat{I}^\poi_\skw}}{T}\pm z_{0.975}\cdot\sqrt{\frac{\sigma_{\hat{I}^\poi_\skw}^2}{T}},\qquad\hat{I}^{\skw}=\frac{\hat{\mu}_{(3)}}{\hat{\mu}\hat{\mu}_{(2)}}.
\end{align*}
In Table \ref{TabComp}, the estimates for the Poisson index of dispersion and the skewness index, along with their critical values for various levels of missingness are presented. Let us first look at the estimated mean and autocorrelation. As missingness levels increase, the autocorrelation diminishes, while the estimated mean remains relatively close to the true mean, even when 60\% of the observations are missing. Upon considering the Poisson index of dispersion and its associated critical values, we observe that the upper critical value is violated in every case of missingness. Thus we conclude that the Poi-INAR(1) process is not applicable. Interestingly, our decision remains consistent even under the highest missingness scenario. Particularly, at $\uptau=0.4$, we observe a significant deviation in the Poisson index of dispersion $(1.62)$ from the null value. Consistent findings are observed for the skewness index. Consequently, an alternative AR(1)-type model that incorporates extra dispersion and skewness is required. In conclusion, our asymptotic results hold up well in the presence of missing data.
\begin{table}[h]
	\centering
	\caption{Estimates of the compensation data for different levels of missingness including full data ($\uptau=1$).}
	\label{TabComp}
	\smallskip
	\scalebox{1}{
		\begin{tabular}{lrr|rrr|rrrrrl}
			\hline
			&&&&\multicolumn{2}{c}{critical values}&&\multicolumn{2}{c}{critical values}\\
			$\uptau$&$ \hat{\mu} $&$ \hat{\rho} $&$\hat{I}^{\poi}$& lower&upper&$\hat{I}^{\skw}$&lower&upper\\
			\hline
			1     & 6.133& 0.558& 1.907&  0.621& 1.320& 1.328& 0.870& 1.108\\ 
			0.85& 6.343& 0.462& 1.853&  0.644& 1.308& 1.314& 0.881& 1.100\\
			0.70& 6.476& 0.302& 1.836&  0.658& 1.304& 1.306& 0.888& 1.098\\ 
			0.55& 6.591& 0.269& 1.977&  0.623& 1.334& 1.333& 0.879& 1.106\\
			0.40& 6.396& 0.295& 1.620&  0.557& 1.387& 1.164& 0.852& 1.126\\
			\hline
	\end{tabular}}
\end{table}
%
\section{Conclusions}\label{Conclusions}
%
In this article, we considered unbounded and bounded count time series with missing data. We were able to derive general formulas for the asymptotics of relevant statistics, namely types of dispersion and skewness index. We derived closed-form expressions for the special cases of Poisson and binomial autoregressive processes. Thus, the considered indices can be used to test the null of a Poisson or binomial distribution, respectively, in the presence of missing data. The developed approaches were applied to two real-world data examples, and simulations were used to analyze the finite-sample performance of the resulting asymptotic approximations. It became evident through our simulation study that our asymptotic derivations are well-suited for approximating the true finite-sample distributions.  Therefore, simply omitting the missing observations may lead to misleading and false conclusions. As a result, when examining count time series, missing observations should be carefully taken into account. 

A task for further research is to investigate the derivation of our asymptotics under different types of missingness. Additionally, expanding the scope of the derived asymptotics to higher-order autoregressive processes such as the Poi-INAR(2) process, or to a different class of processes, such as the integer-valued moving average (INMA) process could be another research direction. The parameter estimation in the presence of missing observations, also recall the existing approach by \cite{AndKarl10}, might be another field of research, which could be investigated in the future.


	\clearpage
	
	\appendix
	\small
	\numberwithin{theorem}{section}
	
%
\section{Derivations, Proofs \& auxiliary results}\label{A_AuxRes}
%
%
\subsection{Central Limit Theorem}\label{A_CLT}
%
Let us start by deriving Theorem \ref{CLT_Fac}. For $ (\X_t^\ast) $, we have the CLT
\begin{equation}\label{Rodman}
	\sqrt{T}\Big(\overline{\X^\ast}-\bmu^\ast\Big) \ \xrightarrow{\text{d}}\ \norm\Big(\bold{0}, \bold{\Sigma}^\ast\Big)\quad \text{with}\quad \bold{\Sigma}^\ast=(\sigma_{ij}^\ast)_{i,j=0,\ldots,m},
\end{equation}
where
\begin{equation}
	\label{Stockton}
	\sigma_{ij}^\ast=CoV[X_{0,i}^\ast,X_{0,j}^\ast]+\sum_{h=1}^\infty \Big(CoV[X_{0,i}^\ast,X_{h,j}^\ast]+CoV[X_{h,i}^\ast,X_{0,j}^\ast]\Big).
\end{equation}
These covariances compute as 
\begin{equation}
	\label{Jazz}
	\sigma_{ij}^\ast=
	\begin{cases}
		\uptau(1-\uptau)+2\sum_{h=1}^\infty \gamma_O(h) & \text{if }i=j=0,\\
		\sigma_{00}^\ast\mu_{(j)}         &  \text{if }i=0,\ j>0,\\
		\uptau(\mu_{(i,j)}(0)-\mu_{(i)}\mu_{(j)})+\sigma_{00}^\ast\mu_{(i)}\mu_{(j)}& \text{if }i,j>0.\\
		\quad+\sum_{h=1}^\infty \uptau(h)\Big(\mu_{(j,i)}(h)+\mu_{(i,j)}(h)-2\mu_{(i)}\mu_{(j)}\Big) \\
	\end{cases}
\end{equation}
Here, $ \overline{\X^\ast}=(\overline{O},\overline{O\,\X}^\top)^\top $ are the required components for the calculation of $ \hat{\bmu} $.
\begin{proof} The covariances in \eqref{Stockton} satisfy  
	\[
	CoV[X_{t,i}^\ast,X_{s,j}^\ast]=
	\begin{cases}
		CoV[O_t,O_s]=\gamma_O(\vert t-s \vert), & \text{if} \ i=j=0,\\
		CoV[O_t,O_s(X_s)_{(j)}], & \text{if} \ i=0, \ j> 0,\\
		CoV[O_t(X_t)_{(i)},O_s(X_s)_{(j)}], & \text{if} \ i,j>0.
	\end{cases}
	\]
	First, let us acknowledge that for $ i=j=0 $, we can conclude $ \sigma_{00} $. Thus, we get 
	\begin{align*}
		\sigma_{00}^\ast&=CoV[X_{0,0}^\ast,X_{0,0}^\ast]+2\sum_{h=1}^\infty CoV[X_{h,0}^\ast,X_{0,0}^\ast]\\
		&=\gamma_O(0)+2\sum_{h=1}^\infty \gamma_O(h)\\
		&=\uptau(1-\uptau)+2\sum_{h=1}^\infty \gamma_O(h).
	\end{align*}
	Since $ (O_t) $ is independent of $ (X_t) $, we have
	\begin{align*}
		CoV[O_t,O_s(X_s)_{(j)}]&=\e[O_tO_s(X_s)_{(j)}]-\e[O_t]\e[O_s(X_s)_{(j)}]\\
		&= \e[O_tO_s]\e[(X_s)_{(j)}]-\e[O_t]\e[O_s]\e[(X_s)_{(j)}]\\
		&= CoV[O_t,O_s]\e[(X_s)_{(j)}]  \\
		&= \gamma_O(\vert t-s \vert)\mu_{(j)}.
	\end{align*}
	This provides $ \sigma_{j0}^\ast=\sigma_{0j}^\ast=\sigma_{00}^\ast\cdot\mu_{(j)} $ for $ j>0 $. Next, we obtain
	\begin{align*}
		CoV[O_t(X_t)_{(i)},O_s(X_s)_{(j)}]&=\e[O_t(X_t)_{(i)}O_s(X_s)_{(j)}]-\e[O_t(X_t)_{(i)}]\e[O_s(X_s)_{(j)}]\\
		&= \e[O_tO_s]\e[(X_t)_{(i)}(X_s)_{(j)}]-\e[O_t]\e[O_s]\e[(X_t)_{(i)}]\e[(X_s)_{(j)}]\\
		&= \uptau(\vert t-s \vert)\mu_{(i,j)}(t-s)-\uptau^2\mu_{(i)}\mu_{(j)},
	\end{align*}
	for $ i,j>0 $. Note that this expression can also be expressed as
	\begin{align*}
		CoV[O_t(X_t)_{(i)},O_s(X_s)_{(j)}]&=
		\uptau(\vert t-s\vert)CoV[(X_t)_{(i)},(X_s)_{(j)}]+ \gamma_O(\vert t-s \vert)\mu_{(i)}\mu_{(j)}.
	\end{align*}
	Finally, we get 
	\begin{align*}
		\sigma_{ij}^\ast&=CoV[O_0(X_0)_{(i)},O_0(X_0)_{(j)}]+\sum_{h=1}^\infty\Big(CoV[O_0(X_0)_{(i)},O_h(X_h)_{(j)}]
		\\[1ex]&\quad+CoV[O_h(X_h)_{(i)},O_0(X_0)_{(j)}]\Big)\\
		&=\uptau \mu_{(i,j)}(0)-\uptau^2\mu_{(i)}\mu_{(j)}+\sum_{h=1}^\infty\Big(
		\uptau(h)\mu_{(j,i)}(h) + \uptau(h)\mu_{(i,j)}(h)-2\uptau^2\mu_{(i)}\mu_{(j)}	\Big).
	\end{align*}
	Using that $\uptau(h)=\uptau^2+\gamma_O(h) $, we get
	\begin{align*}
		\sigma_{ij}^\ast&=\uptau \mu_{(i,j)}(0)-(\uptau-\gamma_O(0))\mu_{(i)}\mu_{(j)}+2\mu_{(i)}\mu_{(j)}\sum_{h=1}^\infty\gamma_O(h)
		\\[1ex]&\quad+\sum_{h=1}^\infty \uptau(h)\Big(\mu_{(j,i)}(h)+\mu_{(i,j)}(h)-2\mu_{(i)}\mu_{(j)}\ \Big)\\
		&=\uptau(\mu_{(i,j)}(0)-\mu_{(i)}\mu_{(j)})+\sigma_{00}^\ast\mu_{(i)}\mu_{(j)}+\sum_{h=1}^\infty \uptau(h)\Big(\mu_{(j,i)}(h)+\mu_{(i,j)}(h)-2\mu_{(i)}\mu_{(j)}\Big),
	\end{align*}
	where we utilized the calculations of $ \sigma_{00}^\ast $ in the last step. 
\end{proof}
Now, let us define the function $ \f:[0,1]\times[0,\infty)^m\rightarrow[0,\infty)^m $ by $ f_j (x_0,x_1,\ldots,x_m)=x_j/x_0$ for $ j=1,\ldots,m $. Then, $ \hat{\bmu}=\f(\overline{\X^\ast}) $, $ \bmu=\f(\bmu^\ast) $, and the gradient of $ f_j(\x) $ is 
\begin{align*}
	\nabla f_j(\x) =(-\tfrac{x_j}{x_0^2},0,\ldots,0,\tfrac{1}{x_0},0,\ldots,0),\quad \text{with }\tfrac{1}{x_0} \text{ at position }j.
\end{align*}
Hence, the Jacobian of $ \f $ evaluated in $ \bmu^\ast $, equals
\begin{align*}
	\D= \frac{1}{\uptau}
	\begin{pmatrix} 
		-\mu & 1 & 0 & \cdots & 0\\
		-\mu_{(2)} & 0 & 1 & \ddots  & \vdots\\
		\vdots & \vdots & \ddots & \ddots & 0\\
		-\mu_{(m)} & 0 & \cdots & 0  & 1\\
	\end{pmatrix}.
\end{align*}
So the linear Taylor approximation $ \hat{\bmu}\approx \bmu + \D(\overline{\X^\ast}-\bmu^\ast)$ together with the CLT implies 
\begin{align*}
	\sqrt{T}\Big(\hat{\bmu}-\bmu\Big) \ \xrightarrow{\text{d}}\ \norm\Big(\bold{0}, \bold{\Sigma}\Big)\quad \text{with}\quad \bold{\Sigma}=\D\bold{\Sigma}^\ast\D^\top.
\end{align*}
We can compute the entries $ (\sigma_{ij})_{i,j=1,\ldots,m} $ as 
\begin{align*}
	\sigma_{ij}&=\sum_{k,l=0}^{m}d_{ik}d_{jl}\sigma_{kl}^\ast
	= d_{i0}d_{j0}\sigma_{00}^\ast + d_{i0}d_{jj}\sigma_{0j}^\ast + d_{ii}d_{j0}\sigma_{i0}^\ast + d_{ii}d_{jj}\sigma_{ij}^\ast\\
	&=\tfrac{1}{\uptau^2}\Big( \mu_{(i)}\mu_{(j)}\sigma_{00}^\ast - \mu_{(i)}\sigma_{0j}^\ast - \mu_{(j)}\sigma_{i0}^\ast + \sigma_{ij}^\ast\Big)=\tfrac{1}{\uptau^2}\Big(\sigma_{ij}^\ast- \mu_{(i)}\mu_{(j)}\sigma_{00}^\ast \Big)\\
	&=\tfrac{1}{\uptau}\big(\mu_{(i,j)}(0)-\mu_{(i)}\mu_{(j)}\big)+\tfrac{1}{\uptau^2}\sum_{h=1}^\infty \uptau(h)\Big(\mu_{(j,i)}(h)+\mu_{(i,j)}(h)-2\mu_{(i)}\mu_{(j)}\Big).
\end{align*}
In the last step, we utilized that $ \sigma_{0j}^\ast=\sigma_{00}^\ast\cdot \mu_{(j)} $.
To approximate the bias of $ \hat{\bmu} $, we use the quadratic Taylor approximation $ \hat{\mu}_j\approx \mu_j + \D^{(j)}(\overline{\X^\ast}-\bmu^\ast)+\tfrac{1}{2}(\overline{\X^\ast}-\bmu^\ast)^\top\Hes^{(j)}(\overline{\X^\ast}-\bmu^\ast)$, where $ \D^{(j)} $ denotes the $ j $th row of $ \D $, and where $ \Hes^{(j)} $ is the Hessian of $ f_j $ evaluated in $ \bmu^\ast $. For the non-zero second-order partial derivatives, one gets 
\begin{align}
	\frac{\partial^2}{\partial x_0^2}f_j=\frac{2x_j}{x_0^3},\qquad
	\frac{\partial^2}{\partial x_0\partial x_j}f_j=-\frac{1}{x_0^2},\qquad \frac{\partial^2}{\partial x_j^2}f_j=0.
\end{align}
Evaluating these derivatives in $ \bmu^\ast $ yields the non-zero entries 
\begin{align}
	h_{00}^{(j)}=\frac{2\mu_{(j)}}{\uptau^2},\qquad
	h_{0j}^{(j)}=-\frac{1}{\uptau^2},
\end{align}
of the Hessian $ \Hes^{(j)} $. This implies the approximation 
\begin{align}
	T\ \e[\hat{\mu}_j-\mu_j]\approx \tfrac{1}{2}h_{00}^{(j)}\sigma_{00}  + h_{0j}^{(j)}\sigma_{0j}= \tfrac{1}{\uptau^2}\big(\mu_j\sigma_{00}  - \sigma_{0j} \big)=0,
\end{align}
where we used that $ \sigma_{0j}=\sigma_{00}\cdot\mu_{(j)} $. As a result, the bias is negligible for practice.	
%
\subsection{Considered models}\label{A_Model}
%
We shall begin this section by presenting an essential lemma for factorial moments. Then, in Sections  \ref{Ch_PInd}--\ref{Ch_SkewInd}, we shall derive crucial results  essential to determine our considered statistics.
\begin{lemma}\label{Lebron}
	For the mixed factorial moments of lag zero, it holds that 
	
	\vspace{0.3cm}
	\begin{tabular}{llll}
		$(i)$& $ \mu_{(1)} = \mu$ & $(v)$ &$  \mu_{(1,3)}(0) = \mu_{(4)}+3\mu_{(3)},  $\\
		$(ii)$ &$ \mu_{(1,1)}(0) = \mu_{(2)}+\mu, $ & $(vi)$ &$  \mu_{(2,3)}(0) = \mu_{(5)}+6\mu_{(4)}+6\mu_{(3)}$,\\
		$(iii)$& $ \mu_{(1,2)}(0) = \mu_{(3)}+2\mu_{(2)}, $&$(vii)$ &$  \mu_{(3,3)}(0) = \mu_{(6)}+9\mu_{(5)}+18\mu_{(4)}+6\mu_{(3)}$.\\
		$(iv)$& $ \mu_{(2,2)}(0) = \mu_{(4)}+4\mu_{(3)}+2\mu_{(2)}, $
	\end{tabular} 
\end{lemma}
\begin{proof}
	Statement $ (i) $ is trivial. 
	\begin{align*}
		(ii)\quad\mu_{(1,1)}(0)&=\e[X_t^2]=\e[X_t(X_t-1+1)]=\e[X_t(X_t-1)+X_t]\\
		&=\e[X_t(X_t-1)]+\e[X_t]=\mu_{(2)}+\mu.\\[0.2cm]	
		(iii)\quad\mu_{(1,2)}(0)&=\e[(X_t)(X_t)_{(2)}]=\e[X_t^2(X_t-1)]=\e[X_t(X_t-1)(X_t-2+2)]\\
		&=\e[X_t(X_t-1)(X_t-2)+2X_t(X_t-1)]\\
		&=\e[X_t(X_t-1)(X_t-2)]+2\e[X_t(X_t-1)]=\mu_{(3)}+2\mu_{(2)}.\\[0.2cm]	
		(iv)\quad\mu_{(2,2)}(0)&=\e[(X_t)_{(2)}(X_t)_{(2)}]=\e[X_t^2(X_t-1)^2]\\
		&=\e[X_t(X_t-1)(X_t-2+2)(X_t-3+2)]\\
		&=\e[X_t(X_t-1)(X_t-2)(X_t-3)+2X_t(X_t-1)(X_t-2-1)\\
		&\quad +2X_t(X_t-1)(X_t-2)+4X_t(X_t-1)]\\
		&=\e[X_t(X_t-1)(X_t-2)(X_t-3)]+4\e[X_t(X_t-1)(X_t-2)]\\
		&\quad+2\e[X_t(X_t-1)]\\
		&=\mu_{(4)}+4\mu_{(3)}+2\mu_{(2)}\\
		(v)\quad\mu_{(1,3)}(0)&=\e[X_t(X_t)_{(3)}]=\e[X_t(X_t-1)(X_t-2)(X_t-3+3)]\\
		&=\e[X_t(X_t-1)(X_t-2)(X_t-3)]+3\e[X_t(X_t-1)(X_t-2)]\\
		&=\mu_{(4)}+3\mu_{(3)}.\\[0.2cm]
		(vi)\quad\mu_{(2,3)}(0)&=\e[(X_t)_{(2)}(X_t)_{(3)}]=\e[X_t(X_t-1)(X_t-2)(X_t-3+3)(X_t-4+3)]\\
		&=\e[X_t(X_t-1)(X_t-2)(X_t-3)(X_t-4)]+6\e[X_t(X_t-1)(X_t-2)]\\
		&\quad+6\e[X_t(X_t-1)(X_t-2)(X_t-3)]\\
		&=\mu_{(5)}+6\mu_{(4)}+6\mu_{(3)}.\\[0.2cm]		  
		(vii)\quad\mu_{(3,3)}(0)&=\e[(X_t)_{(3)}(X_t)_{(3)}]\\
		&=\e[X_t(X_t-1)(X_t-2)(X_t-3+3)(X_t-4+3)(X_t-5+3)]\\
		&=\e[X_t(X_t-1)(X_t-2)(X_t-3)(X_t-4)(X_t-5)]\\
		&\quad+9\e[X_t(X_t-1)(X_t-2)(X_t-3)(X_t-4)]\\
		&\quad+18\e[X_t(X_t-1)(X_t-2)(X_t-3)]+6\e[X_t(X_t-1)(X_t-2)]\\
		&=\mu_{(6)}+9\mu_{(5)}+18\mu_{(4)}+6\mu_{(3)}.
	\end{align*}
\end{proof}
%
\subsubsection{Poi-INAR(1) model}\label{A_PoiModel}
%
The Poi-INAR(1) model shall be used for the Poisson index of dispersion as well as for the skewness index, if  one is concerned with unbounded counts. Thus, let us consider the following lemmata. 
\begin{lemma}\label{AnsuFati}
	By Proposition \ref{Kobe}, we have
	
	\vspace{0.3cm} 
	\begin{tabular}{llll}
		$(i)$ &$ \mu_{(1,1)}(h) = \mu^2+\mu\rho^h, $ &$(iv)$ &$  \mu_{(1,3)}(h) = \mu^4+3\mu^3\rho^h,  $\\
		$(ii)$& $ \mu_{(1,2)}(h) = \mu^3+2\mu^2\rho^h, $&$(v)$ &$  \mu_{(2,3)}(h) = \mu^5+6\mu^4\rho^h+6\mu^3\rho^{2h}$,\\
		$(iii)$& $ \mu_{(2,2)}(h) = \mu^4+4\mu^3\rho^h+2\mu^2\rho^{2h}, $&$(vi)$ &$  \mu_{(3,3)}(h) = \mu^6+9\mu^5\rho^h+18\mu^4\rho^{2h}+6\mu^3\rho^{3h}$.
	\end{tabular} 
\end{lemma}
\vspace{0.3cm} 
Using Lemma \ref{Lebron}, Lemma \ref{AnsuFati} and applying them to Corollary \ref{Ilaxis}
$ (i) $, yields the following. 
\begin{lemma}\label{Pique}
	Let us assume that the missing data follow a Markov model, \ie $ \uptau(h)=\uptau^2+\uptau(1-\uptau)r^h $. Then, we have
	
	\vspace{0.3cm}
	\begin{tabular}{ll}
		$(i)$ &$\sigma_{11}=\mu\bigg( \frac{1}{\uptau}\frac{1+r\rho}{1-r\rho}+\frac{2(1-r)\rho}{(1-r\rho)(1-\rho)} \bigg) $,\\
		$(ii)$ &$ \sigma_{12}=2\mu\cdot\sigma_{11}$,\\
		$(iii)$ &$\sigma_{13}=3\mu^2\cdot\sigma_{11} $,\\
		$(iv)$ &$\sigma_{22}=4\mu^2\cdot\sigma_{11}+2\mu^2\bigg( \frac{1}{\uptau}\frac{1+r\rho^2}{1-r\rho^2}+\frac{2(1-r)\rho^2}{(1-r\rho^2)(1-\rho^2)} \bigg) $,\\
		$(v)$ &$\sigma_{23}=3\mu\cdot\sigma_{22}-6\mu^3\cdot\sigma_{11} $,\\
		$(vi)$ &$\sigma_{33}=9\mu^2\cdot\sigma_{22}-27\mu^4\cdot\sigma_{11}+6\mu^3\bigg( \frac{1}{\uptau}\frac{1+r\rho^3}{1-r\rho^3}+\frac{2(1-r)\rho^3}{(1-r\rho^3)(1-\rho^3)} \bigg) $.
	\end{tabular} 
\end{lemma}
\begin{proof}
	Let us start by proving $ (i) $. According to Corollary \ref{Ilaxis} 
	$ (i) $ we have
	\begin{align*} 
		\sigma_{11}&=\tfrac{1}{\uptau}\Big[\mu_{(1,1)}(0)-\mu^2+\tfrac{2}{\uptau}\sum_{h=1}^\infty \uptau(h)\Big(\mu_{(1,1)}(h)-\mu^2\Big) \Big]\\
		&=\tfrac{1}{\uptau}\mu\Big[1+\tfrac{2}{\uptau}\sum_{h=1}^\infty \uptau(h)\rho^h \Big].
	\end{align*}
	In the last step, we used Lemma \ref{Lebron} and Lemma \ref{AnsuFati}. Now, we assume that the missing data follow the Markov model  $ \uptau(h)=\uptau^2+\uptau(1-\uptau)r^h $. Then, we get 
	\begin{align}\label{Porzingis}
		\sigma_{11}&=\mu\tfrac{1}{\uptau} \Bigg[1+2\uptau\sum_{h=1}^\infty\rho^{h}+2(1-\uptau)\sum_{h=1}^\infty(r\rho)^h \Bigg]\nonumber\\
		&=\mu\tfrac{1}{\uptau} \Bigg(\frac{1+(2\uptau-1)\rho}{1-\rho}-\frac{2r(\uptau-1)\rho}{1-r\rho} \Bigg)\nonumber\\
		&=\mu\Bigg( \frac{1}{\uptau}\frac{1+r\rho}{1-r\rho}+\frac{2(1-r)\rho}{(1-r\rho)(1-\rho)} \Bigg). 
	\end{align}
	In the same manner, one can show $ (ii) $ and $ (iii) $. Next, let us determine $ \sigma_{22} $. Analogously to $ (i) $, we get 
	\begin{align*} 
		\sigma_{22}&=\tfrac{1}{\uptau}\Big[\mu_{(2,2)}(0)-\mu^4+\tfrac{2}{\uptau}\sum_{h=1}^\infty 	\uptau(h)\Big(\mu_{(2,2)}(h)-\mu^4\Big) \Big]\\
		&=\tfrac{4}{\uptau}\mu^3\Big(1+\tfrac{2}{\uptau}\sum_{h=1}^\infty \uptau(h)\rho^{h} 	\Big)+\tfrac{2}{\uptau}\mu^2\Big(1+\tfrac{2}{\uptau}\sum_{h=1}^\infty \uptau(h)\rho^{2h} \Big).
	\end{align*}
	Using the Markov model $ \uptau(h)=\uptau^2+\uptau(1-\uptau)r^h $, and from \eqref{Porzingis}, we get 
	\begin{align*}
		\sigma_{22}&=4\mu^2\cdot\sigma_{11}+2\mu^2\tfrac{1}{\uptau} \Bigg[1+2\uptau\sum_{h=1}^\infty\rho^{2h}+2(1-\uptau)\sum_{h=1}^\infty(r\rho^2)^h \Bigg]\nonumber\\
		&=4\mu^2\cdot\sigma_{11}+2\mu^2\tfrac{1}{\uptau} \Bigg(\frac{1+(2\uptau-1)\rho^2}{1-\rho^2}-\frac{2r(\uptau-1)\rho^2}{1-r\rho^2} \Bigg)\nonumber\\
		&=4\mu^2\cdot\sigma_{11}+2\mu^2\Bigg( \frac{1}{\uptau}\frac{1+r\rho^2}{1-r\rho^2}+\frac{2(1-r)\rho^2}{(1-r\rho^2)(1-\rho^2)} \Bigg). 
	\end{align*}
	For $(v)$, we get for the Markov model $ \uptau(h)=\uptau^2+\uptau(1-\uptau)r^h $ the following:
	\begin{align*} 
		\sigma_{23}&=\tfrac{1}{\uptau}\Big[\mu_{(2,3)}(0)-\mu^5+\tfrac{2}{\uptau}\sum_{h=1}^\infty 		\uptau(h)\Big(\mu_{(2,3)}(h)-\mu^5\Big) \Big]\\
		&=\tfrac{6}{\uptau}\mu^4\Big(1+\tfrac{2}{\uptau}\sum_{h=1}^\infty \uptau(h)\rho^{h} 		\Big)+\tfrac{6}{\uptau}\mu^3\Big(1+\tfrac{2}{\uptau}\sum_{h=1}^\infty \uptau(h)\rho^{2h} \Big)\\
		&=3\mu\cdot\sigma_{22}-6\mu^3\cdot\sigma_{11}.
	\end{align*}
	Finally, for $ (vi) $, we get 
	\begin{align*} 
		\sigma_{33}&=\tfrac{1}{\uptau}\Big[\mu_{(3,3)}(0)-\mu^6+\tfrac{2}{\uptau}\sum_{h=1}^\infty\uptau(h)\Big(\mu_{(3,3)}(h)-\mu^6\Big) \Big]\\
		&=\tfrac{9}{\uptau}\mu^5\Big(1+\tfrac{2}{\uptau}\sum_{h=1}^\infty \uptau(h)\rho^{h} 		\Big)+\tfrac{18}{\uptau}\mu^4\Big(1+\tfrac{2}{\uptau}\sum_{h=1}^\infty \uptau(h)\rho^{2h} \Big)+\tfrac{6}{\uptau}\mu^3\Big(1+\tfrac{2}{\uptau}\sum_{h=1}^\infty \uptau(h)\rho^{3h} 	\Big).
	\end{align*}	
	Again, the Markov model $ \uptau(h)=\uptau^2+\uptau(1-\uptau)r^h $ yields
	\begin{align*}
		\sigma_{33}&=9\mu^2\cdot\sigma_{22}-27\mu^4\cdot\sigma_{11}+6\mu^3\tfrac{1}{\uptau} \Bigg[1+2\uptau\sum_{h=1}^\infty\rho^{3h}+2(1-\uptau)\sum_{h=1}^\infty(r\rho^3)^h \Bigg]\nonumber\\
		&=9\mu^2\cdot\sigma_{22}-27\mu^4\cdot\sigma_{11}+6\mu^3\tfrac{1}{\uptau} \Bigg(\frac{1+(2\uptau-1)\rho^3}{1-\rho^3}-\frac{2r(\uptau-1)\rho^3}{1-r\rho^3} \Bigg)\nonumber\\
		&=9\mu^2\cdot\sigma_{22}-27\mu^4\cdot\sigma_{11}+6\mu^3\Bigg( \frac{1}{\uptau}\frac{1+r\rho^3}{1-r\rho^3}+\frac{2(1-r)\rho^3}{(1-r\rho^3)(1-\rho^3)} \Bigg).
	\end{align*}
\end{proof}
%
\subsubsection{BAR(1) model}\label{A_BinModel}
%
The BAR(1) model shall be used for the binomial index of dispersion as well as for the skewness index, if one is concerned with bounded counts. Thus, let us consider the following lemmata. 
\begin{lemma}\label{Neymar}
	By Proposition \ref{Hulk}, we have
	
	\vspace{0.3cm} 
	\begin{tabular}{ll}
		$(i)$&$ \mu_{(1,1)}(h) = n^2\pi^2+n\pi(1-\pi)\rho^h$,\\
		$(ii)$& $ \mu_{(1,2)}(h) = nn_{(2)}\pi^3+2n_{(2)}\pi^2(1-\pi)\rho^h$,\\
		$(iii)$& $ \mu_{(2,2)}(h) = \big(n_{(2)}\big)^2\pi^4+4(n-1)n_{(2)}(1-\pi)\pi^3\rho^{h}+2n_{(2)}(1-\pi)^2\pi^2\rho^{2h}$,\\
		$(iv)$&$  \mu_{(1,3)}(h) =nn_{(3)}\pi^4+3n_{(3)}(1-\pi)\pi^3 \rho^h $,\\
		$(v)$&$  \mu_{(2,3)}(h) =n_{(2)}n_{(3)}\pi^5+6n_{(3)}(n-1)(1-\pi)\pi^4 \rho^h +6n_{(3)}(1-\pi)^2\pi^3 \rho^{2h} $,\\
		$(vi)$&$  \mu_{(3,3)}(h) =(n_{(3)})^2\pi^6+9(n-1)(n-2)n_{(3)}(n-1)(1-\pi)\pi^5\rho^h$\\ 
		&$\hspace{1.8cm}+18(n-2)n_{(3)}(1-\pi)^2\pi^4 \rho^{2h}+6n_{(3)}(1-\pi)^3\pi^3 \rho^{3h} $.
	\end{tabular} 
\end{lemma}
\vspace{0.3cm} 
Using Lemma \ref{Lebron}, Lemma \ref{Neymar}, and applying them to Corollary \ref{Ilaxis} 
$ (ii) $, yields the following.

\begin{lemma}\label{Pedri}
	Let us assume again that the missing data follow a Markov model, \ie $ \uptau(h)=\uptau^2+\uptau(1-\uptau)r^h $. Then, we have
	
	\vspace{0.3cm}
	\begin{tabular}{ll}
		$(i)$ &$ \sigma_{11}=n\pi(1-\pi)\bigg( \frac{1}{\uptau}\frac{1+r\rho}{1-r\rho}+\frac{2(1-r)\rho}{(1-r\rho)(1-\rho)} \bigg) $,\\
		$(ii)$ &$\sigma_{12}=2(n-1)\pi\cdot\sigma_{11}$, \\
		$(iii)$ &$ \sigma_{13}=3(n-1)(n-2)\pi^2\cdot\sigma_{11}$,  \\
		$(iv)$& $ \sigma_{22}=4(n-1)^2\pi^2\cdot\sigma_{11}+2n_{(2)}(1-\pi)^2\pi^2\bigg( \frac{1}{\uptau}\frac{1+r\rho^2}{1-r\rho^2}+\frac{2(1-r)\rho^2}{(1-r\rho^2)(1-\rho^2)} \bigg), $\\
		$(v)$&$\sigma_{23}=3(n-2)\pi\cdot\sigma_{22}-6(n-1)^2(n-2)\pi^3\cdot\sigma_{11}, $\\
		$(vi)$& $ \sigma_{33}=9(n-2)^2\pi^2\cdot \sigma_{22} -27(n-1)^2(n-2)^2\pi^4\cdot\sigma_{11}$\\
		&$\hspace{1cm}+6n_{(3)}(1-\pi)^3\pi^3\bigg( \frac{1}{\uptau}\frac{1+r\rho^3}{1-r\rho^3}+\frac{2(1-r)\rho^3}{(1-r\rho^3)(1-\rho^3)} \bigg)$.\\	
	\end{tabular} 
\end{lemma}
\begin{proof}
	Let us start by proving $ (i) $. According to Corollary \ref{Ilaxis} $ (ii) $, we have
	\begin{align*} 
		\sigma_{11}&=\tfrac{1}{\uptau}\Big[\mu_{(1,1)}(0)-n^2\pi^2+\tfrac{2}{\uptau}\sum_{h=1}^\infty \uptau(h)\Big(\mu_{(1,1)}(h)-n^2\pi^2\Big) \Big]\\
		&=\tfrac{n\pi(1-\pi)}{\uptau}\Big[1+\tfrac{2}{\uptau}\sum_{h=1}^\infty \uptau(h)\rho^h \Big].
	\end{align*}
	In the last step, we used Lemma \ref{Lebron} and Lemma \ref{Neymar}. Now, we assume that the Missing Data follows the Markov model  $ \uptau(h)=\uptau^2+\uptau(1-\uptau)r^h $. Then, like in the Proof of Lemma \ref{Pique} $ (i) $, we get
	\begin{align*} 
		\sigma_{11}=n\pi(1-\pi)\bigg( \frac{1}{\uptau}\frac{1+r\rho}{1-r\rho}+\frac{2(1-r)\rho}{(1-r\rho)(1-\rho)} \bigg). 
	\end{align*}
	In a same manner, one can show $ (ii) $ and $ (iii) $. Next, let us determine $ \sigma_{22} $. Analogously to $ (i) $, we get 
	\begin{align*} 
		\sigma_{22}&=\tfrac{1}{\uptau}\Big[\mu_{(2,2)}(0)-(n_{(2)})^2\pi^4+\tfrac{2}{\uptau}\sum_{h=1}^\infty 	\uptau(h)\Big(\mu_{(2,2)}(h)-(n_{(2)})^2\pi^4\Big) \Big]\\
		&=\tfrac{n_{(2)}\pi^2}{\uptau}\Bigg[4(n-1)(1-\pi)\pi\Big(1+\tfrac{2}{\uptau}\sum_{h=1}^\infty \uptau(h)\rho^{h} 	\Big)+2(1-\pi)^2\pi^2\Big(1+\tfrac{2}{\uptau}\sum_{h=1}^\infty \uptau(h)\rho^{2h} \Big)\Bigg].
	\end{align*}
	Using the Markov model $ \uptau(h)=\uptau^2+\uptau(1-\uptau)r^h $ as well as the result for $ \sigma_{11} $, we get 
	\begin{align*} 
		\sigma_{22}&=4(n-1)^2\pi^2\cdot\sigma_{11}+2n_{(2)}(1-\pi)^2\tfrac{1}{\uptau} \Bigg(\frac{1+(2\uptau-1)\rho^2}{1-\rho^2}-\frac{2r(\uptau-1)\rho^2}{1-r\rho^2} \Bigg)\\
		&=4(n-1)^2\pi^2\cdot\sigma_{11}+2n_{(2)}(1-\pi)^2\pi^2\bigg( \frac{1}{\uptau}\frac{1+r\rho^2}{1-r\rho^2}+\frac{2(1-r)\rho^2}{(1-r\rho^2)(1-\rho^2)} \bigg). 
	\end{align*}
	For $(v) $, we get 
	\begin{align*} 
		\sigma_{23}&=\tfrac{1}{\uptau}\Big[\mu_{(2,3)}(0)-n_{(2)}n_{(3)}\pi^5+\tfrac{2}{\uptau}\sum_{h=1}^\infty 		\uptau(h)\Big(\mu_{(2,3)}(h)-n_{(2)}n_{(3)}\pi^5\Big) \Big]\\
		&=\tfrac{6n_{(3)}(1-\pi)\pi^3}{\uptau}\Bigg[(n-1)\pi\Big(1+\tfrac{2}{\uptau}\sum_{h=1}^\infty \uptau(h)\rho^{h} 		\Big)+(1-\pi)\Big(1+\tfrac{2}{\uptau}\sum_{h=1}^\infty \uptau(h)\rho^{2h} \Big)\Bigg]\\
		&=3(n-2)\pi\cdot\sigma_{22}-6(n-1)^2(n-2)\pi^3\cdot\sigma_{11}.
	\end{align*}
	Finally, for $ (vi) $, we get 
	\begin{align*} 
		\sigma_{33}&=\tfrac{1}{\uptau}\Big[\mu_{(3,3)}(0)-(n_{(3)})^2\pi^6+\tfrac{2}{\uptau}\sum_{h=1}^\infty 		\uptau(h)\Big(\mu_{(3,3)}(h)-(n_{(3)})^2\pi^6\Big) \Big]\\
		&=\tfrac{3n_{(3)}(1-\pi)\pi^3}{\uptau}\Bigg[3(n-1)(n-2)\pi^2\Big(1+\tfrac{2}{\uptau}\sum_{h=1}^\infty \uptau(h)\rho^{h}\Big)+
		\\[1ex]&\quad+6(n-2)(1-\pi)\pi\Big(1+\tfrac{2}{\uptau}\sum_{h=1}^\infty \uptau(h)\rho^{2h} \Big)+2(1-\pi)^2\Big(1+\tfrac{2}{\uptau}\sum_{h=1}^\infty \uptau(h)\rho^{3h} 		\Big)\Bigg].
	\end{align*}	
	Again, the Markov model $ \uptau(h)=\uptau^2+\uptau(1-\uptau)r^h $ yields
	\begin{align*}
		\sigma_{33}&=9(n-2)^2\pi^2\cdot\sigma_{22}-27(n-1)^2(n-2)^2\pi^4\cdot\sigma_{11}
		\\[1ex]&\quad+6n_{(3)}(1-\pi)^3\pi^3\frac{1}{\uptau}\Bigg(\frac{1+(2\uptau-1)\rho^3}{1-\rho^3}-\frac{2r(\uptau-1)\rho^3}{1-r\rho^3} \Bigg)\nonumber\\
		&=9(n-2)^2\pi^2\cdot\sigma_{22}-27(n-1)^2(n-2)^2\pi^4\cdot\sigma_{11}	\\[1ex]&\quad+6n_{(3)}(1-\pi)^3\pi^3\bigg(\frac{1}{\uptau}\frac{1+r\rho^3}{1-r\rho^3}+\frac{2(1-r)\rho^3}{(1-r\rho^3)(1-\rho^3)} \bigg).
	\end{align*}
\end{proof}
%
\subsection{Poisson Index of Dispersion}\label{A_PInd}
%
In this section, we look at the derivations for the Poisson index of dispersion. Moreover, we start with the general approach for the asymptotic variance and bias, and then consider a specific Poi-INAR(1) model. Let us start with the general approach.
%
\subsubsection{Proof of Theorem \ref{DeJong}}\label{A_PInd_ThDeJong}
%
\setlength{\fboxrule}{1pt}
\fbox{Asymptotic variance}\\

First, let us draw our attention to the asymptotic variance. Let us start by applying the covariances $\sigma_{ij}$ from Theorem \ref{CLT_Fac} to \eqref{Doncic}, which leads to
\begin{align*}
	\sigma_{\hat{I}^\poi}^2&=\tfrac{1}{T\uptau} \Bigg[ \big(\tfrac{\mu_{(2)}}{\mu^2}+1\big)^2(\mu_{(1,1)}(0)-\mu^2)-\tfrac{2}{\mu}\big(\tfrac{\mu_{(2)}}{\mu^2}+1\big)(\mu_{(1,2)}(0)-\mu\mu_{(2)})
	\\[1ex]&\quad +\tfrac{1}{\mu^2}(\mu_{(2,2)}(0)-\mu_{(2)}^2)+\tfrac{2}{\uptau}\sum_{h=1}^\infty\uptau(h)\Big( \big(\tfrac{\mu_{(2)}}{\mu^2}+1\big)^2\big(\mu_{(1,1)}(h)-\mu^2\big)
	\\[1ex]&\quad+\tfrac{1}{\mu^2}\big(\mu_{(2,2)}(h)-\mu_{(2)}^2\big) -\tfrac{1}{\mu}\big(\tfrac{\mu_{(2)}}{\mu^2}+1\big)\big(\mu_{(2,1)}(h)+\mu_{(1,2)}(h)-2\mu\mu_{(2)}\big) \Big) \Bigg]\\
	&=\tfrac{1}{T\uptau\mu^2} \Bigg[ \big(\tfrac{\mu_{(2)}}{\mu}+\mu\big)^2(\mu_{(1,1)}(0)-\mu^2)-2\big(\tfrac{\mu_{(2)}}{\mu}+\mu\big)(\mu_{(1,2)}(0)-\mu\mu_{(2)})
	\\[1ex]&\quad +\mu_{(2,2)}(0)-\mu_{(2)}^2+\tfrac{2}{\uptau}\sum_{h=1}^\infty\uptau(h)\Big( \big(\tfrac{\mu_{(2)}}{\mu}+\mu\big)^2\big(\mu_{(1,1)}(h)-\mu^2\big)
	\\[1ex]&\quad+\mu_{(2,2)}(h)-\mu_{(2)}^2 -\big(\tfrac{\mu_{(2)}}{\mu}+\mu\big)\big(\mu_{(2,1)}(h)+\mu_{(1,2)}(h)-2\mu\mu_{(2)}\big) \Big) \Bigg].
\end{align*}
Now, we can use Lemma \ref{Lebron} to simplify $ \sigma_{\hat{I}^\poi}^2 $. Thus, we obtain 
\begin{align*}
	\sigma_{\hat{I}^\poi}^2&=\tfrac{1}{T\uptau\mu^2} \Bigg[\big(\tfrac{\mu_{(2)}}{\mu}+\mu\big)^2(\mu_{(2)}+\mu-\mu^2)-2\big(\tfrac{\mu_{(2)}}{\mu}+\mu\big)(\mu_{(3)}+2\mu_{(2)}-\mu\mu_{(2)})
	\\[1ex]&\quad+\mu_{(4)}+4\mu_{(3)}+2\mu_{(2)}-\mu_{(2)}^2+\tfrac{2}{\uptau}\sum_{h=1}^\infty\uptau(h)\Big( \big(\tfrac{\mu_{(2)}}{\mu}+\mu\big)^2\big(\mu_{(1,1)}(h)-\mu^2\big)
	\\[1ex]&\quad-\big(\tfrac{\mu_{(2)}}{\mu}+\mu\big)\big(\mu_{(2,1)}(h)+\mu_{(1,2)}(h)\big)+\mu_{(2,2)}(h)+\mu_{(2)}^2+2\mu^2\mu_{(2)} \Big) \Bigg]\\
	&=\tfrac{1}{T\uptau\mu^2} \Bigg[ 	\big(\tfrac{\mu_{(2)}}{\mu}+\mu\big)^2(\mu_{(2)}+\mu)-2\big(\tfrac{\mu_{(2)}}{\mu}+\mu\big)(\mu_{(3)}+2\mu_{(2)})
	\\[1ex]&\quad+\mu_{(4)}+4\mu_{(3)}+2\mu_{(2)}-\mu^4+\tfrac{2}{\uptau}\sum_{h=1}^\infty\uptau(h)\Big( \big(\tfrac{\mu_{(2)}}{\mu}+\mu\big)^2\mu_{(1,1)}(h)
	\\[1ex]&\quad-\big(\tfrac{\mu_{(2)}}{\mu}+\mu\big)\big(\mu_{(2,1)}(h)+\mu_{(1,2)}(h)\big)+\mu_{(2,2)}(h)-\mu^4 \Big) \Bigg].
\end{align*}

\setlength{\fboxrule}{1pt}
\fbox{Asymptotic bias}\\

In a similar fashion to the asymptotic variance, the asymptotic bias becomes
\begin{align*}
	\mathbb{B}_{\hat{I}^\poi}&=\tfrac{1}{T}\Big(\tfrac{1}{2}h_{11}\sigma_{11}+h_{12}\sigma_{12} \Big)=\tfrac{1}{T} \Bigg[ 	\tfrac{\mu_{(2)}}{\mu^3}\Big(\tfrac{1}{\uptau}(\mu_{(1,1)}(0)-\mu^2) +\tfrac{2}{\uptau^2}\sum_{h=1}^\infty\uptau(h)\big(\mu_{(1,1)}(h)-\mu^2\big)\Big)
	\\[1ex]&\quad- \tfrac{1}{\mu^2}\Big(\tfrac{1}{\uptau}(\mu_{(1,2)}(0)-\mu\mu_{(2)}) 	+\tfrac{1}{\uptau^2}\sum_{h=1}^\infty\uptau(h)\big(\mu_{(2,1)}(h)+\mu_{(1,2)}(h)-2\mu\mu_{(2)}\big)\Big)\Bigg]\\
	&=\tfrac{1}{T\uptau\mu^3} \Bigg[ \mu_{(2)}\big(\mu_{(1,1)}(0)-\mu^2 	\big)-\mu\big(\mu_{(1,2)}(0)-\mu\mu_{(2)}\big)
	\\[1ex]&\quad+\tfrac{1}{\uptau}\sum_{h=1}^\infty\uptau(h)\Big(2\mu_{(2)}\big(\mu_{(1,1)}(h)-\mu^2 \big)-\mu\big(\mu_{(2,1)}(h)+\mu_{(1,2)}(h)-2\mu\mu_{(2)}\big)\Big) \Bigg].
\end{align*}
Finally, using Lemma \ref{Lebron} leads to
\begin{align*}
	\mathbb{B}_{\hat{I}^\poi}=\tfrac{1}{T\uptau\mu^3} \Bigg[  	\mu_{(2)}^2-\mu\big(\mu_{(2)}+\mu_{(3)}\big)+\tfrac{2}{\uptau}\sum_{h=1}^\infty\uptau(h)\Big(\mu_{(2)}\mu_{(1,1)}(h)
	-\tfrac{\mu}{2}\big(\mu_{(2,1)}(h)+\mu_{(1,2)}(h)\big)\Big)\Bigg].
\end{align*}
%
\subsubsection{Proof of Corollary \ref{Nico}}\label{A_PInd_CorNico}
%
For the proof, we shall utilize our results from Section \ref{A_PoiModel}, where we assumed that the missing data follow a Markov model, \ie $ \uptau(h)=\uptau^2+\uptau(1-\uptau)r^h $.

First, let us derive the asymptotic variance. As already mentioned, for a Poisson distribution, one has $ \mu_{(k)}=\mu^{k} $. Thus, we obtain from \eqref{Aguero} 
that $ \D=\big(-2,\tfrac{1}{\mu}\big)$. Now, using Lemma \ref{Pique} leads to
\begin{align*}
	\sigma_{\hat{I}^\poi}^2&=\tfrac{1}{T} \Big( d_1^2\sigma_{11} +d_2^2\sigma_{22} +2 d_1d_2\sigma_{12}\Big)=\tfrac{1}{T}\bigg[-4\cdot\sigma_{11}+\tfrac{1}{\mu^2}\cdot\sigma_{22}\bigg].
\end{align*}
From Lemma \ref{Pique} $(iv) $, we get
\begin{align*}
	\sigma_{\hat{I}^\poi}^2&=\tfrac{1}{T}\Bigg[ -4\cdot\sigma_{11}+\tfrac{1}{\mu^2}\Bigg( 4\mu^2\cdot\sigma_{11}+2\mu^2\bigg( \frac{1}{\uptau}\frac{1+r\rho^2}{1-r\rho^2}+\frac{2(1-r)\rho^2}{(1-r\rho^2)(1-\rho^2)} \bigg)\Bigg)\Bigg]\\
	&=\tfrac{2}{T}\Bigg[ \frac{1}{\uptau}\frac{1+r\rho^2}{1-r\rho^2}+\frac{2(1-r)\rho^2}{(1-r\rho^2)(1-\rho^2)}\Bigg].
\end{align*}
For the bias, we obtain from \eqref{Aguero} 
that $ h_{11}=\tfrac{2}{\mu} $, $h_{12}=-\tfrac{1}{\mu^2} $. Again, using Lemma \ref{Pique} leads to
\begin{align*}
	\mathbb{B}_{\hat{I}^\poi}&=-\tfrac{1}{T} \bigg[ \tfrac{1}{\mu}\cdot\sigma_{11}\bigg]=-\tfrac{1}{T}\Bigg[ \frac{1}{\uptau}\frac{1+r\rho}{1-r\rho}+\frac{2(1-r)\rho}{(1-r\rho)(1-\rho)} \Bigg].
\end{align*}
%
\subsection{Binomial Index of Dispersion}\label{A_BinInd}
In this section, we proceed as in Section \ref{A_PInd}, that is, we start by deriving a general approach for asymptotic variance and bias of the binomial index of dispersion, and then specify it for the BAR(1) model.

%
\subsubsection{Proof of Theorem \ref{Memphis}}\label{A_PInd_ThMemphis}
%
\setlength{\fboxrule}{1pt}
\fbox{Asymptotic variance}\\

Let us start by applying the covariances $\sigma_{ij}$ from Theorem \ref{CLT_Fac} 
to \eqref{Depay}, which leads to
\begin{align*}
	\sigma_{\hat{I}^\bin}^2&=\tfrac{1}{T} \Bigg[ 	\Bigg(\tfrac{n\big(\mu^2(1-n)-n\mu_{(2)}+2\mu\mu_{(2)}\big)}{\mu^2(n-\mu)^2}\Bigg)^2\Big(\tfrac{1}{\uptau}(\mu_{(1,1)}(0)-\mu^2) +\tfrac{2}{\uptau^2}\sum_{h=1}^\infty\uptau(h)\big(\mu_{(1,1)}(h)-\mu^2\big)\Big)
	\\[1ex]&\quad+2\Bigg(\tfrac{n\big(\mu^2(1-n)-n\mu_{(2)}+2\mu\mu_{(2)}\big)}{\mu^2(n-\mu)^2}\Bigg)\bigg(\frac{n}{\mu(n-\mu)}\bigg)\Big(\tfrac{1}{\uptau}(\mu_{(1,2)}(0)-\mu\mu_{(2)}) \\[1ex]&\quad+\tfrac{1}{\uptau^2}\sum_{h=1}^\infty\uptau(h)\big(\mu_{(2,1)}(h)+\mu_{(1,2)}(h)-2\mu\mu_{(2)}\big)\Big)
	\\[1ex]&\quad+\tfrac{n^2}{\mu^2(n-\mu)^2}\Big(\tfrac{1}{\uptau}(\mu_{(2,2)}(0)-\mu_{(2)}^2) +\tfrac{2}{\uptau^2}\sum_{h=1}^\infty\uptau(h)\big(\mu_{(2,2)}(h)-\mu_{(2)}^2\big)\Big) \Bigg]\\
	&=\tfrac{1}{T\uptau} \Bigg[ 	\Bigg(\tfrac{n\big(\mu^2(1-n)-n\mu_{(2)}+2\mu\mu_{(2)}\big)}{\mu^2(n-\mu)^2}\Bigg)^2\big(\mu_{(1,1)}(0)-\mu^2\big)
	\\[1ex]&\quad+2\Bigg(\tfrac{n\big(\mu^2(1-n)-n\mu_{(2)}+2\mu\mu_{(2)}\big)}{\mu^2(n-\mu)^2}\Bigg)\bigg(\tfrac{n}{\mu(n-\mu)}\bigg)\big(\mu_{(1,2)}(0)-\mu\mu_{(2)}\big)
	\\[1ex]&\quad+\tfrac{n^2}{\mu^2(n-\mu)^2}\big(\mu_{(2,2)}(0)-\mu_{(2)}^2\big)
	\\[1ex]&\quad +\tfrac{2}{\uptau}\sum_{h=1}^\infty\uptau(h)\Bigg( 	\Bigg(\tfrac{n\big(\mu^2(1-n)-n\mu_{(2)}+2\mu\mu_{(2)}\big)}{\mu^2(n-\mu)^2}\Bigg)^2\big(\mu_{(1,1)}(h)-\mu^2\big)
	\\[1ex]&\quad+\Bigg(\tfrac{n\big(\mu^2(1-n)-n\mu_{(2)}+2\mu\mu_{(2)}\big)}{\mu^2(n-\mu)^2}\Bigg)\bigg(\tfrac{n}{\mu(n-\mu)}\bigg)\big(\mu_{(2,1)}(h)+\mu_{(1,2)}(h)-2\mu\mu_{(2)}\big)
	\\[1ex]&\quad+\tfrac{n^2}{\mu^2(n-\mu)^2}\big(\mu_{(2,2)}(h)-\mu_{(2)}^2\big) \Bigg) \Bigg].
\end{align*}
Here, we can use Lemma \ref{Lebron} to simplify $ \sigma_{\hat{I}^\bin}^2 $. Thus, we obtain
\begin{align*}
	\sigma_{\hat{I}^\bin}^2&=\tfrac{n^2}{T\uptau\mu^4(n-\mu)^4} \Bigg[ 	\big(\mu^2(1-n)-n\mu_{(2)}+2\mu\mu_{(2)}\big)^2\big(\mu_{(2)}+\mu-\mu^2\big)
	\\[1ex]&\quad+2\mu(n-\mu)\big(\mu^2(1-n)-n\mu_{(2)}+2\mu\mu_{(2)}\big)\big(\mu_{(3)}+2\mu_{(2)}-\mu\mu_{(2)}\big)
	\\[1ex]&\quad+\mu^2(n-\mu)^2\big(\mu_{(4)}+4\mu_{(3)}+2\mu_{(2)}-\mu_{(2)}^2\big)
	\\[1ex]&\quad+\tfrac{2}{\uptau}\sum_{h=1}^\infty\uptau(h)\Bigg(\big(\mu^2(1-n)-n\mu_{(2)}+2\mu\mu_{(2)}\big)^2\big(\mu_{(1,1)}(h)-\mu^2\big)
	\\[1ex]&\quad+\mu(n-\mu)\big(\mu^2(1-n)-n\mu_{(2)}+2\mu\mu_{(2)}\big)\big(\mu_{(2,1)}(h)+\mu_{(1,2)}(h)-2\mu\mu_{(2)}\big)
	\\[1ex]&\quad +\mu^2(n-\mu)^2\big(\mu_{(2,2)}(h)-\mu_{(2)}^2\big) \Bigg) \Bigg].
\end{align*}

\setlength{\fboxrule}{1pt}
\fbox{Asymptotic bias}\\

In a similar fashion, the asymptotic bias becomes
\begin{align*}
	\mathbb{B}_{\hat{I}^\bin}&=\tfrac{1}{T} \Big(	\tfrac{1}{2}h_{11}\sigma_{11} + 	h_{12}\sigma_{12}\Big)\\
	&=\tfrac{1}{T}\Bigg[\tfrac{n\big(\mu^3(1-n)+n^2\mu_{(2)}+3\mu\mu_{(2)}(\mu-n)\big)}{\mu^3(n-\mu)^3}\Bigg(\tfrac{1}{\uptau}(\mu_{(1,1)}(0)-\mu^2)+\tfrac{2}{\uptau^2}\sum_{h=1}^\infty\uptau(h)\big(\mu_{(1,1)}(h)-\mu^2\big)\Bigg)
	\\[1ex]&\quad+\tfrac{n(2\mu-n)}{\mu^2(n-\mu)^2}\Big(\tfrac{1}{\uptau}(\mu_{(1,2)}(0)-\mu\mu_{(2)}) +\tfrac{1}{\uptau^2}\sum_{h=1}^\infty\uptau(h)\big(\mu_{(2,1)}(h)+\mu_{(1,2)}(h)-2\mu\mu_{(2)}\big)\Big)\Bigg]\\
	&=\tfrac{n}{T\uptau\mu^3(n-\mu)^3}\Bigg[\big(\mu^3(1-n)+n^2\mu_{(2)}+3\mu\mu_{(2)}(\mu-n)\big)\big(\mu_{(2)}+\mu-\mu^2\big)
	\\[1ex]&\quad+\mu(n-\mu)(2\mu-n)(\mu_{(3)}+2\mu_{(2)}-\mu\mu_{(2)})
	\\[1ex]&\quad+\tfrac{1}{\uptau^2}\sum_{h=1}^\infty\uptau(h)\Bigg(2\big(\mu^3(1-n)+n^2\mu_{(2)}+3\mu\mu_{(2)}(\mu-n)\big)\big(\mu_{(1,1)}(h)-\mu^2\big)
	\\[1ex]&\quad+\mu(n-\mu)(2\mu-n)\big(\mu_{(2,1)}(h)+\mu_{(1,2)}(h)-2\mu\mu_{(2)}\big)\Bigg)\Bigg],
\end{align*}
where in the last step, we used Lemma \ref{Lebron} once again.
%
\subsubsection{Proof of Corollary \ref{Demir}}\label{A_PInd_CorDemir}
%
For the proof, we shall utilize our results from Section \ref{A_BinModel}, where we assumed that the missing data follow a Markov model, \ie $ \uptau(h)=\uptau^2+\uptau(1-\uptau)r^h $.\\

First, let us derive the asymptotic variance. As already mentioned, for a binomial distribution, one has $ \mu_{(k)}=n_{(k)}\pi^{k} $. Thus, we obtain from \eqref{Gavi1} that $ \D=\big(-\tfrac{2(n-1)}{n(1-\pi)}, \tfrac{1}{n\pi(1-\pi)}\big)$. Now, using Lemma \ref{Pedri} leads to
\begin{align*}
	\sigma_{\hat{I}^\bin}^2&=\tfrac{1}{T} \Big( d_1^2\sigma_{11} +d_2^2\sigma_{22} +2 d_1d_2\sigma_{12}\Big)=\tfrac{1}{T}\bigg[-\tfrac{4(n-1)^2}{n^2(1-\pi)^2}\cdot\sigma_{11}+\tfrac{1}{n^2\pi^2(1-\pi)^2}\cdot\sigma_{22}\bigg]
\end{align*}
From Lemma \ref{Pedri} $(iv) $, we get
\begin{align*}
	\sigma_{\hat{I}^\bin}^2&=\tfrac{1}{T}\Bigg[ -\tfrac{4(n-1)^2}{n^2(1-\pi)^2}\cdot\sigma_{11}+\tfrac{1}{n^2\pi^2(1-\pi)^2}\Bigg( 4(n-1)^2\pi^2\cdot\sigma_{11}
	\\[1ex]&\quad+2n_{(2)}(1-\pi)^2\pi^2\bigg( \frac{1}{\uptau}\frac{1+r\rho^2}{1-r\rho^2}+\frac{2(1-r)\rho^2}{(1-r\rho^2)(1-\rho^2)} \bigg)\Bigg)\Bigg]\\
	&=\tfrac{2}{T}\Big(1-\tfrac{1}{n}\Big)\Bigg[ \frac{1}{\uptau}\frac{1+r\rho^2}{1-r\rho^2}+\frac{2(1-r)\rho^2}{(1-r\rho^2)(1-\rho^2)}\Bigg].
\end{align*}
For the bias, we obtain from \eqref{Gavi2} that $ h_{11}=\tfrac{2(n-1)}{n(1-\pi)} $, $h_{12}=\tfrac{2\pi-1}{n^2\pi^2(1-\pi)^2} $. Again, using Lemma \ref{Pedri} leads to
\begin{align*}
	\mathbb{B}_{\hat{I}^\bin}&=-\tfrac{1}{T} \bigg[ \tfrac{(n-1)}{n^2\pi(1-\pi)}\cdot\sigma_{11}\bigg]=-\tfrac{1}{T}\Big(1-\tfrac{1}{n}\Big)\Bigg[ \frac{1}{\uptau}\frac{1+r\rho}{1-r\rho}+\frac{2(1-r)\rho}{(1-r\rho)(1-\rho)} \Bigg].
\end{align*}
%
\subsection{Skewness Index}\label{A_SkewInd}
%
Let us start by stating the sample counterpart to the skewness index, which is given by $ \hat{I}_\skw =\hat{\mu}_{(3)}/\hat{\mu}_{(2)}\hat{\mu}$. Then, we proceed from Theorem \ref{CLT_Fac} 
and define a new function $ g$ by  
\begin{align*}	
	g(x_1,x_2,x_3)=\frac{x_3}{x_1x_2}. 
\end{align*}
Thus, $ g$ has the partial derivatives
\begin{align*}
	&\frac{\partial}{\partial x_1}g=-\frac{x_3}{x_1^2x_2}, \qquad \frac{\partial}{\partial x_2}g=-\frac{x_3}{x_1x_2^2}, \qquad \frac{\partial}{\partial x_3}g=\frac{1}{x_1x_2},\\
	&\frac{\partial^2}{\partial x_1^2}g=\frac{2x_3}{x_1^3x_2}, \qquad \frac{\partial^2}{\partial x_2^2}g=\frac{2x_3}{x_1x_2^3}, \qquad \frac{\partial^2}{\partial x_3^2}g=0,\\
	&\frac{\partial^2}{\partial x_1\partial x_2}g=\frac{x_3}{x_1^2x_2^2}, \qquad \frac{\partial^2}{\partial x_1\partial x_3}g=-\frac{1}{x_1^2x_2}, \qquad \frac{\partial^2}{\partial x_2\partial x_3}g=-\frac{1}{x_1x_2^2}.
\end{align*}
So, we get the Jacobian $ \D $ and the Hessian $ \Hes $ by evaluating the partial derivatives in $ \bmu$, which leads to 
\begin{align*}
	\D=\frac{1}{\mu_{(2)}\mu}\bigg(-\frac{\mu_{(3)}}{\mu} ,\ -\frac{\mu_{(3)}}{\mu_{(2)}} ,\ 1\bigg),\\
	\\
	\Hes=\frac{1}{\mu_{(2)}\mu}\begin{pmatrix} 
		\frac{2\mu_{(3)}}{\mu^2} & \frac{\mu_{(3)}}{\mu\mu_{(2)}} & -\frac{1}{\mu} \\
		\frac{\mu_{(3)}}{\mu\mu_{(2)}} &  \frac{2\mu_{(3)}}{\mu_{(2)}^2} & -\frac{1}{\mu_{(2)}}\\
		-\frac{1}{\mu} & -\frac{1}{\mu_{(2)}} & 0
	\end{pmatrix}.
\end{align*}
This, however, leads to the Taylor approximation $ \hat{I}_{\text{Skew}}\approx I_{\text{Skew}} + \D(\hat{\bmu}-\bmu)+\tfrac{1}{2}(\hat{\bmu}-\bmu)^\top\Hes(\hat{\bmu}-\bmu) $, which can be used to conclude the asymptotic variance and bias of $ \hat{I}_{\text{Skew}} $, that is
\begin{align}\label{Araujo1}
	\sigma_{\hat{I}_\skw}^2&=\tfrac{1}{T} \Big( d_1^2\sigma_{11} +d_2^2\sigma_{22} +d_3^2\sigma_{33} +2 d_1d_2\sigma_{12} +2 d_1d_3\sigma_{13}  +2 d_2d_3\sigma_{23} \Big),\\
	\mathbb{B}_{\hat{I}_\skw}&=\tfrac{1}{T} \Big( \tfrac{1}{2}\big( h_{11}\sigma_{11} +h_{22}\sigma_{22} +h_{33}\sigma_{33}\big) +h_{12}\sigma_{12} +h_{13}\sigma_{13}  +h_{23}\sigma_{23} \Big).\label{Araujo2}
\end{align}
\begin{lemma}\label{Batman}
	We have for Jacobian $ \D $ and the Hessian $ \Hes $ of the Poisson and the binomial distribution:
	\begin{align*}
		&(i)\quad \D^{\text{Poi}}=\tfrac{1}{\mu^3}\big(-\mu^2,-\mu,1\big), \quad \Hes^{\text{Poi}}=\frac{1}{\mu^3}\begin{pmatrix} 
			2\mu & 1 & -\tfrac{1}{\mu} \\
			1 & \tfrac{2}{\mu} & -\tfrac{1}{\mu^2} \\
			-\tfrac{1}{\mu} & -\tfrac{1}{\mu^2} & 0
		\end{pmatrix};\\[0.4cm]
		&(ii)\quad \D^{\text{Bin}}=\tfrac{1}{nn_{(2)}\pi^3}\big(-(n-1)(n-2)\pi^2,-(n-2)\pi,1\big),\\[0.2cm] &\qquad \Hes^{\text{Bin}}=\frac{1}{nn_{(2)}\pi^3}\begin{pmatrix} 
			\tfrac{2n_{(3)}\pi}{n^2} & \tfrac{n-2}{n} & -\tfrac{1}{n\pi} \\
			\tfrac{n-2}{n} & \tfrac{2(n-2)}{n_{(2)}\pi} & -\tfrac{1}{n_{(2)}\pi^2} \\
			-\tfrac{1}{n\pi} & -\tfrac{1}{n_{(2)}\pi^2} & 0
		\end{pmatrix}.\\
	\end{align*}
\end{lemma}
%
\subsubsection{Proof of Theorem \ref{Ronaldo1}}\label{A_SkewInd_ThRo1}
%
For the proof, we shall utilize our results from Section \ref{A_PoiModel}, where we assumed that the missing data follow a Markov model, \ie $ \uptau(h)=\uptau^2+\uptau(1-\uptau)r^h $.\\

First, let us derive the asymptotic variance. We plug-in our results from Lemma \ref{Batman} $ (i) $ as well as Lemma \ref{Pique} into \eqref{Araujo1}. We get
\begin{align*}
	\sigma_{\hat{I}_\skw^\poi}^2&=\tfrac{1}{T} \Big( d_1^2\sigma_{11} +d_2^2\sigma_{22} +d_3^2\sigma_{33} +2 d_1d_2\sigma_{12} +2 d_1d_3\sigma_{13}  +2 d_2d_3\sigma_{23} \Big)\\
	&=\tfrac{1}{T}\Bigg[\tfrac{1}{\mu^2}\cdot\sigma_{11}+\tfrac{1}{\mu^4}\cdot\sigma_{22}-2\bigg(\tfrac{2}{\mu^2}\cdot\sigma_{11}+\tfrac{1}{\mu^5}\big(3\mu\cdot\sigma_{22}-6\mu^3\sigma_{11} \big) \bigg)
	\\[1ex]&\quad+\tfrac{1}{\mu^6}\Bigg(9\mu^2\cdot\sigma_{22}-27\mu^4\cdot\sigma_{11}    +6\mu^3\bigg(\frac{1}{\uptau}\frac{1+r\rho^3}{1-r\rho^3}+\frac{2(1-r)\rho^3}{(1-r\rho^3)(1-\rho^3)} \bigg) \Bigg)
	\Bigg]\\
	&=\tfrac{1}{T}\Bigg[ \tfrac{-16}{\mu^2}\cdot\sigma_{11}+\tfrac{4}{\mu^4}\cdot\sigma_{22}+6\mu^3\bigg(\frac{1}{\uptau}\frac{1+r\rho^3}{1-r\rho^3}+\frac{2(1-r)\rho^3}{(1-r\rho^3)(1-\rho^3)} \bigg) \Bigg].
\end{align*}
Now, using Lemma \ref{Pique} $(iv) $ yields
\begin{align*}
	\sigma_{\hat{I}_\skw^\poi}^2&=\tfrac{1}{T}\Bigg[ -\tfrac{16}{\mu^2}\cdot\sigma_{11}+\tfrac{4}{\mu^4}\Bigg( 4\mu^2\cdot\sigma_{11}+2\mu^2\bigg( \frac{1}{\uptau}\frac{1+r\rho^2}{1-r\rho^2}+\frac{2(1-r)\rho^2}{(1-r\rho^2)(1-\rho^2)} \bigg)\Bigg)
	\\[1ex]&\quad+6\mu^3\bigg(\frac{1}{\uptau}\frac{1+r\rho^3}{1-r\rho^3}+\frac{2(1-r)\rho^3}{(1-r\rho^3)(1-\rho^3)} \bigg) \Bigg]\\
	&=\tfrac{1}{T\mu^3}\Bigg[8\mu\bigg( \frac{1}{\uptau}\frac{1+r\rho^2}{1-r\rho^2}+\frac{2(1-r)\rho^2}{(1-r\rho^2)(1-\rho^2)} \bigg)+6\bigg( \frac{1}{\uptau}\frac{1+r\rho^3}{1-r\rho^3}+\frac{2(1-r)\rho^3}{(1-r\rho^3)(1-\rho^3)} \bigg)\Bigg].
\end{align*}
For the bias, we use Lemma \ref{Batman} $ (i) $. Again, using  Lemma \ref{Pique} and plugging-in our results into \eqref{Araujo2}, leads to
\begin{align*}
	\mathbb{B}_{\hat{I}_\skw^\poi}&=\tfrac{1}{T} \Big( \tfrac{1}{2}\big( h_{11}\sigma_{11} +h_{22}\sigma_{22} +h_{33}\sigma_{33}\big) +h_{12}\sigma_{12} +h_{13}\sigma_{13}  +h_{23}\sigma_{23} \Big)\\
	&=\tfrac{1}{T}\Big[\tfrac{1}{2}\big( \tfrac{2}{\mu^2}\cdot\sigma_{11} +\tfrac{2}{\mu^4}\cdot\sigma_{22}\big)+\tfrac{1}{\mu^3}\cdot\sigma_{12}-\tfrac{1}{\mu^4}\cdot\sigma_{13}-\tfrac{1}{\mu^5}\cdot\sigma_{23}\Big]\\
	&=\tfrac{1}{T\mu^2}\Big[\sigma_{11} +\tfrac{1}{\mu^2}\cdot\sigma_{22}+\tfrac{1}{\mu}\cdot\sigma_{12}-\tfrac{1}{\mu^2}\cdot\sigma_{13}-\tfrac{1}{\mu^3}\cdot\sigma_{23}\Big]\\
	&=\tfrac{1}{T\mu^2}\Big[6\cdot\sigma_{11} -\tfrac{2}{\mu^2}\cdot\sigma_{22}\Big].
\end{align*}
Now, using Lemma \ref{Pique} $(i)$ and $ (iii) $ yields
\begin{align*}
	\mathbb{B}_{\hat{I}_\skw^\poi}&=\tfrac{1}{T\mu^2}\Bigg[ 6\cdot\sigma_{11}-\tfrac{2}{\mu^2}\Bigg( 4\mu^2\cdot\sigma_{11}+2\mu^2\bigg( \frac{1}{\uptau}\frac{1+r\rho^2}{1-r\rho^2}+\frac{2(1-r)\rho^2}{(1-r\rho^2)(1-\rho^2)} \bigg)\Bigg)		\Bigg]\\
	&=-\tfrac{2}{T\mu^2}\Bigg[ \sigma_{11}+2\bigg( \frac{1}{\uptau}\frac{1+r\rho^2}{1-r\rho^2}+\frac{2(1-r)\rho^2}{(1-r\rho^2)(1-\rho^2)} \bigg)		\Bigg]\\
	&=-\tfrac{2}{T\mu^2}\Bigg[
	\mu\bigg( \frac{1}{\uptau}\frac{1+r\rho}{1-r\rho}+\frac{2(1-r)\rho}{(1-r\rho)(1-\rho)} \bigg)+2\bigg( \frac{1}{\uptau}\frac{1+r\rho^2}{1-r\rho^2}+\frac{2(1-r)\rho^2}{(1-r\rho^2)(1-\rho^2)} \bigg)\Bigg].
\end{align*}
%
\subsubsection{Proof of Theorem \ref{Ronaldo2}}\label{A_SkewInd_ThRo2}
%
For the proof, we utilize our results from Section \ref{A_BinModel}, where we assumed that the missing data follow a Markov model, \ie $ \uptau(h)=\uptau^2+\uptau(1-\uptau)r^h $.\\

First, let us start with the asymptotic variance. Let us plug-in our results from Lemma \ref{Batman} $ (ii) $ as well as Lemma \ref{Pedri} into \eqref{Araujo1}, which leads to 
\begin{align*}
	\sigma_{\hat{I}_\skw^\bin}^2&=\tfrac{1}{T} \Big( d_1^2\sigma_{11} +d_2^2\sigma_{22} +d_3^2\sigma_{33} +2 d_1d_2\sigma_{12} +2 d_1d_3\sigma_{13}  +2 d_2d_3\sigma_{23} \Big)\\
	&=\tfrac{1}{T}\Bigg[\tfrac{(n-2)^2}{n^4\pi^2}\cdot\sigma_{11}+\tfrac{(n-2)^2}{n^4(n_{(2)})^2\pi^2}\cdot\sigma_{22}+\tfrac{1}{n^2(n_{(2)})^2\pi^6}\Bigg( 9(n-2)^2\pi^2\cdot\sigma_{22}
	\\[1ex]&\quad-27(n-1)^2(n-2)^2\pi^4\cdot\sigma_{11}+6\tfrac{(n-2)(1-\pi)^3}{n^2n_{(2)}\pi^3}\bigg(\frac{1}{\uptau}\frac{1+r\rho^3}{1-r\rho^3}+\frac{2(1-r)\rho^3}{(1-r\rho^3)(1-\rho^3)} \bigg)
	\Bigg)
	\\[1ex]&\quad+\tfrac{4(n-1)(n-2)^2}{n^3n_{(2)}\pi^2}\cdot\sigma_{11}-\tfrac{6(n-1)(n-2)^2}{n^3n_{(2)}\pi^2}\cdot\sigma_{11}
	\\[1ex]&\quad-\tfrac{2(n-2)}{n^3(n_{(2)})^2\pi^5}\bigg( 3(n-2)\pi\cdot\sigma_{22}-6(n-1)^2(n-2)\pi^3\cdot\sigma_{11}\bigg)
	\Bigg]\\
	&=\tfrac{1}{T}\Bigg[ \tfrac{-16(n-2)^2}{n^4\pi^2}\cdot\sigma_{11}+\tfrac{4(n-2)^2}{n^2(n_{(2)})^2\pi^4}\cdot\sigma_{22}+6\tfrac{(n-2)(1-\pi)^3}{n^2n_{(2)}\pi^3}\bigg(\frac{1}{\uptau}\frac{1+r\rho^3}{1-r\rho^3}+\frac{2(1-r)\rho^3}{(1-r\rho^3)(1-\rho^3)} \bigg) \Bigg].
\end{align*}
Now, using Lemma \ref{Pedri} $(iv) $ yields
\begin{align*}
	\sigma_{\hat{I}_\skw^\bin}^2&=\tfrac{1}{T}\Bigg[ \tfrac{8(n-2)^2(1-\pi)^2}{n^2n_{(2)}\pi^2} \bigg( \frac{1}{\uptau}\frac{1+r\rho^2}{1-r\rho^2}+\frac{2(1-r)\rho^2}{(1-r\rho^2)(1-\rho^2)} \bigg)\
	\\[1ex]&\quad+6\tfrac{(n-2)(1-\pi)^3}{n^2n_{(2)}\pi^3}\bigg(\frac{1}{\uptau}\frac{1+r\rho^3}{1-r\rho^3}+\frac{2(1-r)\rho^3}{(1-r\rho^3)(1-\rho^3)} \bigg) \Bigg]\\
	&=\tfrac{(n-2)(1-\pi)^3}{(n-1)Tn^3\pi^3}\Bigg[8\cdot\frac{(n-2)\pi}{1-\pi}\bigg( \frac{1}{\uptau}\frac{1+r\rho^2}{1-r\rho^2}+\frac{2(1-r)\rho^2}{(1-r\rho^2)(1-\rho^2)} \bigg)
	\\[1ex]&\quad+6\bigg( \frac{1}{\uptau}\frac{1+r\rho^3}{1-r\rho^3}+\frac{2(1-r)\rho^3}{(1-r\rho^3)(1-\rho^3)} \bigg)\Bigg]\\
	&=\Big(\tfrac{(n-2)(n-\mu)^3}{(n-1)n^3}\Big)\tfrac{1}{T\mu^3}\Bigg[\tfrac{n-2}{n-\mu}\cdot8\mu\bigg( \frac{1}{\uptau}\frac{1+r\rho^2}{1-r\rho^2}+\frac{2(1-r)\rho^2}{(1-r\rho^2)(1-\rho^2)} \bigg)
	\\[1ex]&\quad+6\bigg( \frac{1}{\uptau}\frac{1+r\rho^3}{1-r\rho^3}+\frac{2(1-r)\rho^3}{(1-r\rho^3)(1-\rho^3)} \bigg)\Bigg].
\end{align*}
In the last step, we used that $ \pi=\tfrac{\mu}{n}$. For the bias, we use Lemma \ref{Batman} $ (ii) $. Once again, using  Lemma \ref{Pedri} and plugging-in our results into \eqref{Araujo2},  leads to
\begin{align*}
	\mathbb{B}_{\hat{I}_\skw^\bin}&=\tfrac{1}{T} \Big( \tfrac{1}{2}\big( h_{11}\sigma_{11} +h_{22}\sigma_{22} +h_{33}\sigma_{33}\big) +h_{12}\sigma_{12} +h_{13}\sigma_{13}  +h_{23}\sigma_{23} \Big)\\
	&=\tfrac{1}{T}\Big[\tfrac{1}{2}\bigg( \tfrac{2(n-2)}{n^3\pi^2}\cdot\sigma_{11} +\tfrac{2(n-2)}{n(n_{(2)})^2\pi^4}\cdot\sigma_{22}\bigg)+\tfrac{2(n-2)}{n^3\pi^2}\cdot\sigma_{11}-\tfrac{3(n-2)}{n^3\pi^2}\cdot\sigma_{11}
	\\[1ex]&\quad-\tfrac{1}{n(n_{(2)})^2\pi^5}\bigg(3(n-2)\pi\sigma_{22}-6(n-1)^2(n-2)^2\pi^4\sigma_{11} \bigg)\Big]\\
	&=\tfrac{1}{T}\Big[\tfrac{6(n-2)}{n^3\pi^2}\cdot\sigma_{11} -\tfrac{2(n-2)}{n(n_{(2)})^2\pi^4}\cdot\sigma_{22}\Big].
\end{align*}
Now, using Lemma \ref{Pedri} $(i)$ and $ (iv) $ yields
\begin{align*}
	\mathbb{B}_{\hat{I}_\skw^\bin}&=-\tfrac{2}{T}\Bigg[
	\tfrac{(n-2)(1-\pi)}{n^2\pi}\bigg( \frac{1}{\uptau}\frac{1+r\rho}{1-r\rho}+\frac{2(1-r)\rho}{(1-r\rho)(1-\rho)} \bigg)
	\\[1ex]&\quad+\tfrac{2(n-2)(1-\pi)^2}{n^2n_{(2)}\pi^2}\bigg( \frac{1}{\uptau}\frac{1+r\rho^2}{1-r\rho^2}+\frac{2(1-r)\rho^2}{(1-r\rho^2)(1-\rho^2)} \bigg)\Bigg]\\
	&=-\tfrac{2(n-2)(1-\pi)^2}{(n-1)Tn^2\pi^2}\Bigg[\frac{(n-1)\pi}{1-\pi}\bigg( \frac{1}{\uptau}\frac{1+r\rho}{1-r\rho}+\frac{2(1-r)\rho}{(1-r\rho)(1-\rho)} \bigg)
	\\[1ex]&\quad+2\bigg( \frac{1}{\uptau}\frac{1+r\rho^2}{1-r\rho^2}+\frac{2(1-r)\rho^2}{(1-r\rho^2)(1-\rho^2)} \bigg)\Bigg]\\
	&=-\Big(\tfrac{(n-2)(n-\mu)^2}{(n-1)n^2}\Big)\tfrac{2}{T\mu^2}\Bigg[\tfrac{n-1}{n-\mu}\cdot\mu\bigg( \frac{1}{\uptau}\frac{1+r\rho}{1-r\rho}+\frac{2(1-r)\rho}{(1-r\rho)(1-\rho)} \bigg)
	\\[1ex]&\quad+2\bigg( \frac{1}{\uptau}\frac{1+r\rho^2}{1-r\rho^2}+\frac{2(1-r)\rho^2}{(1-r\rho^2)(1-\rho^2)} \bigg)\Bigg].
\end{align*}
%
\subsection{Autocorrelation and Missing Data}\label{A_AppData}
%
In this section, we briefly describe our method for estimating the autocorrelation function for cloud from data with missing values. We use the method provided by \cite{DunRob81}. Let us recall that the amplitude modulation of $ (\X_t) $ is $ (O_t\cdot\X_t) $, implying that we estimate $ \bmu $ by 
\begin{align*}
	\hat{\bmu} = \frac{\overline{O\,\X}}{\overline{O}}.
\end{align*}
In \citet[p. 260f]{DunRob81}, they first consider a known mean of zero. When the mean is unknown a mean correction is necessary, see \citet[p. 277f]{DunRob81}.
As a result, we can estimate the autocorrelation function as follows:
\begin{align}\label{rhodunsmuir}
	\hat{\rho}_{DR}(l)=\frac{\hat{C}_{DR}(l)}{\hat{C}_{DR}(0)},\qquad 0\leq l<T,
\end{align}
where $ \hat{C}_{DR}(l)=\frac{1}{T}\sum_{t=1}^{T-l}O_tO_{t+l}(X_t -\hat{\mu}) ( X_{t+l} -\hat{\mu}) $. If we now assume that we have \iid data, the asymptotic variance, according to \cite[p. 274]{DunRob81}, is
\begin{align*}
	\frac{1}{\uptau(l)},\quad \text{for } l=0,\ldots,T-l,
\end{align*}
where $ \uptau(l)=\tfrac{1}{T}\sum_{t=1}^{T-l}\e[O_tO_{t+l}] =\gamma_{O}(l)+\uptau^2$. 
If using $\hat{\rho}_{DR}(l)$ to test the null hypothesis of serial independence at lag~$l$ on level~$\alpha$, then the critical value is  $ \pm z_{1-\frac{\alpha}{2}}/\sqrt{\uptau(l)} $, and is therefore dependent on the lag~$l$.
%
\subsection{Tables}\label{A_Tables}
%
In this section, we present the full tables from our simulation study with 10,000 replications per scenario. For unbounded counts, the Poi-INAR(1) model is assumed, which has the Poisson marginal distribution \poi($\mu$) with $\mu\in(0,\infty)$. For bounded counts, we assume the BAR(1) model, which has the binomial marginal distribution \bin($n,\pi$) with $\pi\in(0,1)$ and $ n\in\bbn $. 
\begin{table}[H]
	\centering
	\caption{Asymptotic vs. simulated mean and standard deviation (SD) of $\hat{I}^\poi$ and $\hat{I}^\poi_{\skw}$ data; time series of length $ T $ is generated by 
		Poi-INAR(1)  counts with fixed $\mu=3$, $ \rho=0.5 $.}
	\label{A_tabPoi}
	\smallskip
	\scalebox{0.56}{
		\begin{tabular}{llr|rlrl|rlrl}
			\toprule
			&&& \multicolumn{2}{c}{mean of $\hat{I}^\poi$} & \multicolumn{2}{c|}{SD of $\hat{I}^\poi$}& \multicolumn{2}{c}{mean of $\hat{I}^\poi_{\skw}$} & \multicolumn{2}{c}{SD of $\hat{I}^\poi_{\skw}$} \\
			$\uptau$ & $r$ & $T$  & $ \quad $sim$ \quad $ & asym$ \quad $ & $ \quad $sim$ \quad $ & asym$ \quad $  & $ \quad $sim$ \quad $ & asym$ \quad $ & $ \quad $sim$ \quad $ & asym$ \quad $ \\
			\midrule
			1    & 0    & 100 & 0.971 & 0.970 & 0.177 & 0.183 & 0.974 & 0.973 & 0.124 & 0.133 \\ 
			0.8 & 0    & 100 & 0.967 & 0.968 & 0.189 & 0.196 & 0.971 & 0.970 & 0.132 & 0.143 \\ 
			0.6 & 0    & 100 & 0.965 & 0.963 & 0.210 & 0.216 & 0.966 & 0.965 & 0.146 & 0.158 \\ 
			0.4 & 0    & 100 & 0.955 & 0.955 & 0.238 & 0.252 & 0.956 & 0.956 & 0.164 & 0.185 \\ \midrule
			1    & 0.3 & 100 & 0.972 & 0.970 & 0.177 & 0.183 & 0.974 & 0.973 & 0.123 & 0.133 \\ 
			0.8 & 0.3 & 100 & 0.967 & 0.967 & 0.190 & 0.198 & 0.971 & 0.969 & 0.133 & 0.144 \\ 
			0.6 & 0.3 & 100 & 0.962 & 0.961 & 0.213 & 0.221 & 0.965 & 0.963 & 0.148 & 0.162 \\ 
			0.4 & 0.3 & 100 & 0.951 & 0.950 & 0.249 & 0.261 & 0.954 & 0.951 & 0.172 & 0.192 \\ \midrule
			1    & 0.6 & 100 & 0.971 & 0.970 & 0.177 & 0.183 & 0.974 & 0.973 & 0.124 & 0.133 \\ 
			0.8 & 0.6 & 100 & 0.968 & 0.965 & 0.191 & 0.200 & 0.971 & 0.968 & 0.132 & 0.146 \\ 
			0.6 & 0.6 & 100 & 0.956 & 0.958 & 0.219 & 0.227 & 0.961 & 0.960 & 0.151 & 0.166 \\ 
			0.4 & 0.6 & 100 & 0.942 & 0.942 & 0.257 & 0.272 & 0.947 & 0.945 & 0.175 & 0.199 \\ 
			\midrule
			1    & 0    & 250 & 0.988 & 0.988 & 0.114 & 0.115 & 0.989 & 0.989 & 0.081 & 0.084 \\ 
			0.8 & 0    & 250 & 0.987 & 0.987 & 0.124 & 0.124 & 0.989 & 0.988 & 0.089 & 0.090 \\ 
			0.6 & 0    & 250 & 0.987 & 0.985 & 0.134 & 0.137 & 0.987 & 0.986 & 0.096 & 0.100 \\ 
			0.4 & 0    & 250 & 0.980 & 0.982 & 0.157 & 0.159 & 0.981 & 0.982 & 0.112 & 0.117 \\ \midrule
			1    & 0.3 & 250 & 0.990 & 0.988 & 0.114 & 0.115 & 0.990 & 0.989 & 0.082 & 0.084 \\ 
			0.8 & 0.3 & 250 & 0.988 & 0.987 & 0.122 & 0.125 & 0.988 & 0.988 & 0.087 & 0.091 \\ 
			0.6 & 0.3 & 250 & 0.984 & 0.984 & 0.139 & 0.140 & 0.985 & 0.985 & 0.098 & 0.102 \\ 
			0.4 & 0.3 & 250 & 0.981 & 0.980 & 0.162 & 0.165 & 0.982 & 0.981 & 0.115 & 0.121 \\ \midrule
			1    & 0.6 & 250 & 0.989 & 0.988 & 0.113 & 0.115 & 0.989 & 0.989 & 0.081 & 0.084 \\ 
			0.8 & 0.6 & 250 & 0.987 & 0.986 & 0.126 & 0.127 & 0.987 & 0.987 & 0.089 & 0.092 \\ 
			0.6 & 0.6 & 250 & 0.984 & 0.983 & 0.140 & 0.143 & 0.985 & 0.984 & 0.100 & 0.105 \\ 
			0.4 & 0.6 & 250 & 0.975 & 0.977 & 0.169 & 0.172 & 0.978 & 0.978 & 0.118 & 0.126 \\ \midrule
			1    & 0    & 500 & 0.995 & 0.994 & 0.081 & 0.082 & 0.995 & 0.995 & 0.059 & 0.059 \\ 
			0.8 & 0    & 500 & 0.994 & 0.994 & 0.088 & 0.088 & 0.995 & 0.994 & 0.063 & 0.064 \\ 
			0.6 & 0    & 500 & 0.992 & 0.993 & 0.096 & 0.097 & 0.992 & 0.993 & 0.069 & 0.071 \\ 
			0.4 & 0    & 500 & 0.988 & 0.991 & 0.109 & 0.113 & 0.989 & 0.991 & 0.080 & 0.083 \\ \midrule
			1    & 0.3 & 500 & 0.995 & 0.994 & 0.081 & 0.082 & 0.995 & 0.995 & 0.059 & 0.059 \\ 
			0.8 & 0.3 & 500 & 0.993 & 0.993 & 0.088 & 0.088 & 0.993 & 0.994 & 0.064 & 0.065 \\ 
			0.6 & 0.3 & 500 & 0.991 & 0.992 & 0.097 & 0.099 & 0.991 & 0.993 & 0.071 & 0.072 \\ 
			0.4 & 0.3 & 500 & 0.988 & 0.990 & 0.116 & 0.117 & 0.989 & 0.990 & 0.084 & 0.086 \\ \midrule
			1    & 0.6 & 500 & 0.995 & 0.994 & 0.081 & 0.082 & 0.994 & 0.995 & 0.058 & 0.059 \\ 
			0.8 & 0.6 & 500 & 0.994 & 0.993 & 0.088 & 0.090 & 0.993 & 0.994 & 0.063 & 0.065 \\ 
			0.6 & 0.6 & 500 & 0.991 & 0.992 & 0.100 & 0.101 & 0.991 & 0.992 & 0.072 & 0.074 \\ 
			0.4 & 0.6 & 500 & 0.987 & 0.988 & 0.121 & 0.122 & 0.988 & 0.989 & 0.087 & 0.089 \\ \midrule
			1    & 0    & 1000 & 0.997 & 0.997 & 0.058 & 0.058 & 0.997 & 0.997 & 0.042 & 0.042 \\ 
			0.8 & 0    & 1000 & 0.997 & 0.997 & 0.061 & 0.062 & 0.997 & 0.997 & 0.045 & 0.045 \\ 
			0.6 & 0    & 1000 & 0.996 & 0.996 & 0.068 & 0.068 & 0.996 & 0.997 & 0.050 & 0.050 \\ 
			0.4 & 0    & 1000 & 0.995 & 0.996 & 0.080 & 0.080 & 0.996 & 0.996 & 0.059 & 0.059 \\ \midrule
			1    & 0.3 & 1000 & 0.997 & 0.997 & 0.058 & 0.058 & 0.998 & 0.997 & 0.042 & 0.042 \\ 
			0.8 & 0.3 & 1000 & 0.998 & 0.997 & 0.063 & 0.063 & 0.998 & 0.997 & 0.045 & 0.046 \\ 
			0.6 & 0.3 & 1000 & 0.996 & 0.996 & 0.070 & 0.070 & 0.997 & 0.996 & 0.051 & 0.051 \\ 
			0.4 & 0.3 & 1000 & 0.993 & 0.995 & 0.083 & 0.083 & 0.994 & 0.995 & 0.060 & 0.061 \\ \midrule
			1    & 0.6 & 1000 & 0.998 & 0.997 & 0.058 & 0.058 & 0.998 & 0.997 & 0.042 & 0.042 \\ 
			0.8 & 0.6 & 1000 & 0.997 & 0.997 & 0.064 & 0.063 & 0.997 & 0.997 & 0.046 & 0.046 \\ 
			0.6 & 0.6 & 1000 & 0.997 & 0.996 & 0.071 & 0.072 & 0.996 & 0.996 & 0.052 & 0.052 \\ 
			0.4 & 0.6 & 1000 & 0.994 & 0.994 & 0.086 & 0.086 & 0.994 & 0.995 & 0.062 & 0.063 \\ 
			\bottomrule
	\end{tabular}}
\end{table}

\begin{table}[H]
	\centering
	\caption{Asymptotic vs. simulated mean and SD of $\hat{I}^\bin$ and $\hat{I}^\bin_{\skw}$ data; time series of length $ T $ is generated by 
		BAR(1)  counts with  fixed $\mu=3$, $ \rho=0.5 $ and $ n=10 $.}
	\label{tabBin10}
	\smallskip
	\scalebox{0.56}{
		\begin{tabular}{llr|rlrl|rlrl}
			\toprule
			&&& \multicolumn{2}{c}{mean of $\hat{I}^\bin$} & \multicolumn{2}{c|}{SD of $\hat{I}^\bin$}& \multicolumn{2}{c}{mean of $\hat{I}^\bin_{\skw}$} & \multicolumn{2}{c}{SD of $\hat{I}^\bin_{\skw}$} \\
			$\uptau$ & $r$ & $T$  & $ \quad $sim$ \quad $ & asym$ \quad $ & $ \quad $sim$ \quad $ & asym$ \quad $  & $ \quad $sim$ \quad $ & asym$ \quad $ & $ \quad $sim$ \quad $ & asym$ \quad $ \\
			\midrule
			1    & 0    & 100 & 0.972 & 0.973 & 0.169 & 0.173 & 0.785 & 0.786 & 0.075 & 0.078 \\ 
			0.8 & 0    & 100 & 0.971 & 0.971 & 0.181 & 0.186 & 0.784 & 0.784 & 0.081 & 0.084 \\ 
			0.6 & 0    & 100 & 0.964 & 0.967 & 0.201 & 0.205 & 0.781 & 0.782 & 0.091 & 0.092 \\ 
			0.4 & 0    & 100 & 0.960 & 0.960 & 0.230 & 0.239 & 0.777 & 0.777 & 0.103 & 0.108 \\ \midrule
			1    & 0.3 & 100 & 0.973 & 0.973 & 0.170 & 0.173 & 0.786 & 0.786 & 0.076 & 0.078 \\ 
			0.8 & 0.3 & 100 & 0.973 & 0.970 & 0.180 & 0.188 & 0.786 & 0.784 & 0.081 & 0.084 \\ 
			0.6 & 0.3 & 100 & 0.963 & 0.965 & 0.203 & 0.210 & 0.780 & 0.781 & 0.091 & 0.094 \\ 
			0.4 & 0.3 & 100 & 0.961 & 0.955 & 0.241 & 0.248 & 0.777 & 0.775 & 0.109 & 0.112 \\ \midrule
			1    & 0.6 & 100 & 0.972 & 0.973 & 0.169 & 0.173 & 0.785 & 0.786 & 0.075 & 0.078 \\ 
			0.8 & 0.6 & 100 & 0.970 & 0.969 & 0.186 & 0.190 & 0.783 & 0.783 & 0.084 & 0.085 \\ 
			0.6 & 0.6 & 100 & 0.961 & 0.962 & 0.210 & 0.215 & 0.779 & 0.779 & 0.094 & 0.097 \\ 
			0.4 & 0.6 & 100 & 0.946 & 0.948 & 0.254 & 0.258 & 0.771 & 0.771 & 0.115 & 0.116 \\ \midrule
			1    & 0    & 250 & 0.990 & 0.989 & 0.108 & 0.110 & 0.795 & 0.794 & 0.049 & 0.049 \\ 
			0.8 & 0    & 250 & 0.987 & 0.988 & 0.118 & 0.117 & 0.793 & 0.794 & 0.053 & 0.053 \\ 
			0.6 & 0    & 250 & 0.986 & 0.987 & 0.128 & 0.130 & 0.792 & 0.793 & 0.057 & 0.058 \\ 
			0.4 & 0    & 250 & 0.986 & 0.984 & 0.148 & 0.151 & 0.792 & 0.791 & 0.067 & 0.068 \\ \midrule
			1    & 0.3 & 250 & 0.988 & 0.989 & 0.107 & 0.110 & 0.794 & 0.794 & 0.048 & 0.049 \\ 
			0.8 & 0.3 & 250 & 0.984 & 0.988 & 0.117 & 0.119 & 0.792 & 0.793 & 0.053 & 0.053 \\ 
			0.6 & 0.3 & 250 & 0.985 & 0.986 & 0.131 & 0.133 & 0.792 & 0.792 & 0.058 & 0.060 \\ 
			0.4 & 0.3 & 250 & 0.982 & 0.982 & 0.153 & 0.157 & 0.790 & 0.790 & 0.069 & 0.071 \\ \midrule
			1    & 0.6 & 250 & 0.988 & 0.989 & 0.109 & 0.110 & 0.794 & 0.794 & 0.049 & 0.049 \\ 
			0.8 & 0.6 & 250 & 0.986 & 0.988 & 0.119 & 0.120 & 0.793 & 0.793 & 0.054 & 0.054 \\ 
			0.6 & 0.6 & 250 & 0.986 & 0.985 & 0.135 & 0.136 & 0.792 & 0.792 & 0.061 & 0.061 \\ 
			0.4 & 0.6 & 250 & 0.979 & 0.979 & 0.162 & 0.163 & 0.789 & 0.788 & 0.072 & 0.073 \\ \midrule
			1    & 0    & 500 & 0.995 & 0.995 & 0.077 & 0.077 & 0.797 & 0.797 & 0.034 & 0.035 \\ 
			0.8 & 0    & 500 & 0.994 & 0.994 & 0.082 & 0.083 & 0.797 & 0.797 & 0.037 & 0.037 \\ 
			0.6 & 0    & 500 & 0.992 & 0.993 & 0.090 & 0.092 & 0.796 & 0.796 & 0.040 & 0.041 \\ 
			0.4 & 0    & 500 & 0.993 & 0.992 & 0.106 & 0.107 & 0.796 & 0.795 & 0.048 & 0.048 \\ \midrule
			1    & 0.3 & 500 & 0.995 & 0.995 & 0.076 & 0.077 & 0.797 & 0.797 & 0.034 & 0.035 \\ 
			0.8 & 0.3 & 500 & 0.993 & 0.994 & 0.083 & 0.084 & 0.797 & 0.797 & 0.038 & 0.038 \\ 
			0.6 & 0.3 & 500 & 0.992 & 0.993 & 0.093 & 0.094 & 0.796 & 0.796 & 0.042 & 0.042 \\ 
			0.4 & 0.3 & 500 & 0.990 & 0.991 & 0.108 & 0.111 & 0.795 & 0.795 & 0.049 & 0.050 \\ \midrule
			1    & 0.6 & 500 & 0.994 & 0.995 & 0.077 & 0.077 & 0.797 & 0.797 & 0.034 & 0.035 \\ 
			0.8 & 0.6 & 500 & 0.993 & 0.994 & 0.085 & 0.085 & 0.796 & 0.797 & 0.038 & 0.038 \\ 
			0.6 & 0.6 & 500 & 0.991 & 0.992 & 0.096 & 0.096 & 0.795 & 0.796 & 0.043 & 0.043 \\ 
			0.4 & 0.6 & 500 & 0.990 & 0.990 & 0.115 & 0.115 & 0.795 & 0.794 & 0.051 & 0.052 \\ \midrule
			1    & 0    & 1000 & 0.997 & 0.997 & 0.055 & 0.055 & 0.799 & 0.799 & 0.025 & 0.025 \\ 
			0.8 & 0    & 1000 & 0.997 & 0.997 & 0.059 & 0.059 & 0.799 & 0.798 & 0.026 & 0.026 \\ 
			0.6 & 0    & 1000 & 0.997 & 0.997 & 0.065 & 0.065 & 0.798 & 0.798 & 0.029 & 0.029 \\ 
			0.4 & 0    & 1000 & 0.995 & 0.996 & 0.075 & 0.075 & 0.797 & 0.798 & 0.034 & 0.034 \\ \midrule
			1    & 0.3 & 1000 & 0.997 & 0.997 & 0.054 & 0.055 & 0.798 & 0.799 & 0.024 & 0.025 \\ 
			0.8 & 0.3 & 1000 & 0.996 & 0.997 & 0.059 & 0.059 & 0.798 & 0.798 & 0.026 & 0.027 \\ 
			0.6 & 0.3 & 1000 & 0.996 & 0.996 & 0.066 & 0.066 & 0.798 & 0.798 & 0.030 & 0.030 \\ 
			0.4 & 0.3 & 1000 & 0.996 & 0.995 & 0.079 & 0.078 & 0.798 & 0.797 & 0.036 & 0.035 \\ \midrule
			1    & 0.6 & 1000 & 0.997 & 0.997 & 0.054 & 0.055 & 0.798 & 0.799 & 0.024 & 0.025 \\ 
			0.8 & 0.6 & 1000 & 0.997 & 0.997 & 0.060 & 0.060 & 0.798 & 0.798 & 0.027 & 0.027 \\ 
			0.6 & 0.6 & 1000 & 0.997 & 0.996 & 0.068 & 0.068 & 0.798 & 0.798 & 0.031 & 0.031 \\ 
			0.4 & 0.6 & 1000 & 0.995 & 0.995 & 0.081 & 0.082 & 0.797 & 0.797 & 0.037 & 0.037 \\ 
			\bottomrule
	\end{tabular}}
\end{table}

\begin{table}[H]
	\centering
	\caption{Asymptotic vs. simulated mean and SD of $\hat{I}^\bin$ and $\hat{I}^\bin_{\skw}$ data; time series of length $ T $ is generated by 
		BAR(1)  counts with  fixed $\mu=3$, $ \rho=0.5 $ and $ n=25 $.}
	\label{tabBin25}
	\smallskip
	\scalebox{0.56}{
		\begin{tabular}{llr|rlrl|rlrl}
			\toprule
			&&& \multicolumn{2}{c}{mean of $\hat{I}^\bin$} & \multicolumn{2}{c|}{SD of $\hat{I}^\bin$}& \multicolumn{2}{c}{mean of $\hat{I}^\bin_{\skw}$} & \multicolumn{2}{c}{SD of $\hat{I}^\bin_{\skw}$} \\
			$\uptau$ & $r$ & $T$  & $ \quad $sim$ \quad $ & asym$ \quad $ & $ \quad $sim$ \quad $ & asym$ \quad $  & $ \quad $sim$ \quad $ & asym$ \quad $ & $ \quad $sim$ \quad $ & asym$ \quad $ \\
			\midrule
			1    & 0    & 100 & 0.975 & 0.971 & 0.173 & 0.179 & 0.901 & 0.898 & 0.103 & 0.109 \\ 
			0.8 & 0    & 100 & 0.968 & 0.969 & 0.185 & 0.192 & 0.896 & 0.896 & 0.110 & 0.118 \\ 
			0.6 & 0    & 100 & 0.968 & 0.965 & 0.206 & 0.212 & 0.895 & 0.893 & 0.123 & 0.130 \\ 
			0.4 & 0    & 100 & 0.955 & 0.957 & 0.238 & 0.247 & 0.885 & 0.885 & 0.142 & 0.153 \\ \midrule
			1    & 0.3 & 100 & 0.975 & 0.971 & 0.174 & 0.179 & 0.902 & 0.898 & 0.104 & 0.109 \\ 
			0.8 & 0.3 & 100 & 0.971 & 0.968 & 0.190 & 0.194 & 0.897 & 0.896 & 0.113 & 0.119 \\ 
			0.6 & 0.3 & 100 & 0.963 & 0.963 & 0.209 & 0.217 & 0.891 & 0.891 & 0.124 & 0.133 \\ 
			0.4 & 0.3 & 100 & 0.952 & 0.952 & 0.252 & 0.256 & 0.883 & 0.882 & 0.147 & 0.158 \\ \midrule
			1    & 0.6 & 100 & 0.972 & 0.971 & 0.173 & 0.179 & 0.900 & 0.898 & 0.103 & 0.109 \\ 
			0.8 & 0.6 & 100 & 0.968 & 0.967 & 0.190 & 0.196 & 0.896 & 0.895 & 0.112 & 0.120 \\ 
			0.6 & 0.6 & 100 & 0.965 & 0.959 & 0.214 & 0.222 & 0.892 & 0.889 & 0.127 & 0.136 \\ 
			0.4 & 0.6 & 100 & 0.943 & 0.944 & 0.259 & 0.266 & 0.877 & 0.877 & 0.152 & 0.164 \\ \midrule
			1    & 0    & 250 & 0.988 & 0.988 & 0.112 & 0.113 & 0.911 & 0.911 & 0.068 & 0.069 \\ 
			0.8 & 0    & 250 & 0.987 & 0.988 & 0.119 & 0.121 & 0.910 & 0.910 & 0.072 & 0.074 \\ 
			0.6 & 0    & 250 & 0.986 & 0.986 & 0.132 & 0.134 & 0.909 & 0.909 & 0.080 & 0.082 \\ 
			0.4 & 0    & 250 & 0.982 & 0.983 & 0.153 & 0.156 & 0.905 & 0.906 & 0.093 & 0.096 \\ \midrule
			1    & 0.3 & 250 & 0.989 & 0.988 & 0.112 & 0.113 & 0.912 & 0.911 & 0.068 & 0.069 \\ 
			0.8 & 0.3 & 250 & 0.986 & 0.987 & 0.121 & 0.123 & 0.909 & 0.910 & 0.072 & 0.075 \\ 
			0.6 & 0.3 & 250 & 0.984 & 0.985 & 0.135 & 0.137 & 0.908 & 0.908 & 0.082 & 0.084 \\ 
			0.4 & 0.3 & 250 & 0.982 & 0.981 & 0.160 & 0.162 & 0.905 & 0.905 & 0.096 & 0.100 \\ \midrule
			1    & 0.6 & 250 & 0.988 & 0.988 & 0.110 & 0.113 & 0.911 & 0.911 & 0.067 & 0.069 \\ 
			0.8 & 0.6 & 250 & 0.987 & 0.987 & 0.123 & 0.124 & 0.909 & 0.910 & 0.074 & 0.076 \\ 
			0.6 & 0.6 & 250 & 0.982 & 0.984 & 0.137 & 0.140 & 0.906 & 0.907 & 0.083 & 0.086 \\ 
			0.4 & 0.6 & 250 & 0.979 & 0.978 & 0.167 & 0.168 & 0.904 & 0.903 & 0.100 & 0.104 \\ \midrule
			1    & 0    & 500 & 0.995 & 0.994 & 0.079 & 0.080 & 0.916 & 0.916 & 0.048 & 0.049 \\ 
			0.8 & 0    & 500 & 0.995 & 0.994 & 0.085 & 0.086 & 0.916 & 0.915 & 0.052 & 0.053 \\ 
			0.6 & 0    & 500 & 0.991 & 0.993 & 0.093 & 0.095 & 0.914 & 0.915 & 0.057 & 0.058 \\ 
			0.4 & 0    & 500 & 0.991 & 0.991 & 0.109 & 0.110 & 0.913 & 0.913 & 0.067 & 0.068 \\ \midrule
			1    & 0.3 & 500 & 0.994 & 0.994 & 0.080 & 0.080 & 0.916 & 0.916 & 0.048 & 0.049 \\ 
			0.8 & 0.3 & 500 & 0.991 & 0.994 & 0.087 & 0.087 & 0.914 & 0.915 & 0.053 & 0.053 \\ 
			0.6 & 0.3 & 500 & 0.993 & 0.993 & 0.096 & 0.097 & 0.914 & 0.914 & 0.059 & 0.060 \\ 
			0.4 & 0.3 & 500 & 0.990 & 0.990 & 0.114 & 0.114 & 0.912 & 0.912 & 0.069 & 0.071 \\ \midrule
			1    & 0.6 & 500 & 0.993 & 0.994 & 0.080 & 0.080 & 0.915 & 0.916 & 0.049 & 0.049 \\ 
			0.8 & 0.6 & 500 & 0.995 & 0.993 & 0.089 & 0.088 & 0.916 & 0.915 & 0.054 & 0.054 \\ 
			0.6 & 0.6 & 500 & 0.992 & 0.992 & 0.099 & 0.099 & 0.913 & 0.914 & 0.060 & 0.061 \\ 
			0.4 & 0.6 & 500 & 0.988 & 0.989 & 0.117 & 0.119 & 0.911 & 0.911 & 0.072 & 0.073 \\ \midrule
			1    & 0    & 1000 & 0.998 & 0.997 & 0.057 & 0.057 & 0.918 & 0.918 & 0.035 & 0.035 \\ 
			0.8 & 0    & 1000 & 0.997 & 0.997 & 0.061 & 0.061 & 0.917 & 0.918 & 0.037 & 0.037 \\ 
			0.6 & 0    & 1000 & 0.996 & 0.996 & 0.067 & 0.067 & 0.916 & 0.917 & 0.041 & 0.041 \\ 
			0.4 & 0    & 1000 & 0.995 & 0.996 & 0.078 & 0.078 & 0.916 & 0.917 & 0.048 & 0.048 \\ \midrule
			1    & 0.3 & 1000 & 0.997 & 0.997 & 0.056 & 0.057 & 0.917 & 0.918 & 0.034 & 0.035 \\ 
			0.8 & 0.3 & 1000 & 0.997 & 0.997 & 0.061 & 0.061 & 0.917 & 0.918 & 0.037 & 0.038 \\ 
			0.6 & 0.3 & 1000 & 0.997 & 0.996 & 0.068 & 0.068 & 0.918 & 0.917 & 0.042 & 0.042 \\ 
			0.4 & 0.3 & 1000 & 0.995 & 0.995 & 0.080 & 0.081 & 0.916 & 0.916 & 0.049 & 0.050 \\ \midrule
			1    & 0.6 & 1000 & 0.997 & 0.997 & 0.056 & 0.057 & 0.918 & 0.918 & 0.034 & 0.035 \\ 
			0.8 & 0.6 & 1000 & 0.998 & 0.997 & 0.062 & 0.062 & 0.918 & 0.917 & 0.038 & 0.038 \\ 
			0.6 & 0.6 & 1000 & 0.996 & 0.996 & 0.071 & 0.070 & 0.917 & 0.917 & 0.043 & 0.043 \\ 
			0.4 & 0.6 & 1000 & 0.995 & 0.994 & 0.084 & 0.084 & 0.916 & 0.916 & 0.051 & 0.052 \\ 
			\bottomrule
	\end{tabular}}
\end{table}
%
\subsection{Raw moments}\label{A_RawMom}
%
In this section, we present alternative results for the CLT and the Poisson index of dispersion if we are concerned with raw moments. Therefore, let us consider the vector $ (\Y_t )$ which is given by
\begin{equation}
	\Y_t\coloneqq(Y_{t,1},\ldots,Y_{t,m})^\top=(X_t,X_t^2,\ldots,X_t^m)^\top.
\end{equation}
Furthermore, we define the amplitude modulation of $ (\Y_t) $ as $ (O_t\cdot\Y_t) $. Let us define $ \overline{O\,\Y}\coloneqq \tfrac{1}{n}\sum_{t=1}^{n}O_t\Y_t $.  Then, the mean of $\overline{O\,\Y}  $ is given by 
\begin{align}\label{Johnson}
	\e[\overline{O\,\Y}] = \tfrac{1}{n}\sum_{t=1}^{n}\e[O_t]\e[\Y_t]=\Big(\tfrac{1}{n}\sum_{t=1}^{n}\e[O_t]\Big)\bnu,
\end{align}
with $\bnu\coloneqq(\mu,\ldots,\mu_m)^\top=(\e[X_t],\ldots,\e[X_t^m])^\top  $. This implies to estimate $ \bnu $ by
\begin{align}\label{Ewing}
	\hat{\bnu}\coloneqq \frac{\tfrac{1}{n}\sum_{t=1}^{n}O_t\Y_t}{\tfrac{1}{n}\sum_{t=1}^{n}O_t}\eqqcolon \frac{\overline{O\,\Y}}{\overline{O}}.
\end{align}
\begin{definition}
	Let us denote the mixed raw moment as 
	\begin{align}\label{West}
		\mu_{k,l}(h)&\coloneqq\e[X_t^k\cdot X_{t-h}^l],
	\end{align} 
	respectively with $ 0\leq k\leq l $, $ k,l\in\bbn $ and $ h\in \bbz $.  In addition, we declare  

	\vspace{0.3cm} 
	\begin{tabular}{llll}
		$(i)$ &$ \mu_{0,0}(h) = 1, $ &$(iii)$ &$  \mu_{0,l}(h) = \mu_l,  $\\
		$(ii)$& $ \mu_{k,0}(h) = \mu_k, $&$(iv)$ &$  \mu_{k,l}(h) = \mu_{l,k}(-h)$.
	\end{tabular} 
	\vspace{0.3cm} 
	
	Note that for $ h=0 $, one can simplify $ \mu_{k,l}(0)=\mu_{k+l} $.
\end{definition}
%
\subsubsection{Central Limit Theorem}\label{A_RawCLT}
%
The following proposition can be seen as the analogon to Theorem \ref{CLT_Fac}.
\begin{proposition}\label{CLT_Raw}
	Let $ \overline{\Y^\ast}=(\overline{O},\overline{O\,\Y}^\top)^\top $, and define the function $ \f:[0,1]\times[0,\infty)^m\rightarrow[0,\infty)^m $. Then, $ \hat{\bnu}=\f(\overline{\Y^\ast}) $, $ \bnu=\f(\bnu^\ast) $. Thus, the Delta-Method implies 
	\begin{equation}\label{Iniesta}
		\sqrt{T}\Big(\hat{\bnu}-\bnu\Big) \ \xrightarrow{\text{d}}\ \norm\Big(\mathbf{0}, \bold{\Sigma}\Big)\quad \text{with}\quad \bold{\Sigma}=(\sigma_{ij})_{i,j=1,\ldots,m},
	\end{equation}
	where
	\begin{equation}
		\label{Maradona}
		\sigma_{ij}=\tfrac{1}{\uptau}(\mu_{i+j}-\mu_i\mu_j)+\tfrac{1}{\uptau^2}\sum_{h=1}^\infty \uptau(h)\Big(\mu_{j,i}(h)+\mu_{i,j}(h)-2\mu_i\mu_j\Big).
	\end{equation}
	Here, the bias satisfies $ \e[\hat{\bnu}-\bnu]=0+\landau (T^{-1}) $.
\end{proposition}
\begin{proof}
	For the raw moments, the CLT is derived analogously to the case of the factorial moments. Hence, let us define the vector $ (\Y_t^\ast )$ as
	\begin{equation}
		\Y_t^\ast\coloneqq(Y_{t,0}^\ast,\ldots,Y_{t,m}^\ast)^\top=O_t(1,X_t,X_t^2;\ldots,X_t^m)^\top.
	\end{equation}
	For the mean of  $ (\Y_t^\ast) $, we obtain 
	\begin{equation}
		\bnu^\ast\coloneqq \e[\Y_t^\ast]= \uptau(1,\bnu^\top)^\top.
	\end{equation}
	Now, we assume that appropriate mixing assumptions on $ (X_t) $ and $ (O_t) $ hold according to \eqref{A}, \ie {\boldmath $\alpha$}-mixing with exponentially decreasing weights, which then hand over to  $ (\Y_t^\ast) $. Then, we have the CLT
	\begin{equation}\label{Iverson}
		\sqrt{T}\Big(\overline{\Y^\ast}-\bnu^\ast\Big) \ \xrightarrow{\text{d}}\ \norm\Big(\mathbf{0}, \bold{\Sigma}^\ast\Big)\quad \text{with}\quad \mathbf{\Sigma^\ast}=(\sigma_{ij}^\ast)_{i,j=0,\ldots,m},
	\end{equation}
	where
	\begin{equation}\label{Barca} 
		\sigma_{ij}^\ast=CoV[Y_{0,i}^\ast,Y_{0,j}^\ast]+\sum_{h=1}^\infty \Big(CoV[Y_{0,i}^\ast,Y_{h,j}^\ast]+CoV[Y_{h,i}^\ast,Y_{0,j}^\ast]\Big).
	\end{equation}
	These covariances compute as 
	\begin{equation}\label{Lakers}
		\sigma_{ij}^\ast=
		\begin{cases}
			\uptau(1-\uptau)+2\sum_{h=1}^\infty \gamma_O(h) & \text{if }i=j=0,\\
			\sigma_{00}^\ast\mu_j         &  \text{if }i=0,\ j>0,\\
			\uptau(\mu_{i+j}-\mu_i\mu_j)+\sigma_{00}^\ast\mu_i\mu_j& \text{if }i,j>0.\\
			\quad+\sum_{h=1}^\infty \uptau(h)\Big(\mu_{j,i}(h)+\mu_{i,j}(h)-2\mu_i\mu_j\Big) \\
		\end{cases}
	\end{equation}
	Here, $ \overline{\Y^\ast}=(\overline{O},\overline{O\,\Y}^\top)^\top $ are the required components for the calculation of $ \hat{\bnu} $. The derivations of the covariances follow the same steps as for the factorial moments. Thus, we use the same function $ \f $ as defined in Section \ref{A_CLT}. Then, $ \hat{\bnu}=\f(\overline{\Y^\ast}) $, $ \bnu=\f(\bnu^\ast) $, and the 
	Jacobian of $ \f $ evaluated in $ \bnu^\ast $ equals
	\begin{align*}
		\D= \frac{1}{\uptau}
		\begin{pmatrix} 
			-\mu & 1 & 0 & \cdots & 0\\
			-\mu_2 & 0 & 1 & \ddots  & \vdots\\
			\vdots & \vdots & \ddots & \ddots & 0\\
			-\mu_m & 0 & \cdots & 0  & 1\\
		\end{pmatrix}.
	\end{align*}
	So the linear Taylor approximation $ \hat{\bnu}\approx \bnu + \D(\overline{\Y^\ast}-\bnu^\ast)$ together with the CLT implies 
	\begin{align*}
		\sqrt{T}\Big(\hat{\bnu}-\bnu\Big) \ \xrightarrow{\text{d}}\ \norm\Big(\bold{0}, \bold{\Sigma}\Big)\quad \text{with}\quad \bold{\Sigma}=\D\bold{\Sigma}^\ast\D^\top.
	\end{align*}
	We can compute the entries $ (\sigma_{ij})_{i,j=1,\ldots,m} $ as 
	\begin{align*}
		\sigma_{ij}&=\sum_{k,l=0}^{m}d_{ik}d_{jl}\sigma_{kl}^\ast
		= d_{i0}d_{j0}\sigma_{00}^\ast + d_{i0}d_{jj}\sigma_{0j}^\ast + d_{ii}d_{j0}\sigma_{i0}^\ast + d_{ii}d_{jj}\sigma_{ij}^\ast\\
		&=\tfrac{1}{\uptau^2}\Big( \mu_i\mu_j\sigma_{00}^\ast - \mu_i\sigma_{0j}^\ast - \mu_j\sigma_{i0}^\ast + \sigma_{ij}^\ast\Big)=\tfrac{1}{\uptau^2}\Big(\sigma_{ij}^\ast- \mu_i\mu_j\sigma_{00}^\ast \Big)\\
		&=\tfrac{1}{\uptau}(\mu_{i+j}-\mu_i\mu_j)+\tfrac{1}{\uptau^2}\sum_{h=1}^\infty \uptau(h)\Big(\mu_{j,i}(h)+\mu_{i,j}(h)-2\mu_i\mu_j\Big),
	\end{align*}
	where in the last step, we utilized that $ \sigma_{0j}^\ast=\sigma_{00}^\ast\cdot \mu_j $. Analogously to the factorial moments, the bias for the raw moments is also negligible for practice. This is due to the quadratic Taylor approximation and follows the same steps as for the factorial moments.
\end{proof}
%
\subsubsection{Poisson Index of Dispersion}\label{A_RawPIndex}
%
Similarly to Section \ref{Ch_PInd}, we start from \eqref{Iniesta} and define a function $ g $ by
\begin{align*}
	g(x_1,x_2)=\frac{x_2}{x_1}-x_1.
\end{align*}
Since the difference between $ g $ and the function chosen in Section \ref{Ch_PInd} 
is only a constant 1, the partial derivatives coincide. Therefore, we get  the Jacobian $ \D $ and the Hessian $ \Hes $ by evaluating the partial derivatives in $ \bnu $, which leads to
\begin{align*}
	\D=\Big(-\frac{\mu_2}{\mu^2}-1,\frac{1}{\mu}\Big), \qquad \Hes=\frac{1}{\mu^3}
	\begin{pmatrix} 
		2\mu_2 & -\mu \\
		-\mu & 0 \\
	\end{pmatrix}.
\end{align*}

\vspace{0.5cm} 
Hence, the asymptotic variance and bias of $ \hat{I}^\poi $ is
\begin{align}
	\sigma_{\hat{I}^\poi}^2&=\tfrac{1}{T} \Big( d_1^2\sigma_{11}  +d_2^2\sigma_{22}+2 d_1d_2\sigma_{12}  \Big),\label{Davis}\\
	\mathbb{B}_{\hat{I}^\poi}&=\tfrac{1}{T} \Big( \tfrac{1}{2}h_{11}\sigma_{11}  +h_{12}\sigma_{12}  \Big).
\end{align}
\begin{proposition}
	The asymptotic variance of $\hat{I}^\poi$  for any Poisson model is given by 
	\begin{align*}
		\sigma_{\hat{I}^\poi}^2&=\tfrac{1}{T\uptau} \Bigg[ \big(\tfrac{\mu_2}{\mu^2}+1\big)^2(\mu_2-\mu^2)-\tfrac{2}{\mu}\big(\tfrac{\mu_2}{\mu^2}+1\big)(\mu_3-\mu\mu_2)+\tfrac{1}{\mu^2}(\mu_4-\mu_2^2)
		\\[1ex]&\quad +\tfrac{2}{\uptau}\sum_{h=1}^\infty\uptau(h)\Big( \big(\tfrac{\mu_2}{\mu^2}+1\big)^2\big(\mu_{1,1}(h)-\mu^2\big)+\tfrac{1}{\mu^2}\big(\mu_{2,2}(h)-\mu_2^2\big)
		\\[1ex]&\quad -\tfrac{1}{\mu}\big(\tfrac{\mu_2}{\mu^2}+1\big)\big(\mu_{2,1}(h)+\mu_{1,2}(h)-2\mu\mu_2\big) \Big) \Bigg].
	\end{align*}
\end{proposition}
\begin{proof}
	Let us start by applying \eqref{Maradona} to \eqref{Davis}. Then, we get 
	\begin{align*}
		\sigma_{\hat{I}^\poi}^2&=\tfrac{1}{T}\Bigg[\big(-\tfrac{\mu_2}{\mu^2}-1\big)^2\Big(\tfrac{1}{\uptau}(\mu_2-\mu^2)+\tfrac{2}{\uptau^2}\sum_{h=1}^\infty\uptau(h)\big(\mu_{1,1}(h)-\mu^2\big)\Big)
		\\[1ex]&\quad+\tfrac{2}{\mu}\big(-\tfrac{\mu_2}{\mu^2}-1\big)\Big(\tfrac{1}{\uptau}(\mu_3-\mu\mu_2)+\tfrac{1}{\uptau^2}\sum_{h=1}^\infty\uptau(h)\big(\mu_{2,1}(h)+\mu_{1,2}(h)-2\mu\mu_2\big)\Big)
		\\[1ex]&\quad+\tfrac{1}{\mu^2}\Big(\tfrac{1}{\uptau}(\mu_4-\mu_2^2)+\tfrac{2}{\uptau^2}\sum_{h=1}^\infty\uptau(h)\big(\mu_{2,2}(h)-\mu_2^2\big)\Big) \Bigg] \\
		&=\tfrac{1}{T\uptau} \Bigg[\big(\tfrac{\mu_2}{\mu^2}+1\big)^2(\mu_2-\mu^2)-\tfrac{2}{\mu}\big(\tfrac{\mu_2}{\mu^2}+1\big)(\mu_3-\mu\mu_2)+\tfrac{1}{\mu^2}(\mu_4-\mu_2^2)
		\\[1ex]&\quad+\tfrac{2}{\uptau}\sum_{h=1}^\infty\uptau(h)\Big(\big(\tfrac{\mu_2}{\mu^2}+1\big)^2\big(\mu_{1,1}(h)-\mu^2\big)+\tfrac{1}{\mu^2}\big(\mu_{2,2}(h)-\mu_2^2\big)
		\\[1ex]&\quad-\tfrac{1}{\mu}\big(\tfrac{\mu_2}{\mu^2}+1\big)\big(\mu_{2,1}(h)+\mu_{1,2}(h)-2\mu\mu_2\big) \Big) \Bigg].
	\end{align*}
\end{proof}
\begin{proposition}
	The asymptotic bias of $ \hat{I}^\poi $ for any Poisson model is given by 
	\begin{align*}
		\mathbb{B}_{\hat{I}^\poi}=\tfrac{1}{T\uptau\mu^3} \Bigg[\mu_2^2-\mu_3\mu+\tfrac{2}{\uptau}\sum_{h=1}^\infty\uptau(h)\Big(\mu_2\mu_{1,1}(h)-\tfrac{\mu}{2}\big(\mu_{2,1}(h)+\mu_{1,2}(h)\big)\Big)\Bigg].
	\end{align*}
\end{proposition}
\begin{proof}
	The asymptotic bias becomes
	\begin{align*}
		\mathbb{B}_{\hat{I}^\poi}&=\tfrac{1}{T} \Big(\tfrac{1}{2}h_{11}\sigma_{11}+ h_{12}\sigma_{12}\Big)\\
		&=\tfrac{1}{T}\Bigg[\tfrac{\mu_2}{\mu^3}\Big(\tfrac{1}{\uptau}(\mu_2-\mu^2)+\tfrac{2}{\uptau^2}\sum_{h=1}^\infty\uptau(h)\big(\mu_{1,1}(h)-\mu^2\big)\Big)
		\\[1ex]&\quad-\tfrac{1}{\mu^2}\Big(\tfrac{1}{\uptau}(\mu_3-\mu\mu_2)+\tfrac{1}{\uptau^2}\sum_{h=1}^\infty\uptau(h)\big(\mu_{2,1}(h)+\mu_{1,2}(h)-2\mu\mu_2\big)\Big)\Bigg]\\
		&=\tfrac{1}{T\uptau\mu^2}\Bigg[\tfrac{\mu_2}{\mu}(\mu_2-\mu^2)-(\mu_3-\mu\mu_2)
		\\[1ex]&\quad+\tfrac{1}{\uptau}\sum_{h=1}^\infty\uptau(h)\Big(    \tfrac{2\mu_2}{\mu}\big(\mu_{1,1}(h)-\mu^2\big)-\big(\mu_{2,1}(h)+\mu_{1,2}(h)-2\mu\mu_2\big)\Big)\Bigg].
	\end{align*}
	Further simplification leads to
	\begin{align*}
		\mathbb{B}_{\hat{I}^\poi}=\tfrac{1}{T\uptau\mu^3} \Bigg[\mu_2^2-\mu_3\mu+\tfrac{2}{\uptau}\sum_{h=1}^\infty\uptau(h)\Big(\mu_2\cdot\mu_{1,1}(h)-\tfrac{\mu}{2}\big(\mu_{2,1}(h)+\mu_{1,2}(h)\big)\Big)\Bigg].
	\end{align*}
\end{proof}

\end{document}